\documentclass{article}

\usepackage[utf8]{inputenc}
\usepackage{amsmath, amssymb, amsthm,amsfonts,mathrsfs}
\usepackage{ae,aecompl,array,complexity}
\usepackage{color,graphics,hyperref,tikz}
\usepackage{thm-restate}
\usepackage{bbm,verbatim }
\usepackage{stmaryrd}
\usepackage[margin=1.314in]{geometry}

\newcommand{\CC}{\ensuremath{\mathscr{C}}}

\newcommand{\zo}{\{0,1\}}
\newcommand{\F}{\mathbb{F}}
\newcommand{\Fq}{{\mathbb{F}_q}}
\newcommand{\Fn}[2]{\F_{#1}^{#2}}

\newcommand{\Fqn}{\Fn{q}{n}}

\newcommand{\Iint}[2]{\llbracket #1 , #2 \rrbracket}
\newcommand{\N}{\mathbb{N}}
\newcommand{\Z}{\mathbb{Z}}
\newcommand{\Zq}{\Z_q}

\newcommand{\Comp}{\mathbb{C}}

\renewcommand{\D}{\mathcal{D}}
\newcommand{\DQ}{\mathcal{D_Q}}
\newcommand{\DPQ}{\QDP}

\newcommand{\zerov}{\vec{0}}

\DeclareMathOperator{\tr}{tr}
\DeclareMathOperator{\myspan}{span}

\newcommand{\trsp}[1]{{#1}^{\intercal}}
\newcommand{\lmax}{\lambda_{\text{max}}}

\newcommand{\eqdef}{\mathop{=}\limits^{\triangle}}
\newcommand{\Unif}{\leftarrow}

\newcommand{\ie}{\textit{i.e.}}
\newcommand{\eps}{\varepsilon}

\newcommand{\PPGM}{\textrm{P}_{\textrm{PGM}}}
\newcommand{\POPT}{\textrm{P}_{\textrm{OPT}}}
\newcommand{\pusd}{\textrm{p}_{\textrm{usd}}}
\newcommand{\Pusd}{\textrm{P}_{\textrm{usd}}}

\newcommand{\COMMENT}[1]{}
\newcommand{\anote}[1]{}
\newcommand{\jp}[1]{}

\newcommand{\ket}[1]{|#1\rangle}
\newcommand{\bra}[1]{\langle#1|}
\newcommand{\ketbra}[2]{|#1\rangle\langle#2|}
\newcommand{\braket}[2]{\langle #1 | #2 \rangle}
\newcommand{\altketbra}[1]{\ketbra{#1}{#1}}
\newcommand{\kb}[1]{\altketbra{#1}}
\newcommand{\norm}[1]{\ensuremath{\lvert  \kern-1pt\lvert #1 \rvert \kern-1pt \rvert}}

\newcommand{\drgv}{\delta_{\textup{min}}}
\newcommand{\drmax}{\delta_{\textup{max}}}

\renewcommand{\DP}{\mbox{DP}}
\newcommand{\SCP}{\mbox{SCP}}
\newcommand{\LWE}{\mbox{LWE}}
\newcommand{\SLWE}{\mbox{S-LWE}}
\newcommand{\CLWE}{\mbox{C-LWE}}
\newcommand{\SIS}{\mbox{SIS}}

\newcommand{\LPN}{\mbox{LPN}}

\newcommand{\QDP}{\ensuremath{\mathrm{QDP}}}

\newcommand{\BSC}{\mbox{BSC}}

\newcommand{\BSEEC}{\mbox{BSEEC}}

\newcommand{\hh}{\mathcal{H}}

\DeclareMathOperator{\QFTt}{QFT}
\newcommand{\QFT}[1]{\widehat{#1}}
\newcommand{\FT}[1]{\widehat{#1}}

\newcommand{\perpF}[1]{\left( #1 \right)^\perp}
\newcommand{\perpO}[1]{#1^\perp}

\newcommand{\Rbra}[1]{\left( #1 \right)} \newcommand{\Sbra}[1]{\left[ #1 \right]} \newcommand{\Cbra}[1]{\left\{ #1 \right\}}

\newcommand{\mat}[1]{\ensuremath{\boldsymbol{#1}}}

\newcommand{\av}{{\mat{a}}}
\newcommand{\bv}{{\mat{b}}}
\newcommand{\cv}{{\mat{c}}}
\newcommand{\dv}{\mat{d}}
\newcommand{\ev}{\mat{e}}
\newcommand{\mv}{\mat{m}}
\newcommand{\sv}{{\mat{s}}}
\newcommand{\uv}{\mat{u}}

\newcommand{\xv}{{\mat{x}}}
\newcommand{\yv}{{\mat{y}}}
\newcommand{\zv}{{\mat{z}}}

\newcommand{\Am}{\ensuremath{\mathbf{A}}}
\newcommand{\Gm}{\ensuremath{\mathbf{G}}}
\newcommand{\Hm}{\ensuremath{\mathbf{H}}}

\newcommand{\un}{\ensuremath{\mathbbm{1}}}
\newcommand{\zerom}{{\mat{0}}}
\DeclareMathOperator{\rank}{rank}

\newcommand{\tW}{\widetilde{W}}
\newcommand{\wcv}{\widetilde{\cv}}

\newcommand{\OO}[1]{O\Rbra{#1}}

\newcommand{\Th}[1]{\Theta\Rbra{#1}}

\newcommand{\Om}[1]{\Omega\Rbra{#1}}

\newcommand{\Csp}{\C_{\sv}^\perp}
\newcommand{\asCt}{a_{\sv,\C}(t)}

\newcommand{\Ec}{{\mathcal{E}}}
\newcommand{\Hc}{{\mathcal{H}}}

\newtheorem{theorem}{Theorem}

\newtheorem{fact}{Fact}
\newtheorem{definition}{Definition}
\newtheorem{lemma}{Lemma}
\newtheorem{proposition}{Proposition}
\newtheorem{corollary}{Corollary}

\newtheorem{remk}{Remark}

\newtheorem{problem}{Problem}
\newtheorem{notation}{Notation}

\renewcommand{\C}{\mathcal{C}}
\newcommand{\B}{\mathcal{B}}

\newcommand{\cadre}[1]
{
	\begin{tabular}{|p{0.95\textwidth}|}
		\hline
		\vspace*{-0.3cm}
		#1 \\
		\hline
	\end{tabular}
	
}

\newcommand{\return}{\mbox{ return }}
\renewcommand{\and}{\mbox{ and }}
\renewcommand{\wp}{\mbox{ wp. }}

\title{The Quantum Decoding Problem}
\author{Andr\'e Chailloux$^{1}$ and Jean-Pierre Tillich$^{1}$\\
Inria de Paris, {\sf $\{$andre.chailloux,jean-pierre.tillich$\}$@inria.fr}}
\date{}

\begin{document}
\maketitle

\begin{abstract}
	
One of the founding results of lattice based cryptography is a quantum reduction from the Short Integer Solution problem to the Learning with Errors problem introduced by Regev. It has recently been pointed out by Chen, Liu and Zhandry that this reduction can be made more powerful by replacing the learning with errors problem with a quantum equivalent, where the errors are given in quantum superposition.  In the context of codes, this can be adapted to a reduction from finding short codewords to a quantum decoding problem for random linear codes. 

We therefore consider in this paper the quantum decoding problem, where we are given a superposition of noisy versions of a codeword and we want to recover the corresponding codeword.
When we measure the superposition, we get back the usual classical decoding problem for which the best known algorithms are in the constant rate and 
error-rate regime exponential in the codelength. However, we will show here that when the noise rate is small enough, then the quantum decoding problem can be solved 
in quantum polynomial time. Moreover, we also show that the problem can in principle be solved quantumly (albeit not efficiently) for noise rates for which the associated classical decoding problem cannot
be solved at all for information theoretic reasons. 

We then revisit Regev's reduction in the context of codes. We show that using our algorithms for the quantum decoding problem in Regev's reduction matches the best known quantum algorithms for the short codeword problem. This shows in some sense the tightness of Regev's reduction when considering the quantum decoding problem and also paves the way for new quantum algorithms for the short codeword problem.  
\end{abstract}

\newpage
\tableofcontents
\newpage

\section{Introduction}
\subsection{General context}
Error correcting codes which appeared first as the fundamental tool to transmit information reliably through a noisy channel \cite{S48} have found their way outside this kind of applications, such as for instance in average case complexity \cite{L87}, or when locally testable codes were found to be the combinatorial core for probabilistically checkable proofs (PCP) \cite{D07b}. Another important application domain for error correction is cryptography with Shamir's secret sharing scheme \cite{S79}, authentication protocols \cite{S93}, pseudorandom generators \cite{FS96}, 
signature schemes \cite{S93},
or  public-key encryption schemes \cite{M78,A11,MTSB12}.
Contrarily to the applications in reliable communication, data storage, or application in complexity theory where finding suitable families of structured codes is the problem that has to be addressed,  many of these applications in cryptography deal with 
{\em random linear} codes, and more precisely take advantage of the hardness of decoding a generic linear code.

The decoding problem corresponds to decoding the $k$-dimensional vector space $\CC$ ($\ie$, the code) generated by the rows of a randomly generated $\Gm \in \F_q^{k \times n}$ (which is called a {\em generating matrix} of the code) :
\begin{equation}\label{eq:def_code}
\CC \eqdef \Cbra{\uv \Gm \colon \uv \in \F_q^k}.
\end{equation}

Here $\F_q$ denotes the finite field with $q$ elements. In the decoding problem,
we are given the noisy codeword $\cv + \ev$ where $\cv$ belongs to $\CC$ and we are asked to find the original codeword $\cv$.

\begin{problem}[$\DP(q,n,k,f)$] The decoding problem with  positive integer parameters  $q,n,k$ and a probability distribution $f$ on $\Fqn$ is defined as:
	\begin{itemize}\setlength\itemsep{-0.25em}
		\item Input: $(\Gm,\cv + \ev)$
		where $\Gm \in \F_q^{k \times n}$ and $\uv \in \F_q^k$ are sampled uniformly at random over their domain - which generates a random codeword $\cv = \uv \Gm$ -  and $\ev$ is sampled from the distribution $f$.
		\item Goal: from $(\Gm,\cv+\ev)$, find $\cv$.
	\end{itemize}	
\end{problem} 

This problem for random codes has been studied for a long time and despite many efforts on this issue, the best algorithms are exponential in the codelength $n$ for natural noise distributions $f$ in the regime where $k$ is linear in $n$ and the rate $R \eqdef \frac{k}{n}$ bounded away from 0 and 1
\cite{P62,S88,D89,MMT11,BJMM12,MO15,CDMT22}. 

The most common noise distribution studied in this context is the uniform distribution over the errors of fixed Hamming weight $t$, but there are also other distributions, like in the binary case ($q=2$) the i.i.d Bernoulli distribution model which is frequently found in the Learning Parity with Noise problem (\LPN) \cite{GGR98}. When the number of samples of the \LPN\  problem is fixed, this is exactly the 
decoding problem defined above where $n$ is equal to the number of available \LPN \ samples. When the number of samples in \LPN\  is unlimited, this can be viewed as a decoding problem where we might add on the fly as many columns in $\Gm$ as we need (and as many corresponding positions in $\uv \Gm +\ev$). The \LWE \  problem in its standard form \cite{R05} is a slight variation on the input alphabet, it is $\Zq$ rather than the finite field $\F_q$ and as in \LPN, the number of samples is often assumed to be unlimited. The noise distribution is frequently the discrete Gaussian distribution in this case. 

The fact that in \LPN, $n$ can grow unlimited with a fixed value of $k$ and a fixed noise distribution can only make the problem simpler than the decoding problem. 
Interestingly enough, there are now algorithms solving the \LPN \ problem like the Blum-Kalai-Wasserman algorithm \cite{BKW03} which solve the problem with only 
subexponential complexity of the form $2^{\OO{\frac{k}{\log k}}}$ whereas no algorithm with such a complexity is known for $n=\OO{k}$ (all known algorithms have exponential complexity in this case). Note that as soon as  $n=\Om{k^{1+\varepsilon}}$ for any absolute constant $\eps > 0$, the best known algorithm \cite{L05b} is somewhat in between, namely $2^{\OO{\frac{k}{\log \log k}}}$ and consists in building many new \LPN \ samples from the original pool of samples. In terms of the decoding problem given above, this consists in adding artificially new columns to the generator matrix $\Gm$ given above by summing a small number of columns of $\Gm$ (together with the relevant positions of $\uv \Gm + \ev$) to artificially enlarge the value of $n$ and then solve the new decoding problem for this larger matrix.
In our work, we will only be interested in the linear regime setting {\ie } $k = \Theta(n)$.

It should be added here that the \LWE \ problem has proved much more versatile than \LPN \ for building cryptographic primitives. Indeed, it does not only allow to build cryptosystems from it \cite{R05}, but also
allows to obtain advanced  cryptographic functionalities such as fully homomorphic encryption \cite{BV11} or attribute-based encryption \cite{GVW13} for instance. It should also be mentioned that three out of the four signature schemes, public key encryption schemes or key establishment protocols supposed to resist to a quantum computer which were selected by the NIST for standardization are based on the hardness of this problem (see \url{https://csrc.nist.gov/Projects/post-quantum-cryptography/selected-algorithms-2022}).

While the security of many code-based cryptosystems relies on the hardness of the decoding problem, it can also be based on finding a ``short'' codeword (as in \cite{MTSB12} or in \cite{AHIKV17,BLVW19,YZWGL19} to build collision resistant hash functions), a problem which is stated as follows. 

\begin{problem}[$\SCP(q,n,k,w)$] The short codeword problem with parameters $q,n,k,w \in \N$ is defined as:
	\begin{itemize}
		\item Given:  $\Hm \in\F_q^{(n-k)\times n}$ which is sampled uniformly at random,		
		\item Find: $\cv \in \F_q^n \setminus\{\zerom\}$ such that $\Hm \trsp{\cv} = \zerom$ and the weight $|\cv|$ of  $\cv$ satisfies $|\cv|\leq w$.
	\end{itemize}	
\end{problem} 

Here we are looking for a non-zero codeword $\cv$ of weight $ \leq w$ in the $k$-dimensional code $\CC$ defined by the so-called parity-check matrix $\Hm$, namely\footnote{The short codeword problem is usually defined by picking a random parity-check matrix $\Hm \in \F_q^{(n-k)\times n}$ and not a random generating matrix $\Gm \in \F_q^{k \times n}$ but the differences are minor (see for example~\cite{Deb23}) and one could also define this problem via the generating matrix of a code as we did for the decoding problem.}  :
\begin{equation*}
	\CC \eqdef \Cbra{\cv \in \F_q^n \colon \Hm \trsp{\cv} = \zerov}.
\end{equation*}
The weight function which is generally used here is the Hamming weight, \ie\ for a vector $\xv=(x_1,\cdots,x_n) \in \Fqn$, its Hamming weight is defined as
$$|\xv| \eqdef \#\{i \in \Iint{1}{n}: x_i \neq 0\}.
$$
We will only deal with this weight here. The lattice version of this problem is called the Short Integer Solution (\SIS) problem. It consists in replacing the finite field $\F_q$ by $\Zq$ and using as weight function
the euclidean weight $\sqrt{\sum_{i=1} x_i^2}$ (and by representing the elements in $\Zq$ as $\{-\lfloor (q-1)/2 \rfloor, \cdots,0,\cdots,\lceil (q-1)/2 \rceil\}$).
It was introduced in the seminal work \cite{A96} and used there to build a family of one-way functions based on the difficulty of this problem. What made this problem so attractive is that it was shown there
to be as hard on average as a worst case short lattice vector problem.

Decoding and looking for short codewords are problems that have been conjectured to be extremely close. They have been studied for a long time, and the best algorithms for solving these two problems 
are the same, namely Information Set Decoding algorithms \cite{P62,S88,D89,MMT11,BJMM12,MO15,BM17}. A reduction from decoding to the problem of finding short codewords is known but in an $\LPN$ context \cite{AHIKV17,BLVW19,YZWGL19,DR22}. However, until recently and even in an $\LPN$ context, no reduction was known in the other direction before \cite{DRT23} which gave a quantum reduction from \SCP\ to \DP \ which followed the path of the breakthrough result of \cite{R05} which reduced the problem of sampling short lattice vectors to \LWE. Note that the reduction \cite{R05} was not classical  but quantum. 
Later on, it was shown in \cite{SSTX09} that the quantum reduction technique of Regev allows to reduce quantumly \SIS\ to \LWE, and this kind of reduction also applies to structured versions of these problems, namely Ideal-\SIS \ can be reduced to Ideal-\LWE.

There is a fundamental difficulty of reducing the search of low weight codewords to decoding a linear code which is due to the fact that the nature of these two problems is very different. Decoding concentrates on a region of parameters where there is typically just one solution, whereas finding low weight codewords concentrates on a region of parameters where there are many solutions (and typically an exponential number of solutions). This makes these problems inherently very different. This was also the case for the reduction of \SIS\  to \LWE \ and the fact that we can have a reduction from one to another by looking for quantum reductions instead of classical reductions was really a breakthrough at that time.

It is also worthwhile to notice that all these problems, \DP, \LPN, \LWE, \SCP, \SIS \ are all widely believed to be hard also for a quantum computer. The best quantum algorithms for solving these problems have not changed much the picture, the complexity exponent gets essentially only reduced by a constant factor when compared to the best classical algorithms achieving this task, see for instance \cite{B10,BJLM13,KT17a,LMP15,L16,CL21}. Indeed, as explained above, most public-key cryptosystems and digital signature schemes that are being standardized right now by the NIST are based on the presumed hardness of \LWE, and there are also alternate fourth round finalists of the competition \cite{ABCCGLMMMNPPPSSSTW20,AABBBBDDGLPRVZ22,AABBBBDGGGGMPRSTVZ22} which are based on the hardness of binary \DP.

\subsection{Regev's quantum reduction and follow-up work}

Regev's quantum reduction\cite{R05} is at the core of complexity reductions for these problems, which with \cite{A96} essentially started lattice-based cryptography. His approach when rephrased in the coding context is based on the following observation. Suppose that we were able to construct a quantum superposition $\sqrt{\frac{1}{Z}}\sum_{\cv \in \C}\sum_{\ev \in \Fqn} \sqrt{f(e)} \ket{\cv+\ev}$ of noisy codewords of a code $\C$ of dimension $k$ over $\Fq$, for a normalization factor $Z$. If we would apply the quantum Fourier transform on such a state, then because of the periodicity property of such a state we would get a superposition concentrating solely on the codewords of the dual $\C^\perp$ of $\C$, that is $\frac{1}{\sqrt{Z}}\sum_{\cv^\perp \in \C^\perp} \sqrt{\QFT{f}(\cv^\perp)}\ket{\cv^\perp}$. Here $\QFT{f}$ is the (classical) Fourier transform of $f$ that we will properly define in the technical part of the paper. Recall that the dual code is defined as 
\begin{definition}[dual code]
	Let $\C$ be a linear code over $\Fq$, {\ie} a $k$-dimensional subspace of $\Fqn$ for some $k$ and $n$. The dual code $\C^\perp$ is an $(n-k)$ dimensional subspace of $\Fqn$ defined by
	$$
	\C^\perp \eqdef \{\dv \in \Fqn: \dv \cdot \cv =0,\; \forall \cv \in \C\},
	$$
	where $\xv \cdot \yv = \sum_i x_i y_i$ stands for the inner product between the vectors $\xv$ and $\yv$.
\end{definition}
Now, we can expect that if $f$ concentrates on fairly small weights, then $\QFT{f}$ would also concentrate on rather small weights and therefore we would have a way of sampling low weight 
(dual) codewords and solve \SCP \  for the dual code. The point is now that $\sqrt{\frac{1}{Z}}\sum_{\cv \in \C}\sum_{\ev \in \Fqn} \sqrt{f(e)} \ket{\cv+\ev}$ could be obtained 
by solving the \DP \ problem on states that are easy to construct. This is the main idea of Regev's reduction. More precisely, the whole algorithm works as follows

\paragraph{Step 1.} Creation of the tensor product of a uniform superposition of codewords and a quantum superposition of noise
$$
\ket{\phi_1} =  \sqrt{\frac{1}{q^k}} \sum_{\cv \in \C} \ket{\cv} \sum_{\ev \in \Fqn} \sqrt{f(e)} \ket{\ev}.
$$ 
\paragraph{Step 2.} Entangling the codeword with the noise by adding the first register to the second one and then swapping the two registers
$$
\ket{\phi_2} =  \sqrt{\frac{1}{q^k}} \sum_{\cv \in \C}  \sum_{\ev \in \Fqn} \sqrt{f(e)} \ket{\cv+\ev}\ket{\cv}.
$$
\paragraph{Step 3.} Disentangling the two registers by decoding $\cv+\ev$ and therefore finding $\cv$ which allows to erase the second register 
$$
\ket{\phi_2} =  \sqrt{\frac{1}{Z}} \sum_{\cv \in \C} \sum_{\ev \in \Fqn} \sqrt{f(e)} \ket{\cv+\ev}   \ket{\zerom}.
$$
(The different normalizing factor $Z$ arises when the above decoding procedure is imperfect and we condition on measuring $\mathbf{0}$ in the last register.) 
\paragraph{Step 4.} Applying the quantum Fourier transform on the first register and get
$$
\frac{1}{\sqrt{Z}}\sum_{\dv \in \C^\perp} \sqrt{\QFT{f}(\dv)}\ket{\dv}\ket{\zerom}
$$
\paragraph{Step 5.} Measure the first register and get some $\dv$ in $\C^\perp$.

This approach is at the heart of the quantum reductions obtained in \cite{R05,SSTX09,DRT23}. It is also a crucial ingredient in the paper \cite{YZ22} proving verifiable quantum advantage by constructing - among other things -  one-way functions that are even collision resistant against classical adversaries but are easily invertible quantumly.
In \cite{R05,SSTX09,DRT23}, the crucial erasing/disentangling step is performed with the help of a {\em classical} decoding algorithm. Indeed any (classical or quantum) algorithm that can recover $\cv$ from $\cv + \ev$ can be applied coherently to erase the last register in step $3$\footnote{Indeed, having such an algorithm means we can construct the unitary $U : \ket{\cv + \ev}\ket{0} \rightarrow \ket{\cv + \ev}\ket{\cv}$. Applying the inverse of this unitary will give the erasure operation.} .

A key insight observed in~\cite{CLZ22} is that it is actually enough to recover $\ket{\cv}$ from the state $\sum_{\ev \in \Fqn} \sqrt{f(\ev)} \ket{\cv + \ev}$ so we are given a superposition of all the noisy codewords $\cv + \ev$ and not a fixed one. This means we have to solve the following problem

\begin{problem}[$\QDP(q,n,k,f)$] The quantum decoding problem with  positive integer parameters  $q,n,k$ and a probability distribution $f$ on $\Fqn$ is defined as:
	\begin{itemize}
		\item Input: Take $\Gm \in \F_q^{k \times n}$ and $\uv \in \F_q^k$ sampled uniformly at random over their domain. Let $\cv = \uv \Gm$ and $\ket{\psi_\cv} \eqdef \sum_{\ev \in \Fqn} \sqrt{f(e)} \ket{\cv+\ev}$. The (quantum) input to this problem is $(\Gm,\ket{\psi_\cv})$.
		\item Goal: given $(\Gm,\ket{\psi_\cv})$, find $\cv$.
	\end{itemize}	
\end{problem} 

It's not clear a priori whether this is helpful or not. If one measures the state $\ket{\psi_\cv}$ then one recovers a noisy codeword and we are back to the classical decoding problem. 

However, having improvements by directly solving the ({\LWE} variant of the) above problem has been proposed in \cite{CLZ22} where a polynomial time quantum algorithm based on Regev's approach solving \SIS \ is proposed for the $l_\infty$ norm (and not the euclidean norm as is standard there) for extremely high rate codes. Here the decoding problem is obtained by measuring the qudits in an appropriate basis allowing to rule out certain values for the code-symbols, and then they use the Arora-Ge algorithm \cite{AG11} for recovering completely the codeword by solving an algebraic system which for the parameters that are considered there, is of polynomial complexity. Despite the fact that the parameters of the \SIS \ problem are highly degenerate, no efficient classical algorithm performing this task is known. This paper puts forward the \SLWE \  and the \CLWE\  problems. Informally the first problem is the one we solve in Step 3 above and the second one is just to create directly the uniform superposition of noisy codewords obtained at Step 3.

\subsection{Contributions}

Our work has $2$ starting points. First, the quantum reduction of~\cite{DRT23} between the short codeword problem and the decoding problem in the regime relevant for code-based cryptography {\ie} a constant code rate $\frac{k}{n}$ and constant error rate. Then, the key insight of~\cite{CLZ22} that one requires  to solve the quantum decoding problem in the above reduction which can make it more efficient. Instead on focusing too much on the reduction, our aim is first to study here the quantum decoding problem for its own sake. Indeed, the problem is already interesting as a quantum generalization of the decoding problem and the fact it is used in the above reduction creates strong motivation for studying it. 

In this work, we focus only on the Bernoulli noise of parameter $p$. This means we consider the error function 
$$
f(\ev) = (1-p)^{n-|\ev|}\left(\frac{p}{q-1}\right)^{|\ev|}.
$$
which in turn means that for any $\cv = (c_1,\dots,c_n) \in \F_q^n$, we can rewrite 
$$ \ket{\psi_\cv} \eqdef \sum_{\ev \in \F_q^n} \sqrt{f(\ev)} \ket{\cv + \ev} = \bigotimes_{i = 1}^n  \left(\sqrt{1 - p} \ket{c_i} + \sum_{\alpha \in \F_q^*} \sqrt{\frac{p}{q-1}} \ket{c_i+\alpha}\right).$$

For this Bernoulli noise with parameter $p$, the associated quantum decoding problem is written $\QDP(q,n,k,p)$. We show that indeed, the complexity of the quantum decoding problem significantly differs from its classical counterpart. Our contributions can be summarized as

\paragraph{A polynomial time algorithm for \QDP \ when the noise is low enough (but still of constant rate).} 
We will show that the quantum problem \QDP \ defined here is probably much easier than its classical counterpart \DP. Indeed, for fixed rate $R \eqdef \frac{k}{n}$ only exponential-time algorithms are known for \DP \ for natural noise models, for instance the Bernoulli i.i.d model where $q=2$, $\Pr(e_i=1)=p$ for which all algorithms solving it are 
exponential for $p$ in $(0,1)$. This is not the case for the associated \QDP \ problem, where we will show that by using Unambiguous State Discrimination (USD) together with linear algebra we can solve 
the problem in {\em polynomial time} up to some limiting value of $p$ which is strictly between $0$ and $1$ for a fixed rate $R$. We generalize this result for any $q$ by generalizing existing bounds on USD and also present an algorithm for partial binary unambiguous state discrimination which could be of independent interest.

\paragraph{A problem which can be solved above capacity.} There is an information theoretic  limit for {\em any algorithm} solving classically or quantumly the classical decoding problem 
\DP. When the rate exceeds the capacity of the noisy channel specified by $f$ (and if this is an i.i.d. noise) then above the capacity of the noisy channel it is just impossible to solve
with say polynomial error probability the decoding problem just because there are exponentially many candidates at least as likely as the right candidate. The problem becomes intractable just because of this reason. For instance in the Bernoulli model above, the rate $R=k/n$ has to be smaller than $1-h(p)$ where $h(x)$ is the binary entropy function, $h(x)=-x \log_2(x)-(1-x)\log_2(1-x)$.
Somewhat surprisingly, it turns out that we can go above the Shannon capacity for the \QDP \ problem. Moreover, with the help of the Pretty Good Measurement (PGM) we can fully characterize the noise range where the problem is tractable.

	\paragraph{Applying \QDP \ solvers in Regev's reduction.}
	{Both algorithms (the one using USD and the other one based on PGM) can be applied to sample small weight
	dual codewords and solve \SCP. By applying the quantum reduction steps above, together with our polynomial time solving \QDP \ we obtain non-zero codewords of relative weight $\omega \eqdef w/n$ satisfying $\omega \leq \frac{(q-1)(1-R)}{q}$. Interestingly enough, this is precisely the smallest weight that can be reached by the best known polynomial time algorithm, namely a minor variant of the Prange algorithm \cite{P62}.
	On the other hand, we will show that there is no hope to have a proper general reduction of \SCP \ to \QDP, by providing examples showing that we can solve \QDP \ in a certain noise regime and still get nothing useful for \SCP\ after using it in Regev's reduction. However, we can adapt the PGM to still have some small codewords up to the tractability bound.  Our examples really show that we have to analyze properly the state that we have at Step 3. of the reduction on a case by case basis. \\
}

We now perform a detailed description of our contributions.

\subsubsection{Using USD  as a means of improving quantum algorithms for \QDP}
\paragraph{The binary setting.}
 Our first idea, which extends naturally the work of \cite{CLZ22} is to apply USD for the quantum decoding problem. 
  \COMMENT{ Note that due to the fact that from the $l_\infty$ norm choice in the \SIS/\SCP\  problem considered there, the states that have to be discriminated can not be linearly independent and therefore USD can not work. Note also that the smallest alphabet size that can be handled by the \cite{CLZ22} approach is the ternary alphabet $\F_3$. We will show however that in our setting using USD is not only possible but can also be applied to a binary alphabet. }
 We first consider the binary setting, ${\ie }$  $q = 2$. This means the states $\ket{\Psi_\cv}=\sum_{\ev \in \F_2^n}\sqrt{f(\ev)}\ket{\cv+\ev}$ for which we want to recover $\cv$ are of the form 

$$ \ket{\Psi_\cv} = \bigotimes_{i = 1}^n \sqrt{1 - p} \ket{c_i} + \sqrt{p} \ket{1 - c_i}.$$

Consider a fixed coordinate $i$ for which we have the state $ \sqrt{1 - p} \ket{c_i} + \sqrt{p} \ket{1 - c_i}$ which we call $\ket{\psi^p_{c_i}}$. By measuring this state in the computational basis we get $c_i$ wp. $1 - p$ and $1 - c_i$ wp. $p$. This measurement is actually the measurement that distinguishes best $\ket{\psi_0}$ and $\ket{\psi_1}$.

Another measurement of interest is unambiguous state discrimination. 
Here, the goal is not to distinguish optimally between $\ket{\psi^p_0}$ and $\ket{\psi^p_1}$ but to make sure that our guess is always correct but allowing for some abort. In this setting, we have the following 

\begin{proposition}[Unambiguous state discrimination]
	For any $p \in [0,1]$, there exists a quantum measurement that on input $\ket{\psi^p_{c_i}}$ outputs $c_i$ wp. $1 - 2\sqrt{p(1-p)}$ and outputs $\bot$ otherwise.
\end{proposition}

Using this measurement, the probability of guessing correctly $c_i$ is always smaller than $1-p$ for $p \le \frac{1}{2}$. However, we know exactly when we succeed in guessing $c_i$. This will be extremely useful for decoding. Indeed, if we recover $k$ values of $c_i$ we recover the complete codeword $\cv$ with good probability by linear algebra by using the fact that
$\cv = \mv \Gm$ with $\Gm \in \F_2^{k \times n}$. This will lead to

\begin{theorem}\label{Theorem:1}
	Let $R \in [0,1]$. For any $p < \perpF{\frac{R}{2}}$, there exists a quantum algorithm that solves $\QDP(2,n,\lfloor Rn \rfloor,p)$ wp. $1 - 2^{-\Omega(n)}$ in time $\poly(n)$.
Here for a real number $x \in [0,1]$, $\perpO{x}$ stands for $\frac{1-2\sqrt{x(1-x)}}{2}$.
\end{theorem}

\paragraph{Interpretation as changing the noise channel and partial unambiguous state discrimination.}
A nice interpretation of the above algorithm is that when the error is in quantum superposition, one can use quantum measurements to change the noise model. For example in the above, if we are given $\ket{\psi_{c_i}} = \sqrt{1-p}\ket{c_i} + \sqrt{p} \ket{1 - c_i}$ then
\begin{itemize}\setlength\itemsep{-0.17em}
	\item One can measure in the computational basis to obtain $c_i$ that has been flipped wp. $p$. 
	\item One can use unambiguous state discrimination in which case $c_i$ has been erased wp. $2\sqrt{p(1-p)}$. 
\end{itemize}

What we show in Theorem~\ref{Theorem:1} is that the second strategy is actually much more powerful for recovering the codeword $\cv$. A natural question to ask is whether this can further be generalized to other measurements. 

In this work, we actually generalize Unambiguous State Discrimination as follows: given $\ket{\psi_{c_i}}$, the measurement will sometimes output $\bot$ but it can also fail with some small probability. We prove the following

\begin{proposition}[Partial Unambiguous State Discrimination]
	Let $p,s \in [0,\frac{1}{2})$ with $s \le p$ and let $u = \frac{p^\bot}{s^\bot}$. There exists a quantum measurement that when applied to $\ket{\psi_{c_i}} = \sqrt{1-p} \ket{c_i} + \sqrt{p}\ket{1-c_i}$ outputs $c_i$ wp. $u(1-s)$, $(1-c_i)$ wp. $us$ and $\bot$ wp. $1-u$. 
\end{proposition}

Notice that this generalizes both the standard measurement (by taking $s = p$) and unambiguous state discrimination (by taking $s = 0$ which gives $u = 2p^\bot$). This seems a very natural way of generalizing Unambiguous State Discrimination but is not something we have found in the literature and could be of independent interest. We can use this measurement not to provide new polynomial time algorithm but rather to give a reduction between different Quantum Decoding problems, which we detail in the full text.

\COMMENT{

\begin{theorem}
	Let $R \in (0,1)$. Let $p \in [0,\frac{1}{2})$ and $q \in [0,\frac{1}{2})$ satisfying: $q \le p$ and $\frac{p^\perp}{q^\perp}  > R$.
	Let any $u > \frac{p^\perp}{q^\perp}$. Let also $k = \lfloor Rn \rfloor$.
	Then $\DPQ(2,n,k,p) \preccurlyeq \DPQ(2,\lfloor un\rfloor,k,q)$ meaning that if we have an algorithm that solves $\DPQ(2,\lfloor un\rfloor,k,q)$, we can use it to solve $\DPQ(2,n,k,p)$.
\end{theorem}

}

\paragraph{The general setting.}
The unambiguous state discrimination approach works in the $q$-ary setting as well.
A difficulty here is that optimal unambiguous state discrimination is not known in general for more than $2$ states, but in certain situations where we have 
a symmetric set of states \cite{CB98} we know how to perform optimal USD. This would apply in our case case where $q$ is prime.  We have generalized sligthly the approach of \cite{CB98} to be able to apply it to any finite field size $q$. We get finally a result very similar to the binary case

\begin{theorem}\label{Theorem:USD}
	Let $R \in [0,1]$. For any $p < \perpF{\frac{(q-1)R}{q}}$, there exists a quantum algorithm that solves $\QDP(q,n,\lfloor Rn \rfloor,p)$ wp. $1 - 2^{-\Omega(n)}$ in time $\poly(n)$.
\end{theorem}

Here we have used a notation which ``generalizes'' the $\perpO{p}$ notation used in the binary setting.
\begin{notation}
For a real number $x \in [0,1]$, $\perpO{x}$ stands for $\frac{\left(\sqrt{(1-x)(q-1)}-\sqrt{x}\right)^2}{q}$.
\end{notation}
This quantity depends on $q$ which will be clear from the context. Note that when $q=2$ we get $\frac{1-2\sqrt{x(1-x)}}{2}$ which coincides with the one given in the binary case.

\subsubsection{Determining exactly the tractability of the quantum decoding problem}
We are now interested in the tractability of $\QDP(q,n,k,p)$ meaning when is it possible from an information theoretic perspective to solve this problem. In order to study this problem, a fundamental quantity is $\drgv(R)$ defined below,  sometimes referred to as the Gilbert-Varshamov distance 

\begin{notation}
	Let $R \in [0,1]$. We define $\drgv(R) \eqdef h_q^{-1}(1 - R)$, where $h_q(x) \eqdef -(1-x) \log_q(1-x) - x \log_q\left(\frac{x}{q-1} \right)$. $h_q$ is a bijection from $x \in \Sbra{0,\frac{q-1}{q}}$ to $[0,1]$ and we define  $h_q^{-1} : [0,1] \rightarrow \Sbra{0,\frac{q-1}{q}}$ st. $h_q^{-1}(h_q(x)) = x$ for $x \in \Sbra{0,\frac{q-1}{q}}$.
\end{notation}

For the classical setting, it is well understood that $\DP(q,n,k,p)$ is not tractable when $p > \drgv(\frac{k}{n})$, meaning that even an unbounded algorithm will solve the problem wp. $o(1)$. 

We would like now to understand what happens in the quantum setting. 
Techniques based on (partial) unambiguous state discrimination will not work in the regime $p > \drgv(R)$. Since we are only interested in the tractability of the problem, we can consider optimal quantum algorithms for discriminating between the states $\ket{\Psi_\cv}=\sum_{\ev \in \Fqn} \sqrt{f(\ev)}\ket{\cv+\ev}$ where $f$ accounts for the Bernoulli noise of parameter $p$. This problem can be addressed by using the {\em Pretty Good Measurement} (PGM) which has turned out to be
a very useful tool in quantum information. If we define $\PPGM$ as the probability that the pretty good measurement succeeds in solving our problem and define $\POPT$ as the maximal probability that any measurement succeeds, we have~\cite{BK02,Mon06}
$$ \POPT^2 \le \PPGM \le \POPT.$$
This means that if the problem is tractable then $\POPT = \Omega(1)$ which implies $\PPGM = \Omega(1)$. On the other hand, if the problem is intractable then $\POPT = o(1)$ which implies $\PPGM = o(1)$. In conclusion, in order to study the tractability of the quantum decoding problem, it is enough to look at the PGM associated with the problem of distinguishing the states $\{\ket{\Psi_\cv}\}$. We show the following
\begin{theorem}\label{Theorem:PGM}
	Let $R \in (0,1)$.
\begin{itemize} \setlength\itemsep{-0.2em}
	\item For $p < \perpF{\drgv(1-R)}$, $\QDP(q,n,\lfloor Rn \rfloor,p)$ can be solved using the PGM wp. $\PPGM = \Omega(1)$  hence the problem is tractable. 
	\item For $p > \perpF{\drgv(1-R)}$, $\QDP(q,n,\lfloor Rn \rfloor,p)$, the probability that the PGM solves this problem is $P_{PGM} = o(1)$ hence the problem is intractable. 
\end{itemize}
\end{theorem}

The pretty good measurement associated to this distinguishing problem actually has a a lot of structure. It is actually a projective measurement on an orthonormal basis corresponding which can be seen as a Fourier basis involving the shifted dual codes of the code $\C$ we are working on.

\paragraph{Comparing the complexity of the decoding problem and the quantum decoding problem}

With this full characterization, we compare the hardness, and tractability of the classical and quantum decoding problems. For $p = 0$, we have of course a polynomial time algorithm to solve $\DP(q,n,\lfloor R n \rfloor,0)$. For $0 < p \le \drgv(R)$, the problem is tractable and the best known classical or quantum algorithms run in time $2^{\Omega(n)}$. For $p > \drgv(R)$, we know the problem is intractable. In the quantum setting, we obtain a very different picture. A comparison of these results is presented in Figures~\ref{fig:hardEasyDP} and~\ref{Fig:2} where we use the following terminology
\begin{itemize}\setlength\itemsep{-0.2em}
	\item Easy: there exists an algorithm that runs in time $\poly(n)$.
	\item Hard: the best known algorithm runs in time $2^{\Omega(n)}$, but there could potentially be more efficient algorithms.
	\item Intractable: we know that any (even unbounded) algorithm can solve the problem wp. at most $o(1)$. 
\end{itemize}

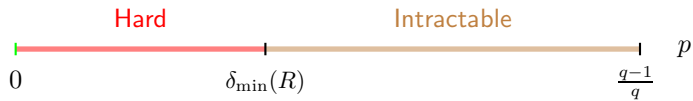
\begin{figure}[htb]
	\centering
	\caption{Hardness and tractability of the decoding problem $\DP(q,n,\lfloor Rn \rfloor,p)$, for any fixed $R \in [0,1]$, as a function of $p$. 
	\label{fig:hardEasyDP}
	}
	
	\begin{tikzpicture}[scale=0.83]
	\tikzstyle{valign}=[text height=1.5ex,text depth=.25ex]
		\draw (3,2.5) node[red]{{\sf Hard}};
	\draw (8,2.5) node[brown]{{\sf Intractable}};
			\draw[line width=2pt,red!50] (1,2) -- (5,2);
	\draw[line width=2pt,brown!50] (5,2) -- (11,2);
			\draw[-,>=latex,line width=2pt,gray] (11,2) -- (11,2)
	node[right,black] {$\displaystyle \quad p$};
	\tikzstyle{valign}=[text height=2ex]
	\draw[thick,green] (1,1.9) node[below,valign,black]{$0$} -- (1,2.1);
	\draw[thick] (5,1.9) node[below,valign]{\scalebox{0.9}{$\drgv(R)$}} -- (5,2.1);
	\draw[thick] (11,1.9) node[below,valign]{\scalebox{0.9}{$\frac{q-1}{q}$}~~} -- (11,2.1);
				\end{tikzpicture} \end{figure}

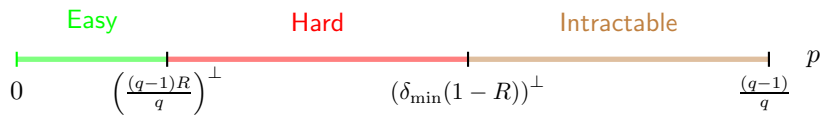
\begin{figure}[htb]
	\centering
	\caption{Hardness and tractability of the quantum decoding problem $\QDP(q,n,\lfloor Rn \rfloor,p)$, for any fixed $R \in [0,1]$, as a function of $p$.
	\label{Fig:2}}

\begin{tikzpicture}[scale=1]
	\tikzstyle{valign}=[text height=1.5ex,text depth=.25ex]
		\draw (2,2.5) node[green]{{\sf Easy}};
	\draw (5,2.5) node[red]{{\sf Hard}};
	\draw (9,2.5) node[brown]{{\sf Intractable}};
			\draw[line width=2pt,green!50] (1,2) -- (3,2);
	node[above,midway,green,valign]{{\sf Easy}} (3,2.5);
	\draw[line width=2pt,red!50] (3,2) -- (7,2);
	node[above,midway,red,valign]{{\sf Hard}} (7,2.5);
	\draw[line width=2pt,brown!50] (7,2) -- (11,2);
	node[above,midway,brown,valign]{{\sf Intractable}} (11,2.5);
		\draw[-,>=latex,line width=2pt,gray] (11,2) -- (11,2)
	node[right,black] {$\displaystyle \quad p$};
	\tikzstyle{valign}=[text height=2ex]
	\draw[thick,green] (1,1.9) node[below,valign,black]{$0$} -- (1,2.1);
	\draw[thick] (3,1.9) node[below,valign]{\scalebox{0.9}{$\perpF{\frac{(q-1)R}{q}}$}} -- (3,2.1);
	\draw[thick] (7,1.9) node[below,valign]{\scalebox{0.9}{$\perpF{\drgv(1-R)}$}} -- (7,2.1);
	\draw[thick] (11,1.9) node[below,valign]{\scalebox{0.9}{$\frac{(q-1)}{q}$}~~} -- (11,2.1);
				\end{tikzpicture}
\end{figure}

This gives a proper characterization of the difficulty of the Quantum Decoding Problem. In our next contribution, we will apply them in Regev's quantum reduction in order to derive some results for the short codeword problem. As we will show, the results from Figure~\ref{Fig:2} will match exactly our knowledge for the short codeword problem.
\subsubsection{Using our algorithms in Regev's reduction}

We are now interested in solving the short codeword problem using Regev's reduction and the algorithms we described in the previous section. The known hardness of the short codeword problem is summarized in the figure below

\begin{figure}[htb]
	\centering
	\caption{Hardness and tractability of the short codeword problem $\SCP(q,n,\lfloor Rn\rfloor, p)$ for a fixed $R \in (0,1)$, as a function of $p$.
	\label{fig:hardEasySC}
	}
		\begin{tikzpicture}[scale=0.83]
	\tikzstyle{valign}=[text height=1.5ex,text depth=.25ex]
		\draw (3,2.5) node[brown]{{\sf Intractable}};
	\draw (6.5,2.5) node[red]{{\sf Hard}};
	\draw (9.5,2.5) node[green]{{\sf Easy}};
			\draw[line width=2pt,brown!50] (1,2) -- (5,2);
		\draw[line width=2pt,red!50] (5,2) -- (8,2);
		\draw[line width=2pt,green!50] (8,2) -- (11,2);
			\draw[-,>=latex,line width=2pt,gray] (11,2) -- (11,2)
	node[right,black] {$\displaystyle \quad \omega$};
	\tikzstyle{valign}=[text height=2ex]
	\draw[thick,brown] (1,1.9) node[below,valign,black]{$0$} -- (1,2.1);
	\draw[thick] (5,1.9) node[below,valign]{\scalebox{0.9}{$\drgv(R)$}} -- (5,2.1);
	\draw[thick] (8,1.9) node[below,valign]{\scalebox{0.9}{$\frac{(q-1)(1-R)}{q}$}} -- (8,2.1);
	\draw[thick] (11,1.9) node[below,valign]{\scalebox{0.9}{$\frac{(q-1)}{q}$}~~} -- (11,2.1);
				\end{tikzpicture}
\end{figure}
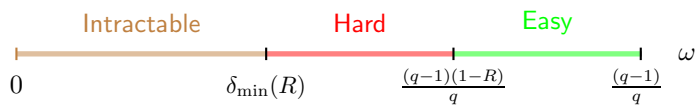

For our coding context, the only known reduction is the following
\begin{proposition}[\cite{DRT23}, informal]\label{Proposition:ReductionDRT}
	Fix integers $n,q \ge 2$ as well as parameters $R,p \in (0,1)$ st. $p \le \drgv(R)$. From any quantum algorithm that solves $\DP(q,n,\lceil (1-R)n \rceil,p)$ with high probability, there exists a quantum algorithm that solves $\SCP(q,n,\lfloor R n\rfloor,p^\perp)$ with high probability where recall that 
	$ p^\perp = \frac{\left(\sqrt{(1-p)(q-1)}-\sqrt{p}\right)^2}{q}.$
\end{proposition}

How can we characterize the efficiency of this reduction?  Let us consider the best algorithms for $\DP(q,n,\lceil (1-R)n\rceil,p^\perp)$ and see what algorithms does it give for $\SCP(q,n,\lfloor Rn\rfloor,p)$. We obtain the following result, summarized in Figure~\ref{fig:dualDPSCP}. We can see that the obtained algorithm for the short decoding problem is significantly worse\footnote{One can check that we always have $\left(\delta{\min}(1-R)\right)^\perp \ge \delta_{\min}(R)$.} than best known algorithm for this problem (see Figure~\ref{fig:hardEasySC}). But in the light of our previous results, this is understandable, Regev's reduction actually requires to solve the quantum decoding problem and we just showed that it is much simpler than the decoding problem. If we could directly use the above proposition with our algorithms, we would obtain the following results, summarized in Figure~\ref{fig:DualQDP}.

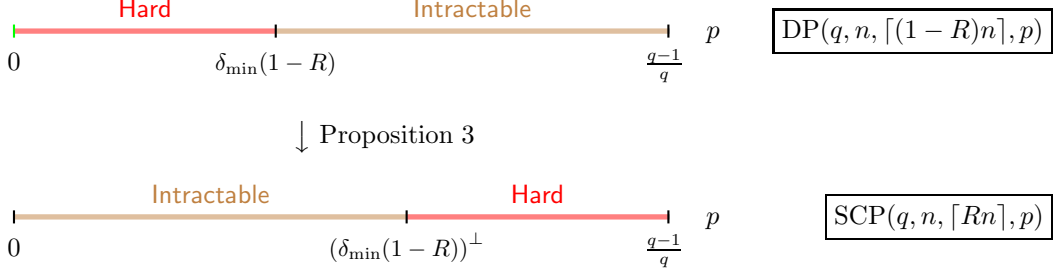
\begin{figure}[htb]
	\centering
	\caption{On the top, best known (classical or quantum) algorithms for $\DP(q,n,\lceil (1-R)n \rceil,p)$. On the bottom, complexity of a quantum algorithm for $\SCP(q,n,\lceil Rn \rceil,p)$ that uses the best algorithm for $\DP(q,n,\lceil (1-R) n\rceil,p)$ and then uses Proposition~\ref{Proposition:ReductionDRT} }
	\begin{tikzpicture}[scale=0.87]
	\tikzstyle{valign}=[text height=1.5ex,text depth=.25ex]
		\draw (3,2.35) node[red]{{\sf Hard}};
			\draw[line width=2pt,red!50] (1,2) -- (5,2);
	\draw[line width=2pt,brown!50] (5,2) --
	node[above,midway,brown,valign]{{\sf Intractable}} (11,2);
		\draw[-,>=latex,line width=2pt,gray] (11,2) -- (11,2)
	node[right,black] {$\displaystyle \quad p\qquad $\fbox{$\DP(q,n,\lceil (1-R) n\rceil,p)$}};
	\tikzstyle{valign}=[text height=2ex]
	\draw[thick,green] (1,1.9) node[below,valign,black]{$0$} -- (1,2.1);
	\draw[thick] (5,1.9) node[below,valign]{\scalebox{0.9}{$\drgv(1-R)$}} -- (5,2.1);
	\draw[thick] (11,1.9) node[below,valign]{\scalebox{0.9}{$\frac{q-1}{q}$}~~} -- (11,2.1);
				\end{tikzpicture}
	\begin{center}
	\hspace*{-4cm}	$\big\downarrow$ Proposition~\ref{Proposition:ReductionDRT}
	\end{center}
	\begin{tikzpicture}[scale=0.87]
	\tikzstyle{valign}=[text height=1.5ex,text depth=.25ex]
			\draw (4,2.35) node[brown]{{\sf Intractable}};
			\draw[line width=2pt,brown!50] (1,2) -- (7,2);
	node[above,midway,red,valign]{{\sf Hard}} (7,2);
	\draw[line width=2pt,red!50] (7,2) --
	node[above,midway,red,valign]{{\sf Hard}} (11,2);
		\draw[-,>=latex,line width=2pt,gray] (11,2) -- (11,2)
	node[right,black] {$ \quad p \qquad \qquad$\fbox{$\SCP(q,n,\lceil Rn \rceil,p)$}};
	\tikzstyle{valign}=[text height=2ex]
	\draw[thick] (1,1.9) node[below,valign,black]{$0$} -- (1,2.1);
			\draw[thick] (7,1.9) node[below,valign]{\scalebox{0.9}{$\perpF{\drgv(1-R)}$}} -- (7,2.1);
	\draw[thick] (11,1.9) node[below,valign]{\scalebox{0.9}{$\frac{q-1}{q}$}~~} -- (11,2.1);
				\end{tikzpicture}

		\label{fig:dualDPSCP}
\end{figure}

$ \ $ 

\begin{figure}[htb]
	\centering
	\caption{On the top, our quantum algorithms for $\QDP(q,n,\lceil (1-R) n\rceil,p)$. On the bottom, complexity of a quantum algorithm for $\SCP(q,n,\lceil Rn \rceil,p)$ that would use our algorithms $\QDP(q,n,\lceil (1-R) n \rceil,p)$ and then Proposition~\ref{Proposition:ReductionDRT} when applicable}
	\begin{tikzpicture}[scale=0.87]
	\tikzstyle{valign}=[text height=1.5ex,text depth=.25ex]
		\draw (2,2.35) node[green]{{\sf Easy}};
	\draw (5,2.35) node[red]{{\sf Hard}};
			\draw[line width=2pt,green!50] (1,2) -- (3,2);
	node[above,midway,green,valign]{{\sf Easy}} (3,2);
	\draw[line width=2pt,red!50] (3,2) -- (7,2);
	node[above,midway,red,valign]{{\sf Hard}} (7,2);
	\draw[line width=2pt,brown!50] (7,2) --
	node[above,midway,brown,valign]{{\sf Intractable}} (11,2);
		\draw[-,>=latex,line width=2pt,gray] (11,2) -- (11,2)
	node[right,black] {$\displaystyle \quad p\qquad $\fbox{$\QDP(q,n,\lceil (1-R)n \rceil,p)$}};
	\tikzstyle{valign}=[text height=2ex]
	\draw[thick,green] (1,1.9) node[below,valign,black]{$0$} -- (1,2.1);
	\draw[thick] (3,1.9) node[below,valign]{\scalebox{0.9}{$\perpF{\frac{(q-1)(1-R)}{q}}$}} -- (3,2.1);
	\draw[thick] (7,1.9) node[below,valign]{\scalebox{0.9}{$\perpF{\drgv(R)}$}} -- (7,2.1);
	\draw[thick] (11,1.9) node[below,valign]{\scalebox{0.9}{$\frac{(q-1)}{q}$}~~} -- (11,2.1);
				\end{tikzpicture}
	\begin{center}
	\hspace*{-2cm}	$\big\downarrow$ Proposition~\ref{Proposition:ReductionDRT} if applicable \\
	\end{center}
	\begin{tikzpicture}[scale=0.87]
	\tikzstyle{valign}=[text height=1.5ex,text depth=.25ex]
		\draw (3,2.35) node[brown]{{\sf Intractable}};
	\draw (6.5,2.35) node[red]{{\sf Hard}};
			\draw[line width=2pt,brown!50] (1,2) -- (5,2);
		\draw[line width=2pt,red!50] (5,2) -- (8,2);
		\draw[line width=2pt,green!50] (8,2) --
	node[above,midway,green,valign]{{\sf Easy}} (11,2);
		\draw[-,>=latex,line width=2pt,gray] (11,2) -- (11,2)
	node[right,black] {$ \quad p \qquad \qquad$\fbox{$\SCP(q,n,\lceil Rn \rceil,p)$}};
	\tikzstyle{valign}=[text height=2ex]
	\draw[thick,brown] (1,1.9) node[below,valign,black]{$0$} -- (1,2.1);
	\draw[thick] (5,1.9) node[below,valign]{\scalebox{0.9}{$\drgv(R)$}} -- (5,2.1);
	\draw[thick] (8,1.9) node[below,valign]{\scalebox{0.9}{$\frac{(q-1)(1-R)}{q}$}} -- (8,2.1);
	\draw[thick] (11,1.9) node[below,valign]{\scalebox{0.9}{$\frac{(q-1)}{q}$}~~} -- (11,2.1);
				\end{tikzpicture}

		\label{fig:DualQDP}
\end{figure}
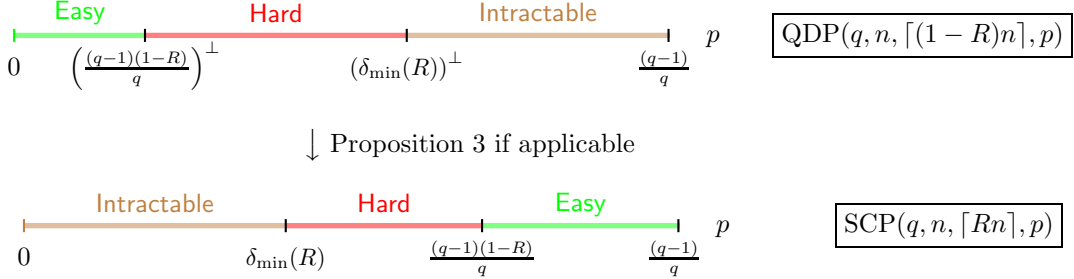

Here, if we could apply Proposition~\ref{Proposition:ReductionDRT} with our algorithms, we would recover exactly the same complexities as the best known algorithms for $\SCP$\footnote{We say we recover the same complexities only in the sense that we recover the areas which are easy,hard,intractable. We're not claiming that within these areas, the running times are exactly the same.}. However, it's not clear whether this is the case. What we do is that for each of our algorithms, we try to perform Regev's reduction and see what we obtain. We show the following:
\begin{itemize}
	\item If we take our polynomial time algorithms involving unambiguous state discrimination for the quantum decoding problem in Regev's reduction,  we can find in quantum polynomial time small codewords down to Prange's bound, {\ie} down to $\frac{(1-R)(q-1)}{q}$ (the Easy zone in Figure~\ref{fig:hardEasySC}). The bound $\perpF{\frac{(1-q)(1-R)}{q}}$ comes from bounds on quantum unambiguous state discrimination and it is quite remarkable that after the quantum reduction, it corresponds exactly to Prange's bound where the short codeword problem is easy.
	\item If we use our algorithm involving the Pretty Good Measurement in Regev's reduction, the following happens: 
		 \begin{enumerate}
		\item If we apply the PGM directly, we will most often be in regimes where we measure $\mathbf{0}$ in the final step so we will not be able to solve the Short Codeword problem. 
		\item We can slightly tweak the PGM so that we can solve the corresponding short codeword problem for all the regimes where it is tractable (the Hard zone in Figure~\ref{fig:hardEasySC}).
		\item We also show another example where we can slightly tweak the PGM but where the reduction utterly fails, meaning that the state we obtain after Step $4$ is $\ket{\bot}$, so measuring will give absolutely no information about a small dual codeword. This shows that there is no hope to perform a generic reduction ({\ie } a generalization of Proposition~\ref{Proposition:ReductionDRT}) between the quantum decoding problem and the short codeword problem with this method. 
	\end{enumerate}
\end{itemize}

These results show that, while it is impossible to have a generic reduction from $\QDP$ to $\SCP$ with this method, it is - at least for our examples - possible to find algorithms for $\QDP$ that will give results according to Figure~\ref{fig:DualQDP}, and recover the areas where the problem is easy and where it is tractable. This can be seen as quite a surprise since our bounds on $\QDP$ essentially come from information theory and best known bounds on $\SCP$ comes from classical coding theory and seem unrelated at first.

 \newcommand{\Mm}{\mathbf{M}}

\section{Preliminaries}
\subsection{Notations and basic probabilities}\label{ss:notation}
\paragraph{Sets, finite field.}
The finite field with $q$ elements is denoted by $\Fq$. 
$\Zq$ denotes the ring of integers modulo $q$.
The cardinality of a finite set $\Ec$ is denoted by $|\Ec|$.
The set of integers $\{a,a+1,\cdots,b\}$ between the integers $a$ and $b$ is denoted by $\Iint{a}{b}$. For a positive integer $n$, $[n]$ denotes 
$\Iint{1}{n}$. $x \Unif S$ means that $x$ is sampled uniformly from the set $S$.
\paragraph{Vector and matrices.}
For a Hermitian matrix $\Mm$ we write that $\Mm \succeq 0$ when $\Mm$ is positive semi-definite. Vectors are row vectors as is standard in the coding community and $\trsp{\xv}$ denotes the transpose of a vector or a matrix. In particular, vectors will always be denoted by bold small letters and matrices with bold capital letters. For a subset $J \subseteq [n]$ of positions of the vector $\xv = x_1,\dots,x_n$,  $\xv_J=(x_j)_{j \in J}$ denotes the vector formed by the entries indexed by $J$. For a matrix $\Gm \in \F_q^{k \times n}$ and a subset of columns $J \subseteq [n]$, $\Gm_J \in \F_q^{k \times |J|}$ denotes the submatrix formed by its columns indexed by $J$. 

\begin{lemma}[Hoeffding's inequality]\label{lem:Hoeffding}
	Let $X_1,\dots,X_n$ be independent random Bernoulli variables with parameter $p$. We have 
	$$ \Pr[\sum_{i = 1}^n X_i \le pn - \alpha \sqrt{n}] \le 2^{-2\alpha^2}.$$
\end{lemma}

\subsection{Random linear codes}
\subsubsection{Basic properties}
For a vector $\xv = x_1,\dots,x_n \in \F_q^n$, we define the Hamming weight $|\xv| = |\{i : x_i \neq 0\}|$.  For $q,n,w \in \mathbb{N^*}$ with $q \ge 2$, we define the (Hamming) sphere of radius $w$ as $S_w^{q,n} = \{\xv \in \F_q^n : |\xv| = w\}$. A code $\C$ can be specified by a generating matrix $\Gm \in \F_q^{k \times n}$, in which case  $\C = \{\uv \Gm: \uv \in \F_q^k\}$ or via a parity-check matrix $\Hm \in \F_q^{n \times (n-k)}$, in which case $\C = \{\yv : \Hm \trsp{\yv} = \mathbf{0}\}$.

\anote{Si ces quantités suivantes ne sont pas utilisées ici, peut-être les mettre directement dans la section correspondante aux propriétés avancées sur les codes aléatoires.}
\begin{definition}[$q$-ary entropy]
	We define the $q$-ary entropy $h_q : [0,1] \rightarrow [0,1]$ s.t. $h(x) = -x\log\Rbra{\frac{x}{q-1}} - (1-x)\log(1-x)$ if $x \in (0,1)$ and $h_q(0) = h_q(1) = 0$. 
\end{definition}
$h_q$ is increasing for $x \in [0,\frac{q-1}{q}]$ and deceasing for $x \in [\frac{q-1}{q},1]$. Moreover, $h_q(\frac{q-1}{q}) = 1$.
\begin{definition}[Inverse $q$-ary entropy]
	$h_q$ is a bijection from $\Sbra{0,\frac{q-1}{q}}$ to $[0,1]$ and we define  $h_q^{-1} : [0,1] \rightarrow \Sbra{0,\frac{q-1}{q}}$ s.t. $h_q^{-1}(h_q(x)) = x$ for $x \in \Sbra{0,\frac{q-1}{q}}$.
\end{definition}

\begin{definition}[Relative Gilbert-Varshamov distance]
	The (relative) Gilbert-Varshamov distance for $q$-ary codes $\drgv(R,q)$ corresponding to the rate $R$ is defined as
	$\drgv(R,q) = h_q^{-1}(1-R)$.
\end{definition}

\begin{definition}[Relative maximum weight]
	The (relative) maximum weight for $q$-ary code $\drmax(R,q)$ is defined as the unique solution $x$ in $[\frac{q-1}{q},1]$ of
	$h_q(x)=R$ if it exists. If such an $x$ does not exist, we just write $\drmax(R,q) = \bot$.
\end{definition}

$\drgv(R,q)$ corresponds to the {\em} typical asymptotic relative minimum distance of a random linear code over $\Fq$ of rate $R$, whereas
the second quantity (when it is not $\bot$) is equal  to the {\em} typical asymptotic relative maximum distance.
Generally $q$ will be clear from the context and we will drop the dependency in $q$ and simply write 
$\drgv(R)$ and $\drmax(R)$.

\begin{definition}[Inverse of a full rank matrix]
	Let $k < n$ and $\Gm \in \F_q^{k \times n}$ be a matrix of full rank $k$. We define the pseudo-inverse $\Gm^{-1} \in \F_q^{n \times k}$ as a matrix satisfying
	$ \forall \uv \in \F_q^k, \ (\uv \Gm)\cdot \Gm^{-1} = \uv.$ \jp{Cette matrice n'est pas unique, j'ai changé "inverse" en "pseudo-inverse" et "the" par "a".}
\end{definition}

\begin{proposition}\label{Proposition:RandomRank}
	Let $m \ge k$ and let $G \Unif \Fq^{k \times m}$. We have 
	$$ \Pr[\rank(G) = k] \ge 1 - q^{k-m}.$$
\end{proposition}

\begin{proposition}\label{Proposition:Recovery}
	Let $\cv = \sv \Gm $ for some $\sv \in \F_q^k$ and $G \in \F_q^{k \times n}$. Let $J \subseteq [n]$ s.t. $G_J$ is of rank $k$. Then we have $\cv = \cv_J \Gm_J^{-1}  \Gm$.
\end{proposition}
\begin{proof}
	Notice that $\cv_J = \sv \Gm_J$. If $\Gm_J$ is of full rank $k$ then $\Gm_J^{-1}$ is well defined and $\cv_J \Gm_J^{-1} = \sv \Gm_J \Gm_J^{-1} = \sv$. From there, we conclude $\cv_J \Gm_J^{-1}  \Gm = \sv \Gm = \cv$.
\end{proof}

\subsubsection{Classical and quantum decoding problems}
Before defining our coding problem, we define the Bernoulli error distributions that we will use. 
\begin{definition}
	For $q\in \mathbb{N^*}$, with $q \ge 2$ and $\omega \in [0,1]$, we define the Bernoulli probability function $b_q : F_q \rightarrow \mathbb{R}$ satisfying $b_q(0) = 1-w$ and $b_q(i) = \frac{w}{q-1}$ for $i \neq 0$.
\end{definition}

\begin{definition}
	For $q\in \mathbb{N^*}$, with $q \ge 2$ and $\omega \in [0,1]$ we define the distribution $\B(q,\omega)$ sampled as follows: $\textrm{pick } x \textrm{ w.p. } b_q(x)$, return $x$.
\end{definition}

We now define the Bernoulli distribution on vectors on $F_q^n$ where each coordinate is taken according to $\B(q,\omega)$.

\begin{definition}
	For $q,n \in \mathbb{N^*}$, $\omega \in [0,1]$, with $q \ge 2$ we define the distribution $\B(q,n,\omega)$ sampled as follows: 
	for $i \in \{1,\dots,n\}$, $x_i \leftarrow \B(q,\omega)$, return $\xv = x_1,\dots,x_n$. Notice that sampling from $\B(q,n,\omega)$ is equivalent to the following sampling procedure: pick $\xv$ w.p. $\left( \frac{\omega}{q-1} \right)^{|\xv|}(1-\omega)^{n - |\xv|}$, return $\xv$.
\end{definition}

What we are interested here is the decoding problem as it arises in cryptography, but we will describe it here by using the langage of information theory. We have a message $\mv \in F_q^k$ which is encoded via a generating matrix $\Gm \in F_q^{k \times n}$. The encoded message $\mv\Gm$ is sent through a channel and an error $\ev$ occurs. The receiver gets the message $\mv\Gm + \ev$ and his goal is to recover $\mv$. Notice that the receiver also knows the generating matrix $\Gm$ so his goal is, given $\Gm$ and $\yv = \mv\Gm + \ev$, to recover $\mv$. 

The way we model the error is that $\ev$ is sampled from the Bernoulli distribution $\B(q,n,\omega)$ for some chosen $\omega$. Note that there are other choices for the error model that can be of interest that we discuss in the next section.  We first define the distribution of input/solution to our decoding problem.

\begin{definition}
	For $q,n,k \in \mathbb{N^*}$, with $q \ge 2$, for $\omega \in [0,1]$, we define the distribution \fbox{$\D(q,n,k,\omega)$} sampled as follows: $\Gm \Unif \zo^{k \times n}, \ \mv \Unif \F_q^k, \ \cv = \mv \Gm, \ \ev \leftarrow \B(q,n,\omega), \ \yv = \cv+ \ev$, return $(\Gm,\yv,\cv)$.
\end{definition}

We can now define our classical decoding problem

\begin{definition}
	For $q,n,k \in \mathbb{N^*}$, with $q \ge 2$, for $\omega \in [0,1]$, the decoding problem \fbox{$\DP(q,n,k,\omega)$} is the following. We sample $(\Gm,\yv,\cv) \leftarrow \D(q,n,k,\omega)$ and the goal is, given only $(\Gm,\yv)$, to recover $\cv$.
\end{definition}

Another problem of interest is finding short codewords. 
\begin{definition}For $q,n,k\in \N$ with $q \ge 2$ and $\omega \in (0,1)$, the short codeword problem \fbox{$\SCP(q,n,k,\omega)$} is the following. We sample $\Hm \Unif \F_q^{n \times (n-k)}$ and the goal is, given $\Hm$, to find $\cv \in \F_q^n \backslash {\{\mathbf{0}\}}$ st. $\Hm \trsp{\cv} = \mathbf{0}$ and $|\cv| \le \omega n$.
\end{definition} 

We now consider the quantum decoding problem. Now, instead of choosing a random error $\ev$ from $\B(q,n,\omega)$ and constructing $\yv = \cv + \ev$, we construct a quantum state that is a superposition of all these noisy codewords. This motivates the following definition for the input/solution distribution.

\begin{definition}
	For $q,n,k \in \mathbb{N^*}$, with $q \ge 2$ and $\omega \in [0,1]$, we define the distribution \fbox{$\DQ(q,n,k,\omega)$} sampled as follows: $  \Gm \Unif \zo^{k \times n}, \ \mv \Unif \F_q^k, \cv = \mv \Gm, \ \ket{\psi_{\cv}} =  \sum_{\ev \in F_q^n} \sqrt{\omega^{|\ev|}(1-\omega)^{n - |\ev|}} \ket{\cv + \ev}$, return $(\Gm,\ket{\psi_{\cv}},\cv)$.
\end{definition}

From there, we define our quantum decoding problem.

\begin{definition}
	For $q,n,k \in \mathbb{N^*}$, with $q \ge 2$, for $\omega \in [0,1]$, the decoding problem \fbox{$\DPQ(q,n,k,\omega)$} is the following. We sample $(\Gm,\ket{\psi_{\cv}},\cv) \leftarrow \DQ(q,n,k,\omega)$ and the goal is, given only $(\Gm,\ket{\psi_{\cv}})$, to recover $\cv$.
\end{definition}

The above definition can be generalized to any probability function $P : \F_q^n \rightarrow \mathbb{R}$ by considering the state $\ket{\psi_{\cv}} = \sum_{\ev \in F_q^n} \sqrt{P(\ev)}\ket{\cv + \ev}$. Moreover, and this is specific to the quantum setting, this can be generalized to any function $f : \F_q^n \rightarrow \mathbb{C}$ with $\norm{f}_2 = 1$ by considering the state $\ket{\psi_{\cv}} = \sum_{\ev \in F_q^n} f(\ev)\ket{\cv + \ev}$. This is what motivates the following definitions. 

\begin{definition}
	For $q,n,k \in \mathbb{N^*}$, with $q \ge 2$, for $f : \F_q^n \rightarrow \mathbb{C}$ with $\norm{f}_2 = 1$, we define the distribution \fbox{$\DQ(q,n,k,f)$} sampled as follows: $  \Gm \Unif \zo^{k \times n}, \ \mv \Unif \F_q^k, \cv = \mv \Gm,  \ \ket{\psi_{\cv}} =  \sum_{\ev \in F_q^n} f(\ev)\ket{\cv + \ev}$, return $(\Gm,\ket{\psi_{\cv}},\cv)$.
\end{definition}

\begin{definition}
	For $q,n,k \in \mathbb{N^*}$, with $q \ge 2$, for $f : \F_q^n \rightarrow \mathbb{C}$ with $\norm{f}_2 = 1$, the decoding problem \fbox{$\DPQ(q,n,k,f)$} is the following. We sample $(\Gm,\ket{\psi_{\cv}},\cv) \leftarrow \DQ(q,n,k,f)$ and the goal is, given only $(\Gm,\ket{\psi_{\cv}})$, to recover $\cv$.
\end{definition}

\COMMENT{\paragraph{Discussion on the above definitions.}
	There are a few things we want to say about these definitions. First, there are slight variations on the decoding problem. The most common is to take $\Gm \Unif \zo^{k \times n}, \ \mv \Unif \F_q^k, \ \ev \leftarrow \B(q,n,\omega), \ \yv = \mv \Gm + \ev$ and, given $(\Gm, \yv)$ to find a codeword $\cv$ such that $|\yv - \cv| \le \omega n$. The main difference is can find any codewords $\cv$ that satisfies $|\yv - \cv| \le \omega n$ and not necessarily the one that was used to construct $\yv$. This is a more standard definition in the context of code-based cryptography. The reason why we defined our decoding problem as in Definition~\ref{Definition:DP} is that this is the problem that we actually need to solve in Regev's reduction. From $\cv + \ev$ we need to recover exactly $\cv$ and not another codeword $\cv'$ close to $\cv + \ev$ otherwise the register erasure between the states $\ket{\phi_1}$ and $\ket{\phi_2}$ described in Section~\ref{Section;IntroRegev} cannot be done. 
	
	A second remark concerns the choice of noise function. The Bernoulli is the most common choice of noise model 
	This is actually mostly a technicality but it allowed the authors of~\cite{DRT23} to perform a clean reduction between the decoding problem and the small codeword problem 
	The reason why we go back to the Bernoulli noise.}

\subsubsection{Punctured codes and Prange's algorithm}
We will use in what follows the notion of punctured and shortened code. 
\begin{definition}[Punctured and shortened code]
	Let $\C \subseteq \Fqn$ be a linear code over $\Fq$ of length $n$. Let $J \subseteq [n]$ be a subset of code positions. The punctured code $\C_J$ with respect to  $J$ is defined as 
	$\C_J = \{\cv_J : \cv \in \C\}.$ The shortened code $C^J$ with respect to $J$ is defined as 
	$\C^J = \{\cv_J : \cv \in \C,\;\cv_{[n]\setminus J}=\zerom\}$ (\ie \ the set of codewords of $\C$ where we keep only the positions in $J$ and which are zero outside $J$).
\end{definition}
It is readily seen that these two operations commute when taking the dual
\begin{lemma}\label{lem:puncture_shorten} For any linear code $\C$ and any subset $J$ of positions of this code
	\begin{eqnarray*}
		\left( \C_J\right)^\perp & = & \left( \C^\perp\right)^J\\
		\left( \C^J\right)^\perp & = & \left( \C^\perp\right)_J.
	\end{eqnarray*}
\end{lemma}

\paragraph{A variation of the Prange algorithm.}
We recall here a result which is essentially folklore in coding theory, namely that there is a probabilistic polynomial time algorithm for finding short codewords in a random linear code of dimension $k$ and length $n$ over $\Fq$ which produces short codewords of weight $\left\lfloor\frac{(q-1)(n-k)}{q}\right\rfloor$. It simply uses linear algebra. For this, consider a parity-check matrix $\Hm \in \F_q^{(n-k)\times n}$ of $\C$ and run $\Th{\sqrt{n}}$ times the following procedure
\begin{enumerate}
	\item Choose uniformly at random subset $J$ of $k$ positions of $\C$. Let $\bar{J} = [n] \setminus J$.
	\item If $\Hm_{\bar{J}}$ is not of rank $n-k$, abort and else choose $\cv$ on $J$ as a random vector of Hamming weight $1$.
	\item Find the remaining entries of $\cv$ by solving the linear system
	$$
	\Hm_{\bar{J}} \trsp{\cv_{\bar{J}}} = - \Hm_J \trsp{\cv_J} 
	$$
	\item If $|\cv|=\left\lfloor\frac{(q-1)(n-k)}{q}\right\rfloor$ output $\cv$.
\end{enumerate}
The rationale behind this algorithm is that the expected weight of such a $\cv$ is $1+\frac{(q-1)(n-k)}{q}$ and that it can be proved that it takes the right weight with probability 
$\Om{\frac{1}{\sqrt{n}}}$. Note that all the known (be they classical or quantum) algorithms that produce asymptotically relative weights $\omega < \frac{(1-q)(1-R)}{q}$ where $R = \frac{k}{n}$ is the code rate have
exponential complexity.

\subsection{Distinguishing quantum states}

\begin{proposition}[Helstrom's measurement]\label{Proposition:Helstrom}
	Let $\ket{\psi_0},\ket{\psi_1}$ be $2$ quantum pure states s.t. $|\braket{\psi_0}{\psi_1}| = u$. There exists a  quantum projective measurement $\Pi = \{\Pi_0,\Pi_1\}$ s.t. $\forall c \in \zo, \ \tr(\Pi_c \kb{\psi_c}) = \frac{1}{2} + \frac{\sqrt{1 - u^2}}{2}$.
\end{proposition}
In the above measurement, the measurement gives the correct answer w.p. $\frac{1}{2} + \frac{\sqrt{1 - u^2}}{2}$ and gives the opposite answer w.p. $\frac{1}{2} - \frac{\sqrt{1 - u^2}}{2}$. Another measurement of interest is the one arising in the context of unambiguous state discrimination. Here we allow the measurement to answer ``I  don't know'' (which corresponds to outcome $2$). What we require from the measurement is that if the measurement does not answer $2$ then it always answers the correct value. The optimal unambiguous measurement is given by the proposition below. 
\begin{proposition}[Unambiguous State Discrimination]\label{Proposition:USD}
	Let $\ket{\psi_0},\ket{\psi_1}$ be $2$ quantum pure states s.t. $|\braket{\psi_0}{\psi_1}| = u$. There exists a POVM $F = \{F_0,F_1,F_2\}$ s.t. $\forall c \in \zo, \tr(F_c \kb{\psi_c}) = 1 - u$ and $\tr(F_2 \kb{\psi_c}) = u$. 
\end{proposition}
The optimal unambiguous measurement is not known when there are more than $2$ states. We present a detailed analysis of USD with $q$ states in Section~\ref{Section:qaryUSD}, where we give known results and also provide some new ones. \\

The final measurement of interest is the Pretty Good Measurement, which is a generic measurement to distinguish $n$ quantum states. 
\begin{definition}[Pretty Good Measurement]\label{Definition:PGM}
	Consider an ensemble $\{\ket{\psi_i}\}_{i \in [n]}$ of $n$ quantum pure states. The Pretty Good Measurement associated to this ensemble is the $POVM$ $\{M_i\}_{i \in [n]}$ with 
	$$ M_i = \rho^{-\frac{1}{2}} \kb{\psi_i} \rho^{-\frac{1}{2}}  \quad \textrm{ given } \quad \rho = \sum_{i \in [n]} \kb{\psi_i}.$$
	One can easily check that each $M_i \succcurlyeq 0$ and that $\sum_i M_i = \rho^{-\frac{1}{2}} \rho \rho^{-\frac{1}{2}} = I$.
\end{definition}
\begin{proposition}\cite{BK02,Mon06}\label{Proposition:PGMOptimality}
	Consider an ensemble $\{\ket{\psi_i}\}_{i \in [n]}$ of $n$ quantum pure states and $\{M_i\}_{i \in [n]}$ the associated pretty good measurement. We consider the setting where $i$ is chosen at random and we want to recover $i$ from $\ket{\psi_i}$. Let $\PPGM$ be the probability of success using the PGM and $\POPT
	$ be the optimal success probability. This means
	\begin{align*}
	\PPGM & = \frac{1}{n} \sum_i \tr\left(\kb{\psi_i}M_i\right) \\
	\POPT
	& = \max_{\{N_i\}} \frac{1}{n} \sum_i \tr(\kb{\psi_i}N_i)
	\end{align*}
	where the maximum is over all POVMs $\{N_i\}_{i \in [n]}$. We have 
	$$ \PPGM \le \POPT
	\le \sqrt{\PPGM}.$$
\end{proposition}

\subsection{The classical and quantum Fourier transform on $\Fqn$}
In this article, we will use the quantum Fourier transform on $\Fqn$ where $\Fq$ is the finite field $\F_q$. 
\paragraph{\bf Definition and basic properties.}
It is based on the characters of the group $(\Fqn,+)$ which are defined as follows (for more details see \cite[Chap 5, \S 1]{LN97}, in particular a description of the characters in terms of the trace function is given in  \cite[Ch. 5, \S 1, Th. 5.7]{LN97}).
\begin{definition}
	Fix $q = p^s$ for a prime integer $p$ and an integer $s \ge 1$. The characters of $\F_q$ are the functions $\chi_{y} : \F_q \rightarrow \mathbb{C}$ indexed by elements $y \in \F_q$ defined as follows
	\begin{eqnarray*}
		\chi_{y}(x) & \eqdef & e^{\frac{2i \pi \tr(x \cdot y)}{p}}, \quad \text{with} \\
		\tr(a) & \eqdef & a + a^p + a^{p^2} + \dots + a^{p^{s-1}}.
	\end{eqnarray*}
	where the product $x \cdot y$ corresponds to the product of elements in $\F_q$. We extend the definition to vectors $\xv,\yv \in \F_q^n$ as follows:
	$$ \chi_{\yv}(\xv) \eqdef \Pi_{i = 1}^n \chi_{y_i}(x_i).$$
\end{definition}
When $q$ is prime, we have $\chi_y(x) = e^{\frac{2i\pi xy}{q}}$. In the case where $q$ is not prime, the above definition is not necessarily easy to handle for computations. Fortunately, characters have many desirable properties that we can use for our calculations.
\begin{proposition}\label{Proposition;Characters}
	The characters $\chi_{\yv} : \Fqn \rightarrow \mathbb{C}$ have the following properties
	\begin{enumerate}\setlength\itemsep{-0.2em}
		\item (Group Homomorphism). $\forall \yv \in \Fqn$, $\chi_{\yv}$ is a group homomorphism from $(\Fqn,+)$ to $(\mathbb{C},\cdot)$ meaning that $\forall \xv, \xv' \in \Fqn$, $\chi_{\yv}(\xv + \xv') = \chi_{\yv}(\xv) \cdot \chi_{\yv}(\xv')$.
		\item \label{eq:symmetry} (Symmetry). $\forall \xv, \yv \in \Fqn, \ \chi_{\yv}(\xv) = \chi_{\xv}(\yv)$
		\item \label{eq:orthogonality} (Orthogonality of characters). The characters are orthogonal functions meaning that $\forall \xv, \xv' \in \Fqn$, $\sum_{\yv \in \Fqn} \chi_{\yv}(\xv)\overline{\chi_{\yv}(\xv')} = q^n \delta_{\xv,\xv'}$. In particular $\sum_{\yv \in \Fqn} |\chi_{\yv}(\xv)|^2 = q$ and $\forall \xv \in \Fqn\setminus\{0\}, \ \sum_{\yv \in \Fqn} \chi_{\yv}(\xv) = 0$.
											\end{enumerate}
\end{proposition}

Notice that these imply some other properties on characters. For instance $\chi_\yv(\mathbf{0}) = 1$ or $|\chi_\yv({\xv})| = 1$ for any $\xv,\yv \in \F_q^n$.
The orthogonality of characters, allows to define a unitary transform which is is nothing but  the classical or the quantum Fourier transform on $\Fqn$.

\begin{definition}
	For a function $ f : \Fqn \rightarrow \C$, we define the (classical) Fourier transform $\hat{f}$ as
	$$
	\FT{f}(\xv) = \frac{1}{\sqrt{q^n}} \sum_{\yv \in \Fqn} \chi_{\xv}(\yv) f(\yv).
	$$
	The quantum Fourier transform $\QFTt$ on $\Fqn$ is the quantum unitary satisfying $\forall \xv \in \Fqn$, 
	\begin{eqnarray*}
		\QFTt \ket{\xv} &= &\frac{1}{\sqrt{q^n}} \sum_{\yv \in \Fqn} \chi_{\xv}(\yv) \ket{\yv}.
	\end{eqnarray*}
	We will also write $\ket{\QFT{\psi}} \eqdef \QFTt \ket{\psi}$.
\end{definition}
Note that when $\ket{\psi} = \sum_{\xv \in \Fqn} f(\xv) \ket{\xv}$ we have
\begin{equation*}
\ket{\QFT{\psi}} = \sum_{\xv \in \Fqn} \FT{f}(\xv) \ket{\xv}.
\end{equation*}
The Fourier transform can also be viewed as expressing the coefficients of a state  in the Fourier basis 
$\left\{\ket{\QFT{\xv}},\xv \in \Fqn \right\}$ as shown by

\begin{fact}\label{fact:inverse}
	Let $\ket{\psi}= \sum_{\yv \in \Fqn} f(\yv) \ket{\yv}$, then
	$$
	\ket{\psi} = \sum_{\xv \in \Fqn} \FT{f}(-\xv) \ket{\QFT{\xv}}.
	$$
\end{fact}
This follows on the spot from the fact that if $\ket{\psi} = \sum_{\xv \in \Fqn} c_{\xv} \ket{\QFT{\xv}}$, then 
$$
c_{\xv} = \braket{\QFT{\xv}}{\psi}= \frac{1}{\sqrt{q^n}} \sum_{\yv \in \Fqn} \overline{\chi_{\xv}(y)} f(y) = \FT{f}(-\xv).
$$

\paragraph{\bf Translations amount to multiplication by a phase in the Fourier basis.}
It will be convenient for what follows to bring in the shift and phase operators which are defined by
\begin{definition}[shift and phase operators]
	For $\bv$ in $\Fqn$, let $X_{\bv}$ be the shift operator $X_{\bv}\ket{\xv} = \ket{\xv + \bv}$ and $Z_{\bv}$ be the phase operator 
	$Z_{\bv} = \chi_{\xv}(\bv) \ket{\xv}$.
\end{definition}

The main properties of the Fourier transform follow from the fact that the characters are the common eigenbasis of all shift operators (and therefore all convolution operators). In the quantum setting, 
this amounts to the  fact that the quantum states $\{\ket{\QFT{\xv}},\xv \in \Fqn\}$ form an eigenbasis of the shift operators as shown by
\begin{proposition}\label{proposition:shift}
	We have for all $\bv$ in $\Fqn$ that 
	$ \ket{\QFT{\xv}}$  is  an eigenstate of $X_{\bv}$ associated to the eigenvalue  $\chi_{\xv}(-\bv)$ and
	\begin{eqnarray}
	X_{\bv} \cdot \QFTt & = & \QFTt \cdot Z_{-\bv} \label{eq:eigenstate}\\
	\QFTt \cdot X_{\bv} &=& Z_{\bv} \cdot \QFTt. \label{eq:dual_eigenstate}
	\end{eqnarray}
	
\end{proposition}
\begin{proof}
	Let $\xv \in \Fqn$. We observe that
	\begin{eqnarray*}
		X_{\bv} \cdot \QFTt    \ket{\xv} &  = & \frac{1}{\sqrt{q^n}} \sum_{\yv \in \Fqn} \chi_{\xv}(\yv) \ket{\yv+\bv} \\
		& = & \frac{1}{\sqrt{q^n}} \sum_{\yv \in \Fqn} \chi_{\xv}(\yv-\bv) \ket{\yv}\\
		& = & \chi_{\xv}(-\bv) \frac{1}{\sqrt{q^n}} \sum_{\yv \in \Fqn} \chi_{\xv}(\yv)  \ket{\yv}\\
		& = &  \QFTt \cdot Z_{-\bv} \ket{\xv}.
	\end{eqnarray*}
	This computation shows that $\ket{\QFT{\xv}}=\QFTt    \ket{\xv}$ is an eigenstate of the shift operator $X_{\bv}$
	associated to the eigenvalue $\chi_{\xv}(-\bv)$. The other equality follows from this and the 
	symmetry property \ref{eq:symmetry} of Proposition \ref{Proposition;Characters} which implies that
	\begin{equation}
	\QFTt^\dagger =\overline{\QFTt}  
	\end{equation}
	where by $\overline{M}$ we mean the (complex) conjugate operator of the operator $M$ which is defined by 
	$\overline{M} \eqdef \sum_{x,y} \overline{M_{x,y}} \ket{x}\bra{y}$ when
	$M = \sum_{x,y} M_{x,y} \ket{x}\bra{y}$. 
	\eqref{eq:eigenstate} namely implies that
	$$
	\QFTt^\dagger \cdot X_{\bv}^{\dagger} = Z_{-\bv}^{\dagger} \cdot \QFTt^\dagger
	$$
	This in turn means that 
	$$
	\overline{\QFTt}\cdot X_{-\bv} = Z_{\bv} \cdot \overline{\QFTt},
	$$
	or equivalently
	$$
	\overline{\overline{\QFTt}\cdot X_{-\bv}} = \overline{Z_{\bv} \cdot \overline{\QFTt}}
	$$
	which gives
	$$
	\QFTt\cdot X_{-\bv} = Z_{-\bv} \cdot \QFTt,
	$$
	and therefore proving \eqref{eq:dual_eigenstate}.
									\end{proof}
We will focus on the following quantum states $\ket{\psi} = \sqrt{1 - \tau}\ket{0} + \sum_{\alpha \in \F_q^*} \sqrt{\frac{\tau}{q-1}}\ket{\alpha}$ associated to a $q$-ary channel of crossover probability $\tau$. Indeed, when we measure such a quantum state, we namely get 
an element of $\Fq$ which can be viewed as a sample of an error output by such a channel. The quantum Fourier transform applied to such states yields a state of the same form, since it is readily verified that
\begin{lemma}\label{lemma:QFTpsi}
	Let $\tau \in [0,\frac{q-1}{q}]$ and $\ket{\psi} = \sqrt{1 - \tau}\ket{0} + \sum_{\alpha \in \F_q^*} \sqrt{\frac{\tau}{q-1}}\ket{\alpha}$. We have 
	$$ \QFTt \ket{\psi} = \sqrt{1 - \tau^\perp}\ket{0} + \sum_{\alpha \in \F_q^*} \sqrt{\frac{\tau^\perp}{q-1}}\ket{\alpha}$$
	with 
	$ \tau^\perp = \frac{\Rbra{\sqrt{(q-1)(1-\tau)} - \sqrt{\tau}}^{2}}{q}.$
\end{lemma}
\begin{proof}
	We write 
	\begin{align*}
	\QFTt \ket{\psi} & = \sqrt{\frac{1-\tau}{q}} \sum_{y \in \F_q} \ket{y} + \sqrt{\frac{\tau}{q(q-1)}} \sum_{y \in \F_q} \sum_{\alpha \in \F_q^*} \chi_\alpha(y)\ket{y} \\
	& = \left(\sqrt{\frac{1-\tau}{q}} + \sqrt{\frac{q \tau}{q-1}}\right)\ket{0} + \sum_{y \in \F_q^*} \left(\sqrt{\frac{1-\tau}{q}} - \sqrt{\frac{\tau}{q(q-1)}}\right) \ket{y}
	\end{align*}
	where in the last equality we used the fact that for $y \neq 0$, we have $\sum_{\alpha \in \Fq} \chi_\alpha(y)=\sum_{\alpha \in \Fq} \chi_y(\alpha) = 0$ (by using first the symmetry property and then the orthogonality property of characters of Proposition \ref{Proposition;Characters}). This 
	implies that 
	$\sum_{\alpha \in \F_q^*} \chi_\alpha(y)= - \chi_0(y)=-1$.
		In order to conclude, notice that 
	\begin{align*}
	\sqrt{\frac{\tau^\perp}{q-1}} = \frac{\sqrt{(q-1)(1-\tau)} - \sqrt{\tau}}{\sqrt{q(q-1)}} = \sqrt{\frac{1-\tau}{q}} - \sqrt{\frac{\tau}{q(q-1)}}
	\end{align*}
	which means we can rewrite $\QFTt \ket{\psi} = \sqrt{1-\tau^\perp}\ket{0} + \sum_{y \in \F_q^*} \sqrt{\frac{\tau^\perp}{q-1}}$.
\end{proof}
We will also need to describe how the quantum Fourier transform acts on shifts of $\ket{\psi}$
\begin{lemma}\label{lemma:QFTpsib}
	Let $\tau \in [0,\frac{q-1}{q}]$, $b \in \F_q$ and denote by $\ket{\psi_b}$ the state $X_b \ket{\psi}$ where $\ket{\psi} \eqdef \sqrt{1 - \tau}\ket{0} + \sum_{\alpha \in \F_q^*} \sqrt{\frac{\tau}{q-1}}\ket{\alpha}$.
	We have
	\begin{eqnarray*}
		\ket{\psi_b} &=& \sqrt{1 - \tau}\ket{b} + \sum_{\alpha \neq b} \sqrt{\frac{\tau}{q-1}}\ket{\alpha}\\
		\QFTt \ket{\psi_b} &= &\sqrt{1 - \tau^\perp}\ket{0} + \sum_{\alpha \in \F_q^*} \chi_\alpha(b) \sqrt{\frac{\tau^\perp}{q-1}}\ket{\alpha}.
	\end{eqnarray*}
\end{lemma}
\begin{proof}
	The first point follows right away from the definition of these quantities, whereas the second point follows on the spot from 
	Fact \ref{proposition:shift} and the previous lemma:
	\begin{eqnarray*}
		\QFTt \ket{\psi_b} &= & \QFTt \cdot X_b \ket{\psi} \\
		& = & Z_b \cdot \QFTt \ket{\psi} \;\;\text{ (by Fact \ref{proposition:shift})}\\
		& =& Z_b  \left( \sqrt{1 - \tau^\perp}\ket{0} + \sum_{\alpha \in \F_q^*}  \sqrt{\frac{\tau^\perp}{q-1}}\ket{\alpha}\right) \;\; \text{(by Lemma \ref{lemma:QFTpsi})}\\
		& = & \sqrt{1 - \tau^\perp}\ket{0} + \sum_{\alpha \in \F_q^*} \chi_\alpha(b) \sqrt{\frac{\tau^\perp}{q-1}}\ket{\alpha}.
	\end{eqnarray*}
\end{proof}

\paragraph{\bf Applying the quantum Fourier transform on periodic states.}
Regev's reduction applies to states which are periodic. In our case, they will be of the form $\frac{1}{\sqrt{Z}} \sum_{\cv \in \C}\sum_{\ev \in \F_q^n} f(\ev) \ket{\cv + \ev}$ where 
$Z$ is some normalizing constant, $\C$ some linear code of length $n$ over $\Fq$ and $f$ some function from $\F_q^n$ to $\Comp$. This state can be written as $\frac{1}{\sqrt{Z}} \sum_{\xv \in \F_q^n} g(\xv) \ket{\xv}$ where $g(\xv) = \sum_{\cv \in \C} f(\xv-\cv)$. We clearly have in this 
case $g(\xv +\cv)=g(\xv)$ for any $\xv \in \F_q^n$ and any $\cv \in \C$. For such states, we have the following

\begin{proposition}\label{proposition:periodic}
Consider a function $f : \F_q^n \mapsto \Comp$.  We have for all linear codes $\C \subseteq \F_q^n$:
$$
\QFTt \left( \frac{1}{\sqrt{Z}} \sum_{\cv \in \C}\sum_{\ev \in \F_q^n} f(\ev) \ket{\cv + \ev} \right) = \frac{|\C|}{\sqrt{Z}} \sum_{\yv \in \C^\perp} \QFT{f}(\yv) \ket{\yv}
$$
where $Z$ is some normalizing constant.
\end{proposition}
\begin{proof}
The proposition follows from the following computation
\begin{eqnarray}
\QFTt \left( \frac{1}{\sqrt{Z}} \sum_{\cv \in \C}\sum_{\ev \in \F_q^n} f(\ev) \ket{\cv + \ev} \right) &=& \frac{1}{\sqrt{Z}} \sum_{\ev \in \F_q^n} f(\ev)   \sum_{\cv \in \C} \sum_{\xv \in \F_q^n} \chi_\yv(\cv+\ev) \ket{\yv} \nonumber \\
& = &  \frac{1}{\sqrt{Z}} \sum_{\ev \in \F_q^n} f(\ev)    \sum_{\yv \in \F_q^n}  \chi_\yv(\ev) \ket{\yv}\sum_{\cv \in \C} \chi_{\yv}(\cv) \nonumber \\
& = &  \frac{|\C|}{\sqrt{Z}} \sum_{\ev \in \F_q^n}\chi_\yv(\ev)  f(\ev) \sum_{\yv \in \C^\perp}  \ket{\yv} \label{eq:periodic}\\
& =& \frac{|\C|}{\sqrt{Z}} \sum_{\yv \in \C^\perp} \QFT{f}(\yv) \ket{\yv} \nonumber
\end{eqnarray}
where \eqref{eq:periodic} follows from a slight generalization of \eqref{eq:orthogonality} of Proposition \ref{Proposition;Characters}, namely that
\begin{eqnarray*}
\sum_{\cv \in \C} \chi_{\yv}(\cv) & = & |C| \;\;\text{if $\yv \in \C^\perp$}\\
& = & 0 \;\; \text{otherwise,}
\end{eqnarray*}
which follows by a similar reasoning by noticing that $\C^\perp$ can be vieved as the set of trivial characters acting on $\C$:
$$\{ \yv \in \F_q^n: \chi_{\yv}(\cv)=1,\;\forall \cv \in \C\}= \C^\perp.$$
\end{proof}
 \section{Algorithms for the binary quantum decoding problem}
\subsection{Quantum polynomial time algorithm using unambiguous state discrimination}\label{Section:BinaryPolynomial}
We present our first quantum algorithm that directly uses unambiguous state discrimination.
\begin{theorem}
	Let $R \in (0,1)$. For any $\omega < \perpF{\frac{R}{2}} \eqdef \frac{1}{2} - \sqrt{\frac{R}{2}(1-\frac{R}{2})}$, there exists a quantum algorithm that solves $\DPQ(2,n,\lfloor Rn \rfloor,\omega)$ w.p. $1 - 2^{-\Omega(n)}$.
\end{theorem}
\begin{proof}
	We fix $R,\omega$, as well as $n \in \mathbb{N}$ and $k = \lfloor Rn \rfloor$. We consider an instance of $QDP(2,R,k,\omega)$ so we have a random matrix $\Gm \Unif \zo^{k \times n}$, $\cv = \mv \Gm$ for a randomly chosen $\mv \Unif \zo^k$ and the state 
	$ \ket{\Psi_{\cv}} = \bigotimes_{i = 1}^n \ket{\psi^\omega_{c_i}}$ where $\ket{\psi^\omega_{c_i}} = \sqrt{1-\omega} \ket{c_i} + \sqrt{\omega} \ket{1-c_i}$.
	We consider the following algorithm for solving our Quantum Decoding Problem \\
	
	\noindent 	\cadre{\begin{center} Quantum algorithm for $\DPQ$ using USD \end{center}
		\begin{enumerate}
			\item Start from $\ket{\Psi_\cv} = \bigotimes_{i = 1}^n \ket{\psi_{c_i}^\omega}$. Notice that $|\braket{\psi^\omega_{0}}{\psi^\omega_{1}}| = 2\sqrt{\omega(1-\omega)}$.
			\item Perform the optimal unambiguous measurement from Proposition~\ref{Proposition:USD} on each qubit of $\ket{\Psi_\cv}$ in order to guess $c_i$, which can be done w.p. $p = 1 - |\braket{\psi^\omega_{0}}{\psi^\omega_{1}}|  = 1 - 2\sqrt{\omega(1-\omega)} \eqdef 2\omega^\perp$. Let $J \subseteq [n]$ be the set of indices where this measurement succeeds. The algorithm recovers here $\cv_J$. 
			\item If $G_J \in \zo^{k \times |J|}$ is of rank $k$, recover $\cv$ from $\cv_J$ by computing $\cv_J G_J^{-1}  G$.
		\end{enumerate}
	}	$ \ $ \\
Let $p = 2\omega^\perp$. Since $\omega < \perpF{\frac{R}{2}}$, we have $2\omega^\perp > R$ and there exists an absolute constant $\gamma > 0$ s.t. 
$p = R + \gamma$. Let $X_i$ be the random variable s.t. $X_i(i \in J) = 1$ and $X_i(i \notin J) = 0$. The $X_i$ are independent random Bernoulli variables with parameter $p$. Using Hoeffding's inequality, we first compute 
\begin{align*}
P_1 = \Pr\left[|J| \ge k + \frac{\gamma n}{2} \right] \ge \Pr\left[\sum_{i = 1}^n X_i \ge pn - \frac{\gamma n}{2} \right] \ge 1 - 2^{-\frac{\gamma^2 n}{2}}
\end{align*}
Then, using Proposition~\ref{Proposition:RandomRank} we compute 
\begin{align*}
P_2 = \Pr\left[\rank(G_J) = k \left| \ |J| \ge k + \frac{\gamma n}{2} \right.\right] \ge 1 - 2^{-\gamma n}.
\end{align*}
Notice that the algorithm recovers $\cv_J$ so from Proposition~\ref{Proposition:Recovery}, if $\rank(G_J) = k$ then the algorithm successfully recovers $\cv$. If we define $P_{Succ}$ to be the probability of success of the algorithm, we therefore have
$$
P_{Succ} \ge \Pr[\rank(G_J) = k] \ge P_1 P_2 \ge 1 - 2^{-\Omega(n)}.$$

\end{proof}

\paragraph{Using complex phases.} It is also possible to put complex phases in front of the error. This means we consider the states 
$$\ket{\Psi_\cv} = \bigotimes_{i = 1}^n \sqrt{1 - \omega}\ket{c_i} + \sqrt{\omega}e^{i \theta} \ket{1 - c_i}.$$
Interesting phenomena appear and we refer to Appendix~\ref{Appendix:A} for a full analysis. 

 \subsection{Reduction between quantum decoding problems in the binary setting}\label{Section:PartialUSD}
The above algorithm is interesting as it presents an polynomial time algorithm for the quantum decoding problem in a regime where its classical counterpart requires - with our current knowledge - an exponential classical or quantum algorithm. However, it completely fails when $\omega > \perpF{\frac{R}{2}}$ and the best algorithm for $\DPQ(2,n,\lfloor Rn \rfloor,\omega)$ is still by first measuring and then solving $\DP(2,n,\lfloor Rn \rfloor,\omega)$. Is there a way to improve the best algorithms $\DPQ(2,n,\lfloor Rn \rfloor,\omega)$ by using ideas of the previous section? The answer is yes. Instead of using USD, we use what we call partial Unambiguous State Discrimination. Our measurement will still abort with some probability but when it does not abort, we still allow a small probability failure but which will typically be smaller than if we used Helstrom's measurement. With this technique we can actually show a general reduction theorem for $\DPQ$.

\begin{theorem}~\label{Theorem:Partial}
	Let $R \in (0,1)$. Let $\omega \in [0,\frac{1}{2})$ and $\omega' \in [0,\frac{1}{2})$ satisfying: $\omega' \le \omega$ and $\frac{\perpO{\omega}}{\perpF{\omega'}}  > R$.
	Let any $p > \frac{\perpO{\omega}}{\perpF{\omega'}}$. Let also $k = \lfloor Rn \rfloor$.
	Then $\DPQ(2,n,k,\omega) \preccurlyeq \DPQ(2,\lfloor pn\rfloor,k,\omega')$ meaning that if we have an algorithm that solves $\DPQ(2,\lfloor pn\rfloor,k,\omega')$, we can use it to solve $\DPQ(2,n,k,\omega)$.
\end{theorem}

In order to prove our theorem, we first present our partial unambiguous state discrimination protocol. As a special case, we obtain our previous algorithm by taking $\omega' = 0$ (the theorem can then be applied when $R < 2\perpO{\omega}$).

\subsubsection{Partial unambiguous state discrimination}
We define $\ket{\psi_b^\omega} = \sqrt{1-\omega}\ket{b} + \sqrt{\omega}\ket{1-b}$. Recall that $\braket{\psi_0^\omega}{\psi_1^\omega} = 2\sqrt{\omega(1-\omega)} = 1 - 2 \perpO{\omega}$.  Fix $\omega,\omega' \in (0,\frac{1}{2})$ with $\omega' \le \omega$. We use the following lemma
\begin{restatable}{lemma}{lemmaUnitaryUSD}
~\label{Lemma:UnitaryUSD}
	Let $\alpha = \sqrt{\frac{\perpO{\omega}}{\perpF{\omega'}}}$ and $\beta = \sqrt{1 - \alpha^2}$. 
	There exists a unitary $U$ operation acting on $\myspan\{\ket{0},\ket{1},\ket{2}\}$ s.t. 
	\begin{align*}
	U \ket{\psi_0^\omega} & = \alpha \ket{\psi_0^{\omega'}} + \beta \ket{2} \\
	U \ket{\psi_1^\omega} & = \alpha \ket{\psi_1^{\omega'}} + \beta \ket{2}
	\end{align*}
\end{restatable}
\begin{proof}
\jp{J'enleve le proof sketch, la preuve en appendice est trop longue et n'est pas plus convaincante que le calcul ici.}
With the choice of $\alpha$ that was made the hermitian product $\braket{\psi_0^\omega}{\psi_1^\omega}$ and their image is preserved. As a matter of fact
\begin{equation}\label{eq:firsthermitproduct}
\braket{\psi_0^\omega}{\psi_1^\omega}  =  2 \sqrt{\omega(1 - \omega)} = 1 - 2\perpO{\omega}.
\end{equation}
Now, if we let 
$\ket{\psi'_b} \eqdef \alpha \ket{\psi_b^{\omega'}} + \beta \ket{2}$ for $b \in\{0,1\}$, then we have
	\begin{eqnarray*} 
	\braket{\psi'_0}{\psi'_1} 
& = & |\alpha|^2\braket{\psi^{\omega'}_0}{\psi^{\omega'}_1} + |\beta|^2 \\
& = & |\alpha|^2(1 - 2 \perpF{\omega'}) + |\beta|^2  \quad \text{( by \eqref{eq:firsthermitproduct})}\\
& =& 1 - 2 |\alpha|^2\perpF{ \omega'} \quad \text{(by using $|\beta|^2=1-|\alpha|^2$)}\\
& =& 1 - 2 \perpO{\omega} \quad \quad \text{(with our choice of } \alpha).
	\end{eqnarray*}
	By definition of $\beta$, $\ket{\psi'_0}$ and $\ket{\psi'_1}$ are both of norm $1$. This together with the equality $\braket{\psi_0^\omega}{\psi_1^\omega}=\braket{\psi'_0}{\psi'_1} $ we just proved shows that $U$ as defined above preserves the hermitian product on $\myspan\{\psi_0^\omega,\psi_1^\omega\}=\myspan\{\ket{0},\ket{1}\}$. It suffices to choose
	$U \ket{2}$ of norm $1$ and orthogonal to both $\ket{\psi'_0}$ and $\ket{\psi'_1}$ to obtain a unitary transform since by construction it preserves the hermitian product on 
	$\myspan\{\ket{0},\ket{1},\ket{2}\}$.
\end{proof}

\begin{proposition}\label{Proposition:PartialUSD}
	Let $\omega,\omega' \in (0,\frac{1}{2})$ with $\omega' < \omega$. There exists a quantum measurement s.t. when it is applied on $\ket{\psi_b^\omega}$, the resulting state is $\ket{\psi_b^{\omega'}}$ w.p. $\frac{\perpO{\omega}}{\perpF{\omega'}}$ and $\ket{2}$ w.p. $1 - \frac{\perpO{\omega}}{\perpF{\omega'}}$.
\end{proposition}
\begin{proof}
	Start from $\ket{\psi_b^\omega}$ and apply the unitary $U$ from Lemma~\ref{Lemma:UnitaryUSD}. Then, perform the two outcomes projective measurement $\left\{\left(\kb{0} + \kb{1}\right),\kb{2}\right\}$ on the state $U\ket{\psi_b^\omega} = \alpha \ket{\psi_b^{\omega'}} + \beta \ket{2}$. We obtain the first outcome w.p. $|\alpha|^2 = \frac{\perpO{\omega}}{\perpF{\omega'}}$ and the resulting state is $\ket{\psi_b^{\omega'}}$ and the second outcome w.p. $|\beta|^2$ and the resulting outcome is $\ket{2}$.
\end{proof}
Unambiguous state discrimination can be seen as a special case of this operation by taking $\omega' = 0$, which gives $\alpha = \sqrt{2\perpO{\omega}}$ and the probability of success is $\alpha^2 = 2\perpO{\omega} = 1 - \braket{\psi_0^\omega}{\psi_1^\omega}$. 

\subsubsection{Proof of Theorem~\ref{Theorem:Partial}}
In order to prove Theorem~\ref{Theorem:Partial}, one can just apply the algorithm of Section~\ref{Section:BinaryPolynomial} in a similar fashion. We take any  $\omega,\omega' \in (0,\frac{1}{2})$ with $\omega' \le \omega$ and $\frac{\perpO{\omega}}{\perpF{\omega'}}  > R$. We also fix $p > \frac{\perpO{\omega}}{\perpF{\omega'}}$ 

We want to solve $\DPQ(2,n,k,\omega)$ using an algorithm that solves $\DPQ(2,\lfloor pn \rfloor ,k,\omega')$. We start from $\Gm \in \F_q^{k \times n}$ as well as $\ket{\psi_\cv} = \bigotimes_{i = 1}^n \ket{\psi_{c_i}^\omega}$. We consider the following algorithm \\ \\
	\noindent 	\cadre{\begin{center} Quantum algorithm for $\DPQ$ using partial USD \end{center}
	\begin{enumerate}
		\item Perform the quantum measurement of Proposition~\ref{Proposition:PartialUSD} on each register of $\ket{\psi_\cv}$.  Let $J \subseteq [n]$ be the set of indices where this measurement succeeds {\ie } where we obtain $\ket{\psi_{c_i}^{\omega'}}$. By discarding the indices not in $J$, we obtain 
		$$ \ket{\phi_{\cv_J}} = \bigotimes_{i \in J} \ket{\psi_{c_i}^{\omega'}}.$$
		\item Notice that $\cv_J \in \C_J$ and recovering $\cv_J$ from $\ket{\phi_{\cv_J}}$ is a quantum decoding problem on $\C_J$, more precisely an instance of $\DPQ(2,|J|,k,\omega')$. As long as $|J| \ge \lfloor p n \rfloor$, we use our $\DPQ(q,n,\lfloor p n \rfloor,\omega')$ (by potentially removing excess coordinates if necessary if $|J| >\lfloor p n \rfloor$) to recover $\cv_J$.
	 \item We recover $\cv$ from $\cv_J$ by computing $\cv_J \Gm_J^{-1} \Gm$.
	\end{enumerate}
}	$ \ $ \\

By definition, we recover $\cv_J$. We just have to bound the probability to recover $\cv$. Notice that in Step $1$, we have from Proposition~\ref{Proposition:PartialUSD} that the measurement will succeed w.p. $\frac{\perpO{\omega}}{\perpF{\omega'}} > p > R$ for each index. As in Section~\ref{Section:BinaryPolynomial}, this implies that with overwhelming probability, $|J| \ge \lfloor p n \rfloor$ which in turn implies that we can recover $\cv$ from $\cv_J$ with overwhelming probability. 

\subsubsection{Interpretation of the above as changing the noise model}
In this section, we show how performing (partial) unambiguous state discrimination on a state $\ket{\psi_b} = \sqrt{1 - \omega}\ket{b} + \sqrt{\omega}\ket{1-b}$ can be seen as a way to change the noise model applied on the bit $b$. We first define different notions of noisy channels in the binary setting. 

\begin{definition}
	For a bit $b$, an error probability $\omega$ and abort probability $p$, we define the distributions of the Binary Symmetric Channel $\BSC(b,\omega)$, of the Binary Erasure Channel $BEC(b,p)$ and of the Binary Symmetric with Errors and Erasures Channel $\BSEEC(b,\omega,p)$ sampled as follows:
	\begin{align*}
	BSC(b,\omega)& : \return b \wp (1-\omega) \mbox{ and } (1-b) \wp \omega. \\
	BEC(b,p)& : \return b \wp (1-p) \mbox{ and } \bot \wp p. \\
	BSEEC(b,\omega,p)& : \return b \wp (1-p)(1-\omega), (1-b) \wp (1-p)\omega \mbox{ and } \bot \wp p. \\
	\end{align*}
\end{definition}
For a bit $b$, flipping it w.p. $\omega$ can be seen as passing $b$ through a binary symmetric channel $BSC(\omega)$. Having this error in superposition means that we have access to the quantum state.
Our results can be interpreted as follows
\begin{proposition}
	From $\ket{\psi_b} = \sqrt{1 - \omega}\ket{b} + \sqrt{\omega} \ket{1-b}$ it is possible to:
	\begin{enumerate}
		\item Generate $y \Unif BSC(b,\omega)$ simply by measuring $\ket{\psi_b}$.
		\item Generate $y \Unif BEC(b,1 - 2\perpO{\omega})$ by performing unambiguous state discrimination on $\ket{\psi_b}$.
		\item Generate $y \Unif BSEEC(b,(\frac{\perpO{\omega}}{1-p})^{\bot},p)$ for any abort probability $p \in [0,1 - 2\perpO{\omega}]$, by performing partial unambiguous state discrimination on $\ket{\psi_b}$.
	\end{enumerate}
\end{proposition}
Notice that the third case generalizes the $2$ first cases by respectively taking $p = 0$ and $p = 1 - 2\perpO{\omega}$. This shows the advantage of having the noise in quantum superposition. It is possible to change the noise from the one coming from a Binary Symmetric Channel to the one coming from a Binary Erasure Channel or a Binary Symmetric with Errors and Erasures Channel. \section{Polynomial time algorithm for $\DPQ$ in the $q$-ary setting}\label{Section:Polynomial}

As we saw in the previous section, unambiguous state discrimination is crucial for polynomial time algorithm for $\DPQ$. While this task is very well understood in the binary case, we do not have any general formula in the $q$-ary setting. Fortunately, the states we consider will have enough structure so that we can fully characterize the optimal unambiguous state discrimination algorithm. We first present this characterization, which is essentially a generalization of the work of~\cite{CB98}. We then use this unambiguous state discrimination in the $q$-ary setting to derive our quantum algorithm for $\DPQ$ in the $q$-ary setting, in the same spirit as what we did in Section~\ref{Section:BinaryPolynomial}.
\subsection{Unambiguous state discrimination in the $q$-ary setting}\label{Section:qaryUSD}

\begin{definition}
	An unambiguous state discrimination measurement associated to some states $\ket{\psi_0},\dots,\ket{\psi_{N-1}}$ is a POVM $\{E_0,\dots,E_{N-1},E_F\}$ (where $E_F$ stands for the failure outcome) s.t. 
	$$ \forall i,j \neq i \in \Iint{0}{N-1}, \ \tr(E_i \kb{\psi_j}) = 0.$$
	To such a POVM, we associate the quantities 
	$P_j \eqdef  \tr(E_j \kb{\psi_j})$ (the probability of correctly guessing $j$ when given $\ket{\psi_j}$, as well as the average success probability $\overline{P_D} \eqdef \frac{1}{N} \sum_{j=0}^{N-1} P_j$.
\end{definition}
The optimal unambiguous measurement is not known when there are more than $2$ states, however it is known in a case where the states we want to distinguish are linearly independent, have the same
{\em a priori} probabilities and are symmetric in the following sense \cite{BKMH97}
\begin{definition}[symmetric states]
	A set $\Cbra{\ket{\psi_0},\cdots,\ket{\psi_{N-1}}}$ in a Hilbert space $\hh$ of dimension $N$ is symmetric if and only if there exists a unitary transformation $U$ of order $N$ on $\hh$ such that for any $i$ and $j$ 
	in $\Iint{0}{N-1}$ we have $\ket{\psi_j} = U^{j-i} \ket{\psi_i}$.
\end{definition}

In such a case, the optimal unambiguous measurement is known \cite{CB98}
\begin{proposition}[Unambiguous State Discrimination of Symmetric States]\label{proposition:USSD}
	Let $\Cbra{\ket{\psi_0},\cdots,\ket{\psi_{N-1}}}$ be a set of $N$ symmetric states associated to a unitary transform $U$. Let $\{E_0,\dots,E_{N-1},E_F\}$ be any unambiguous state discrimination measurement associated to these states and let $P_j$ and $\overline{P_D}$ be the associated success probabilities. $\overline{P_D}$ always satisfies
	\begin{equation}\label{eq:PD}
	\overline{P_D}  \leq  N  \min_{r \in \Iint{0}{N-1}} |c_r|^2 , 
	\end{equation}
where $c_r$ are the coordinates of $\ket{\psi_0}$ in the eigenbasis $\left\{\ket{\gamma_r},r \in \Iint{0}{N-1}\right\}$ of $U$, \ie\ $\ket{\psi_0} = \sum_{i=0}^{N-1} c_r \ket{\gamma_r}$. There is a POVM which meets \eqref{eq:PD} with equality.
\end{proposition}

A corollary of this result is obtained by taking the Hilbert space of dimension a prime number $p$ and take 
$U$ as the shift operator $U\ket{x} = \ket{x+1}$ where addition is performed in $\F_p$. It is easy to verify 
that in this case, the maximal average probability of discrimination $\overline{P_D}^{\text{max}}$ is given by
$$
\overline{P_D}^{\text{max}} = p \min_{r \in \F_p} |\FT{f}(r)|^2
$$ 
when $\ket{\psi_0}= \sum_{e \in \F_p} f(e) \ket{e}$. This is a consequence of the fact 
that an eigenbasis of $U$ is given by $\left\{\ket{\QFT{x}},x \in \F_p\right\}$ (this is implied by Proposition \ref{proposition:shift})
and from
$$
\ket{\psi_0} = \sum_{x \in \F_p} \FT{f}(-x) \ket{\QFT{x}}.
$$
The last equation follows from Fact \ref{fact:inverse}.
We will actually use and prove a slightly more general result, where in particular the dimension of the Hilbert space is not prime anymore (in which case we can not apply 
Proposition \ref{proposition:USSD}) 

\begin{proposition}\label{proposition:USD}
Let $\ket{\psi}= \sum_{\yv \in \Fqn} f(\yv) \ket{\yv}$ for some function $f : \Fqn \rightarrow \mathbb{C}$ s.t. $\norm{f}_2 = 1$ and for 
$\bv \in \Fqn$, let $\ket{\psi_\bv}\eqdef X_\bv \ket{\psi}$. When the states $\ket{\psi_\bv}$ are all linearly independent, unambiguous state discrimination of the states 
$\{\ket{\psi_\bv},\;\bv \in \Fqn\}$ is possible and has a maximal average probability of discrimination given by 
$$
\overline{P_D}^{\text{max}} = q^n \min_{\xv \in \Fqn} |\FT{f}(\xv)|^2.
$$
\end{proposition}

The proof of this statement borrows many ideas from \cite{CB98}. Before giving it, we have to recall a few points (see \cite[\S\ II]{CB98} for more details) about unambiguous state discrimination.
\paragraph{Unambiguous state discrimination of linearly independent states.}
Let $\Hc$ be the Hilbert space spanned by the $\ket{\psi_\bv}$'s for $\bv$ ranging over $\Fqn$. An optimal (leading to the maximal average probability of discrimination) POVM $\{E_\bv, \bv \in \Fqn\} \cup \{E_F\}$ distinguishing unambiguously all  the $\ket{\psi_\bv}$, where $E_{\bv}$ detects unambiguously $\ket{\psi_\bv}$ for all $\bv$ in $\Fqn$, can be chosen of the form
\begin{equation}
E_{\bv} = \frac{P_\bv}{|\braket{\psi_\bv^\perp}{\psi_\bv} |^2} \ket{\psi_\bv^\perp}\bra{\psi_\bv^\perp}
\end{equation}
where $P_\bv$ is the probability of detecting $\ket{\psi_\bv}$ given that the input state was of this form and the $\{\ket{\psi_\bv^\perp},\;\bv \in \Fqn\}$ are the reciprocal states of the 
$\ket{\psi_\bv}$'s. $\ket{\psi_\bv^\perp}$ is the state (unique up to a irrelevant phase) which belongs to $\Hc$ and is orthogonal to all other $\ket{\psi_\av}$ for $\av$ ranging over 
$\Fqn \setminus \{\bv\}$. The average probability of discrimination is then
$$
\overline{P_D} = \frac{1}{q^n} \sum_{\bv \in \Fqn} P_\bv.
$$
Let 
$$E_D \eqdef \sum_{\bv \in \Fqn} E_{\bv}$$
Since $E_D+E_F = \un$ and $E_F$ should be a positive semi-definite operator, it is readily verified that an optimum POVM (i.e. one that gives the maximum average probability of discrimination)
has necessarily its maximum eigenvalue $\lmax(E_D)$ equal to $1$. From these considerations, we see that if we bring in $\Am_\bv \eqdef \frac{1}{|\braket{\psi_\bv^\perp}{\psi_\bv} |^2} \ket{\psi_\bv^\perp}\bra{\psi_\bv^\perp}$ then the problem of maximizing $\overline{P_D}$ is nothing but the problem of maximizing 
$\frac{1}{q^n} \sum_{\bv \in \Fqn} P_\bv$ given that $0 \leq P_\bv \leq 1$ for all $\bv$ in $\Fqn$ and $\un - \sum_{\bv \in \Fqn} P_{\bv} \Am_\bv \succeq 0$ (\ie \ is a positive semi-definite matrix).
No general solution to this problem is known, with the notable exception of the symmetric states case given above and our case given in Proposition \ref{proposition:USD}. An averaging argument can be used in such a case to show that actually in the optimal solution all the $P_j$ can be chosen to be equal which makes the optimization trivial.   

\paragraph{An averaging argument.}
The proof of Proposition \ref{proposition:USSD} of \cite{CB98} relies essentially on an averaging argument which is used to show that there is an optimal POVM that satisfies a certain kind of invariance relation and 
whose individual discrimination probabilities are all the same.
We show that a similar result also holds in our case 
\begin{lemma}\label{lem:averaging}
Assume that an optimal POVM is $\{E_\bv,\bv \in \Fqn\} \cup \{E_F\}$. Denote by $\overline{P_D}^{\text{max}}$ its average probability of discrimination.
 Define for all $\bv$ in $\Fqn$, 
$E_\bv^{\text{ave}} \eqdef \frac{1}{q^n} \sum_{\av \in \Fqn} X_\av E_{\bv-\av} X_{-\av}$. We also let $E_D^{\text{ave}} \eqdef \sum_{\bv \in \Fqn} E_\bv^{\text{ave}}$ and 
$E_F^{\text{ave}} \eqdef \un - E_D^{\text{ave}}$.
Then 
$\{E^{\text{ave}}_\bv,\bv \in \Fqn\} \cup \{E_F^{\text{ave}}\}$ is also an optimal POVM that satisfies for all $\av$ in $\Fqn$ the invariance relation
$X_{\av} E^{\text{ave}}_\bv X_{-\av} = E^{\text{ave}}_\bv$. Moreover for this new POVM, the discrimination probability $P_\bv^{\text{ave}} \eqdef 
\bra{\psi_\bv} E^{\text{ave}}_\bv\ket{\psi_\bv}$ of $\ket{\psi_\bv}$ is equal to the maximal average discrimination 
probability $\overline{P_D}^{\text{max}}$ for all $\bv$ in $\Fqn$.
\end{lemma}
\begin{proof}
Clearly for all $\av$ in $\Fqn$, the POVM $\{X_\av E_\bv X_{-\av},\bv \in \Fqn\} \cup \{X_\av E_F X_{-\av}\}$ gives an unambiguous discrimination for the set of states 
$\{\ket{\psi_\bv},\bv \in \Fqn\}$. We call this POVM, the original POVM shifted by $\av$. However now the operator $X_\av E_\bv X_{-\av}$ detects the state $\ket{\psi_{\av+\bv}}$. Let $P_\bv$ be the discrimination probability of $\ket{\psi_{\bv}}$ by the operator 
$E_\bv$, that is $P_\bv = \bra{\psi_\bv}E_\bv\ket{\psi_\bv}$ and $P_\bv^\av$ be the discrimination probability of the same state, but this time by the POVM $\{X_\av E_\bv X_{-\av},\bv \in \Fqn\} \cup \{X_\av E_F X_{-\av}\}$. Since $\ket{\psi_{\bv}}$ is now detected by $X_\av E_{\bv-\av} X_{-\av}$, we have for all $\bv$ and $\av$ in $\Fqn$
\begin{equation}\label{eq:discrimination_b}
P_{\bv}^\av = P_{\bv-\av}.
\end{equation}
From these considerations, we clearly see that $E_\bv^{\text{ave}} = \frac{1}{q^n} \sum_{\av \in \Fqn}X_\av E_{\bv-\av} X_{-\av}$ detects $\ket{\psi_{\bv}}$ with probability 
$P_\bv^{\text{ave}} = \frac{1}{q^n} \sum_{\av \in \Fqn} P_{\bv-\av} = \overline{P}_D$. However, we also have to show that $\{E^{\text{ave}}_\bv,\bv \in \Fqn\} \cup \{E_F^{\text{ave}}\}$ defines a POVM. 
All the $E_\bv^{\text{ave}}$ are 
clearly positive semi-definite, it remains to check that $E_F^{\text{ave}} \eqdef \un - E_D^{\text{ave}}$ is also positive semi-definite.
For this, we observe that
\begin{eqnarray*}
E_D^{\text{ave}} &\eqdef &\sum_{\bv \in \Fqn} E_\bv^{\text{ave}}\\
& = & \frac{1}{q^n} \sum_{\bv \in \Fqn} \sum_{\av \in \Fqn} X_\av E_{\bv-\av} X_{-\av}\\
& = & \frac{1}{q^n} \sum_{\bv \in \Fqn} \sum_{\av \in \Fqn} X_\av E_{\bv} X_{-\av}\\
& = & \frac{1}{q^n}\sum_{\av \in \Fqn} X_\av E_D X_{-\av}
\end{eqnarray*}
where $E_D \eqdef \sum_{\bv \in \Fqn} E_\bv$. By convexity of the maximum eigenvalue on the space of Hermitian operators on $\hh$ we have
\begin{equation}
\label{eq:convexity} 
\lmax(E_D^{\text{ave}}) \leq \frac{1}{q^n} \sum_{\av \in \Fqn} \lmax(E_D^\av)
\end{equation}
where $E_D^\av = \sum_{\bv \in \Fqn} X_\av E_D X_{-\av}$.
The shifted POVM by $\av$ is indeed a POVM and we have therefore $\lmax(E_D^\av) \leq 1$. This together with \eqref{eq:convexity} shows that $\lmax(E_D^{\text{ave}})\leq 1$ and 
that therefore $E_F^{\text{ave}}= \un - E_D^{\text{ave}}$ is indeed positive semi-definite.
\end{proof}

\paragraph{Choosing the appropriate basis.}
The appropriate basis which simplifies a lot the computation is the common diagonalization basis of all the $X_\bv$'s. It is given by the ``character'' basis
$\{\ket{\QFT{\xv}}, \xv \in \Fqn\}$ (see Proposition \ref{proposition:shift}) and we have
\begin{equation}\label{eq:eigenvalue_shift}
X_\bv \ket{\QFT{\xv}} = \chi_{\xv}(-\bv) \ket{\QFT{\xv}}.
\end{equation}
From this, we deduce that for all $\bv$ in $\Fqn$ we have
\begin{equation}
X_\bv = \sum_{\xv \in \Fqn} \chi_{\xv}(-\bv) \ket{\QFT{\xv}}\bra{\QFT{\xv}}
\end{equation}
If we express $\ket{\psi_\zerom}$ in this basis, we obtain
$$
\ket{\psi_\zerom}= \sum_{\xv \in \Fqn} c_{\xv} \ket{\QFT{\xv}}
$$
then all the other ones are given by
\begin{equation}\label{eq:psib}
\ket{\psi_\bv}= X_\bv \ket{\psi_\zerom} =  \sum_{\xv \in \Fqn} c_{\xv} \chi_{\xv}(-\bv) \ket{\QFT{\xv}}.
\end{equation}
It is readily verified that the reciprocal states are given by
\begin{equation}
\label{eq:reciprocal_state}
\ket{\psi_\bv^\perp}= \frac{1}{\sqrt{Z}} \sum_{\xv \in \Fqn} \frac{1}{\overline{c_\xv}}\chi_{\xv}(-\bv) \ket{\QFT{x}}
\end{equation}
where $Z = \sum_{\xv \in \Fqn} |c_\xv|^{-2}$. Indeed, we observe that for any $\av$ and $\bv$ in $\Fqn$ we have
\begin{align}
\braket{\psi_\av^\perp}{\psi_\bv}  = \frac{1}{\sqrt{Z}} \sum_{\xv \in \Fqn} \overline{\chi_{\xv}(-\av)} \chi_{\xv}(-\bv) = \frac{1}{\sqrt{Z}} \sum_{\xv \in \Fqn} \chi_{\xv}(\av-\bv) = \frac{q^n}{\sqrt{Z}} \delta(\av,\bv) \label{eq:scp_state_and_reciprocal}
\end{align}
where $\delta(\xv,\yv)$ is the Kronecker function which is equal to $1$ iff $\xv=\yv$ and to $0$ otherwise.

We have now all the tools we need to prove Proposition \ref{proposition:USD}.
\begin{proof}[Proof of Proposition \ref{proposition:USD}]
From Lemma \ref{lem:averaging} we can choose the $E_\bv$ of the optimal POVM as
\begin{equation}
\label{eq:Eb}
E_\bv =  \frac{\overline{P_D}}{|\braket{\psi_\bv^\perp}{\psi_\bv} |^2} \ket{\psi_\bv^\perp}\bra{\psi_\bv^\perp}.
\end{equation}
By \eqref{eq:scp_state_and_reciprocal} we know that
$|\braket{\psi_\bv^\perp}{\psi_\bv} |^2 =  \frac{q^{2n}}{Z}$, and therefore by plugging 
this expression in \eqref{eq:Eb} and using \eqref{eq:psib} and \eqref{eq:reciprocal_state} we obtain
\begin{eqnarray*}
E_\bv & = & \frac{Z}{q^{2n}} \frac{\overline{P_D}}{Z} \sum_{\substack{\xv \in \Fqn\\\yv\in \Fqn}}  \frac{1}{\overline{c_\xv}c_\yv} \chi_\xv(-\bv)\overline{\chi_\yv(-b)}\ket{\QFT{\xv}}\bra{\QFT{\yv}}\\
& = & \frac{\overline{P_D}}{q^{2n}} \sum_{\substack{\xv \in \Fqn\\\yv\in \Fqn}}  \frac{1}{\overline{c_\xv}c_\yv} \chi_\bv(-\xv)\chi_\bv(\yv)\ket{\QFT{\xv}}\bra{\QFT{\yv}}\\
& = & \frac{\overline{P_D}}{q^{2n}} \sum_{\substack{\xv \in \Fqn\\\yv\in \Fqn}}  \frac{1}{\overline{c_\xv}c_\yv} \chi_\bv(\yv-\xv)\ket{\QFT{\xv}}\bra{\QFT{\yv}}
\end{eqnarray*}
From this we infer that
\begin{eqnarray*}
E_D & = & \sum_{\bv \in \Fqn} E_\bv \\
& = &  \frac{\overline{P_D}}{q^{2n}} \sum_{\bv \in \Fqn} \sum_{\substack{\xv \in \Fqn\\\yv\in \Fqn}}  \frac{1}{\overline{c_\xv}c_\yv} \chi_\bv(\yv-\xv)\ket{\QFT{\xv}}\bra{\QFT{\yv}}\\
& = & \frac{\overline{P_D}}{q^{2n}} \sum_{\substack{\xv \in \Fqn\\\yv\in \Fqn}}\sum_{\bv \in \Fqn}\frac{1}{\overline{c_\xv}c_\yv} \chi_\bv(\yv-\xv)\ket{\QFT{\xv}}\bra{\QFT{\yv}}\\
& = & \frac{\overline{P_D}}{q^{n}} \sum_{\xv \in \Fqn} \frac{1}{|c_\xv|^2} \ket{\QFT{\xv}}\bra{\QFT{\xv}},
\end{eqnarray*}
where in the last line we used that $\sum_{\bv \in \Fqn} \chi_\bv(\yv-\xv)=q^n \delta(\xv,\yv)$ by Proposition \ref{Proposition;Characters}.
Since the $\ket{\QFT{\xv}}\bra{\QFT{\xv}}$'s form an orthonormal set of projectors we have that $\lmax(E_D) = \frac{\overline{P_D}}{q^n \min_{\xv \in \Fqn}|c_\xv|^2}$.
From the fact that we should have $\lmax(E_D) \leq 1$ in order $E_F$ to be positive semi-definite, we have
$$
\overline{P_D} \leq q^n \min_{\xv \in \Fqn} |c_\xv|^2.
$$
Clearly the optimum is attained when we have equality here and therefore
\begin{eqnarray*}
\overline{P_D}^{\text{max}} & = & q^n \min_{\xv \in \Fqn} |c_\xv|^2\\
& = & q^n \min_{\xv \in \Fqn} |\QFT{f}(\xv)|^2
\end{eqnarray*}
where we used Fact \ref{fact:inverse} for the last point which gives $c_\xv=\QFT{f}(-\xv)$.
\end{proof}

\begin{remk}
It is readily seen that the two crucial ingredients of the proof are that (i) we can take an ``average'' of an optimal solution to show that there is an optimal
solution where all states are discriminated with the same probability, (ii) a basis which simplifies the computation. (i) holds in a more general case where
the set of states is of the form $\{U \ket{\psi}, U \in G\}$ where $G$ is a finite group of unitaries. On top of that, (ii) holds for instance if the group $G$ is Abelian, the nice basis
is then provided by the common diagonalization basis of the $U$'s. It other words, it is straightforward  to generalize  Proposition \ref{proposition:USSD} in the case where the set of states 
is of the form $\{U \ket{\psi}, U \in G\}$ where $G$ is a finite Abelian group.
\end{remk}

\subsection{Quantum polynomial time algorithm for $\DPQ$ in the $q$-ary setting}
The goal of the previous subsection was to extend unambiguous state discrimination to our $q$-ary setting. When we apply Proposition \ref{proposition:USD} in our case we obtain
\begin{proposition}[Unambiguous state discrimination, $q$-ary case]\label{Proposition:USDFinal} 
	Let $\omega \le \frac{q-1}{q}$. For each $a \in \F_q$, we define $\ket{\psi_a} = \sqrt{1-\omega} \ket{a} + \sum_{b \neq a} \sqrt{\frac{\omega}{q-1}} \ket{b}$. There exists a POVM $\{\{E_a\}_{a \in \F_q},E_F\}$ s.t.
	\begin{align*}
	\forall a \in \F_q, \ \tr(E_a \kb{\psi_a}) & \eqdef \Pusd = {\frac{q \cdot \omega^\perp}{q-1}}\\
	\forall a,b \neq a \in \F_q, \ \tr(E_b \kb{\psi_a}) & = 0 \\
	\end{align*}
\end{proposition}
Notice that since $\{\{E_a\}_{a \in \F_q},E_F\}$ is a POVM, this implies for each $a \in \F_q$ $\tr(E_F \kb{\psi_a}) = 1 - \Pusd = 1 - {\frac{q \cdot \omega^\perp}{q-1}} $
\begin{proof}
	 We define $\ket{\psi} = \sum_{x \in \F_q} f(x)\ket{x}$ with $f(0) = \sqrt{1-\omega}$ and $f(x) = \sqrt{\frac{\omega}{q-1}}$ for $x \in \F_q^*$. With this definition, $\ket{\psi_a} = X_a \ket{\psi}$. As computed in Lemma~\ref{lemma:QFTpsi}, we have 
	$$ \QFT{f}(0) = \sqrt{1 - \omega^\perp} \quad ; \quad \QFT{f}(y) = \sqrt{\frac{\omega^\perp}{q-1}} \ \textrm{ for } y \in \F_q^*.$$
	with $ \omega^\perp = \frac{\Rbra{\sqrt{(q-1)(1-\omega)} - \sqrt{\omega}}^{2}}{q}.$ One can check that for $\omega \in [0,\frac{q-1}{q}]$, we have $\hat{f}(0) \ge \hat{f}(y)$ for $y \in \F_q^*$. We use Proposition~\ref{proposition:USD} with $n=1$ to immediately get  
	$$ \Pusd = q \cdot \min_y |\hat{f}(y)|^2 = {\frac{q \cdot \omega^\perp}{q-1}}$$
\end{proof}

It also turns that this operation can be implemented efficiently in poly-log time (in $q$) as shown by
\begin{proposition}\label{proposition:efficientUSD}
Consider the unitary $U$ acting on  $\ket{\psi_a}\ket{0}$ as
\begin{align*}
	U \ket{\QFT{0}}\ket{0} & = \ket{\QFT{0}}\left(u\ket{0} + \sqrt{1 - u^2}\ket{1}\right) & \textrm{ with } u = \sqrt{\frac{\omega^\perp}{(1-\omega^\perp)(q-1)}} \\
	U \ket{\QFT{\alpha}}\ket{0} & = \ket{\QFT{\alpha}}\ket{0}  & \forall \alpha \in \F_q^*
	\end{align*}
	With our choice of function $f$, the above unambiguous state discrimination quantum measurement can be done by applying $U$ on $\ket{\psi_\alpha}\ket{0}$ and 
	and then measuring the output state in the computational basis. 
	This can be done in time $O(\polylog(q))$.
\end{proposition}
\begin{proof}
Let us start the proof by writing  $\ket{{\psi_\alpha}}$ in the Fourier basis $\{\QFT{x},\;x \in \Fq\}$. This can be done by observing that
\begin{eqnarray*}
\ket{{\psi_\alpha}} &= & X_\alpha \ket{\psi} \\
& = & X_\alpha \cdot \QFTt \cdot \QFTt^\dagger \ket{\psi}  \\
& = & X_\alpha \cdot \QFTt \left( \sqrt{1 - \omega^\perp}\ket{0} + \sum_{\gamma \in \F_Q^*}  \sqrt{\frac{\omega^\perp}{q-1}} \ket{\gamma} \right) \; \text{(by Lemma \ref{lemma:QFTpsi} and $\perpF{\perpO{\omega}}=\omega$)}\\
& = & \QFTt \cdot Z_{-\alpha} \left( \sqrt{1 - \omega^\perp}\ket{0} + \sum_{\gamma \in \F_Q^*}  \sqrt{\frac{\omega^\perp}{q-1}} \ket{\gamma} \right)  \text{(by Lemma \ref{lemma:QFTpsib})}\\
& = & \QFTt \left(  \sqrt{1 - \omega^\perp}\ket{0} + \sum_{\gamma \in \F_Q^*}  \chi_{-\alpha}(\gamma)\sqrt{\frac{\omega^\perp}{q-1}} \ket{\gamma}\right) \\
& = & \sqrt{1 - \omega^\perp}\ket{\QFT{0}} + \sum_{\gamma \in \F_Q^*} \chi_{-\alpha}(\gamma) \sqrt{\frac{\omega^\perp}{q-1}} \ket{\QFT{\gamma}}.
\end{eqnarray*}
	
	Applying $U$ on $\ket{\psi_\alpha}\ket{0}$, we obtain 
	$$ \ket{\phi'_\alpha} := U \ket{\Psi_\alpha}\ket{0} = \left(u \sqrt{1 - \omega^\perp} \ket{\QFT{0}} + \sum_{\gamma \in \F_q^*} \chi_{-\alpha}(\gamma) \sqrt{\frac{\omega^\perp}{q-1}} \ket{\QFT{\gamma}}\right)\ket{0} + \sqrt{1 - u^2}\sqrt{1 - \omega^\perp} \ket{\QFT{0}}\ket{1}.$$
	Notice that $u \sqrt{1-\omega^\perp} = \sqrt{\frac{\omega^\perp}{q-1}}$ and that $\ket{\alpha} = \frac{1}{\sqrt{q}} \sum_{\gamma \in \F_q} \chi_{-\alpha}(\gamma) \ket{\QFT{\gamma}}$ (by Fact \ref{fact:inverse}).
	From there, we can rewrite
	$$ \ket{\phi'_\alpha} = \sqrt{\frac{q \omega^\perp}{q-1}} \ket{\alpha}\ket{0} + \sqrt{1 - u^2}\sqrt{1 - \omega^\perp} \ket{\QFT{0}}\ket{1}.$$
	We now measure all the qubits in the computational basis. If the last qubit is $0$, the measurement outputs the value $\alpha$ in the first register. If the last qubit is $1$, we output Fail. The measurement succeeds and outputs the correct value $\alpha$ w.p. $\frac{q \omega^\perp}{q-1}$. The time to perform $U$ is essentially the time to perform two Quantum Fourier Transforms so  $U$ can be efficiently computed in time $O(\polylog(q))$, the whole measurement can be done in time $O(\polylog(q))$.
\end{proof}

We can now present our polynomial time algorithm in the $q$-ary setting:

\begin{theorem}\label{Theorem:polyFull}
	Let $R > 0$ and $\omega \in (0,\frac{q-1}{q})$ satisfying $\frac{q \cdot \omega^\perp}{q-1} > R$. There exists a quantum algorithm that solves $\DPQ(q,n,\lfloor Rn \rfloor, \omega)$ in time $\poly(n,\log(q))$. 
\end{theorem}
\begin{proof}
	We fix $R > 0, k = \lfloor Rn \rfloor$ and $\omega \in (0,\frac{q-1}{q})$ satisfying $\frac{q \omega_\perp}{q-1} > R$. We are given a random generating matrix $\Gm \in \F_q^{k \times n}$ with associated code $\C$ as well as a state $\ket{\psi_\cv} = \bigotimes_{i = 1}^n \ket{\psi_{c_i}}$ for a randomly chosen $\cv \in \C$, where 
	$$\ket{\psi_{c_i}} = \sqrt{1 - \omega}\ket{c_i} + \sum_{x \neq c_i} \sqrt{\frac{\omega}{q-1}} \ket{x}.$$
	As in Section~\ref{Section:BinaryPolynomial}, we consider the following algorithm. \\ \\
		\noindent 	\cadre{\begin{center} Quantum algorithm for $\DPQ$ using $q$-ary USD \end{center}
		\begin{enumerate}
			\item Perform the optimal unambiguous measurement given in Proposition~\ref{Proposition:USDFinal} from Proposition~\ref{Proposition:USD} on each register $i$ in order to guess $c_i$, which can be done w.p. $\Pusd = \frac{q \omega^\perp}{q-1}$. Let $J \subseteq [n]$ be the set of indices where this measurement succeeds. The algorithm recovers here $\cv_J$. 
			\item If $G_J \in \F_q^{k \times |J|}$ is of rank $k$, recover $\cv$ from $\cv_J$ by computing $\cv_J G_J^{-1}  G$.
		\end{enumerate}
	}	$ \ $ \\

By our choice of $\omega$, we have $\Pusd > R$ which means that there exists an absolute constant $\gamma > 0$ s.t. $\Pusd = R + \gamma$. This in turn implies that the success probability of this algorithm is $1 - o(1)$, using the same arguments as in Section~\ref{Section:BinaryPolynomial}.
\end{proof}

 \section{(In)tractability of the quantum decoding problem}\label{Section:Intractability}
In this section we provide a full characterization of the tractability of $\DPQ(q,n,k,\omega)$. We show that the problem is tractable {\ie } there exists a quantum algorithm that solves the problem w.p. $1 - o(1)$ (as $n \rightarrow + \infty$ and $q = \Omega(1)$) for any absolute constant $\omega < \left(\drgv(q,1-k/n)\right)^\perp$. We will simplify the notation as alluded in Subsection \ref{ss:notation} and write from now on just $\drgv(1-k/n)$ instead of $\drgv(q,1-k/n)$. Moreover, $R$ denotes in the whole section the rate $\frac{k}{n}$ of the code we decode:
$$
R \eqdef k/n.
$$
Notice here that we do not put here any restriction on the running time of the algorithm. On the other hand, we show that the problem is intractable {\ie } all quantum algorithms solve the problem w.p. at most $o(1)$ for any absolute constant $\omega > \left(\drgv(1-R)\right)^\perp$.

In order to prove our results, we will focus on a single quantum algorithm: the one that performs a pretty good measurement. Recall that in the quantum decoding problem, we have to recover $\cv$ from $\ket{\psi_\cv} = \sum_{\ev} \sqrt{(\frac{\omega}{q-1})^{|\ev|}(1-\omega)^{n - |\ev|}}\ket{\cv + \ev}$. In order to prove our results, we will focus on a single quantum algorithm performing the pretty good measurement on the Fourier transforms of these states. For the tractability result, we show that the PGM recovers $\cv$ w.p. $1 - o(1)$. For the intractability result, we show that the PGM recovers $\cv$ w.p. $o(1)$. But we know from Proposition~\ref{Proposition:PGMOptimality} that this implies that any quantum algorithm will recover $\cv$ w.p. $o(1)$ hence the intractability result. 

We first study the PGM for any error function $f$ and then apply our results to $f(\ev) = \sqrt{(\frac{\omega}{q-1})^{|\ev|}(1-\omega)^{n - |\ev|}}$ in order to show our (in)tractability results.

\subsection{Computing the PGM associated to the quantum decoding problem}
We fix a generating matrix $\Gm$ and an associated code $\C$. In order to study our PGM, we define the shifted dual codes of $\C$
$$ \Csp  \eqdef  \{\xv \in \F_q^n: \Gm\trsp{\xv}=\sv\} $$ 
Notice that $\C^\bot_{\mathbf{0}} = \C^\bot$ where $\C^\bot$ is the dual code of $\C$. For each shifted dual code $\Csp$, we fix an element $u_\sv \in \Csp$. We have $\Csp = \{u_\sv + \dv : \dv \in \C^\bot\}$. This means that for all $\cv$ in $\C$ and all $\yv$ in $\Csp$ 
\begin{align*}
\chi_\cv(\yv) = \chi_\cv(u_\sv)
\end{align*}
Moreover, for any $\sv, \sv'$ in $\F_q^k$ s.t. $\sv' \neq \sv$, since $u_{\sv} + u_{\sv'} \notin \C^\bot$, we have 
\begin{align}\label{Eq:5.1}
	\sum_{\cv \in \C} \chi_\cv(u_\sv + u_{\sv'}) = 0
\end{align} 

Now fix any error function $f : \F_q^n \rightarrow \mathbb{C}$ s.t. $\norm{f}_2 = 1$, and consider the states $\ket{\psi_\cv} = \sum_{\ev \in \F_q^n} f(\ev)\ket{\cv + \ev}$. The goal is to recover $\cv$. Actually, we will start from $\ket{\widehat{\psi_\cv}} = \QFTt{\ket{\psi_\cv}}$ instead of $\ket{\psi_\cv}$ and apply the Pretty Good Measurement on the ensemble of states $\{\ket{\QFT{\psi_\cv}}\}$. The distinguishing problem is equivalent since applying $\QFTt$ is a unitary operation. 
We first define the states 
	\begin{align*}
		\ket{W_\sv} & \eqdef \sum_{\yv \in \C_\sv^\bot} \hat{f}(\yv) \ket{\yv} \quad \textrm{ (not normalized)} \\
		\ket{\tW_\sv} & \eqdef  \frac{\ket{W_\sv}}{\norm{\ket{W_\sv}}}
\end{align*}
as well as $n_\sv \eqdef \ \norm{\ket{W_\sv}} = \sqrt{\sum_{\yv \in \C_\sv^\bot} |\hat{f}(\yv)|^2}$. We first write $\ket{\QFT{\psi_\cv}}$ in the $\{\ket{W_\sv}\}$ basis. 

\begin{lemma}\label{Lemma:PsiHat}
	$\ket{\widehat{\psi_\cv}} = \sum_{\sv \in \F_q^k} \chi_{\cv}(u_\sv) \ket{W_\sv}$.
\end{lemma}
\begin{proof}
	We write 
	\begin{align*}
		\ket{\widehat{\psi_\cv}} & = \frac{1}{\sqrt{q^n}} \sum_{\yv,\ev \in \F_q^n} \chi_{\cv + \ev}(\yv) f(\ev) \ket{\yv} \\
		& = \frac{1}{\sqrt{q^n}} \sum_{\sv \in \F_q^{k}}\sum_{\xv \in \Csp}\chi_{\cv}(\xv) \sum_{\ev \in \F_q^n} \chi_{\ev}(\xv)f(\ev)\ket{\xv}  \\
		& = 
				\sum_{\sv \in \F_q^{k}} \chi_{\cv}(u_{\sv}) \sum_{\xv \in \Csp} \hat{f}(\xv)\ket{\xv} 
		 = \sum_{\sv \in \Csp} \chi_{\cv}(u_\sv) \ket{W_\sv}.
	\end{align*}
\end{proof}
We can now explicit the PGM associated to the states $\ket{\widehat{\psi_\cv}}$.
\begin{proposition}\label{Proposition:PGMDescription}
	The PGM associated to the ensemble of states $\{\ket{\widehat{\psi_\cv}}\}_{\cv \in C}$ is the projective measurement $\{\kb{Y_\cv}\}_{\cv \in \C}$ where $\ket{Y_\cv} = \frac{1}{\sqrt{q^k}} \sum_{\sv \in \F_q^k} \chi_{\cv}(u_\sv) \ket{\tW_\sv}$.
\end{proposition}
\begin{proof}
	We write the PGM $\{M_{\cv}\}$ associated to the states $\ket{\widehat{\psi_\cv}}$ using Definition~\ref{Definition:PGM}.
	\begin{align*}
		M_\cv = \rho^{-1/2} \kb{\widehat{\psi_\cv}} \rho^{-1/2} \quad \textrm{ given } \rho = \sum_{\cv \in \F_q^k}  \kb{\widehat{\psi_\cv}}
	\end{align*}
	We now write 
	\begin{align*}
		\rho = \sum_{\cv \in \C} \kb{\widehat{\psi_\cv}} = 
				\sum_{\cv \in \C} \sum_{\sv,\sv' \in \F_q^{k}} \chi_{\cv}(u_{\sv} - u_{\sv'}) \ketbra{W_\sv}{W_{\sv'}} = q^k \sum_{\sv \in \F_q^{k}} \kb{W_\sv} 
	\end{align*}
	where we use Equation~\ref{Eq:5.1} as well as $\chi_\cv(\mathbf{0}) = 1$ for the last equality. 
Using the fact that the $\ket{W_\sv}$ are pairwise orthogonal (since they have disjoint support in the computational basis), the  $\kb{\tW_\sv}$'s are pairwise orthogonal projectors and we have
	\begin{align*}
		\rho =  q^k
				\sum_{\sv \in \F_q^{k}} n^2_\sv \kb{\tW_\sv} \ \ \textrm{hence} \ \ \rho^{-1/2} = 
		\frac{1}{\sqrt{q^k}}\sum_{\sv \in \F_q^{k}} \frac{1}{{n_\sv}} \kb{\tW_\sv}  
	\end{align*}
	and 
	\begin{align}\label{Eq:8}\rho^{-1/2} \ket{\widehat{\psi_\cv}} = \frac{1}{\sqrt{q^k}} \sum_{\sv \in \F_q^k} \chi_{\cv}(u_\sv)\ket{\tW_\sv} := \ket{Y_\cv}.
	\end{align} 
	Here $\ket{Y_\cv}$ is a pure state of norm $1$. Also, notice that these states are pairwise orthogonal. So $M_\cv = \kb{Y_{\cv}}$ and the PGM is  just the projective measurement $\{\kb{Y_\cv}\}_{\cv \in \C}$. 
\end{proof}

Finally, we can explicit the probability that the PGM succeeds on the states $\ket{\widehat{\psi_\cv}}$.

\begin{proposition}\label{Proposition:PGMValue}
	The PGM succeeds to recover $\cv$ from  $\ket{\widehat{\psi_\cv}}$ w.p.
	$\frac{1}{{q^k}}\left(\sum_{\sv \in \F_q^k} n_\sv\right)^2$.
\end{proposition}
\begin{proof}
	From the previous proposition, the PGM we use is the projective measurement  $\{\kb{Y_\cv}\}_{\cv \in \C}$ with $\ket{Y_\cv} =  \frac{1}{\sqrt{q^k}} \sum_{\sv \in \F_q^k} \chi_{\cv}(u_\sv) \ket{\tW_\sv}$.
For each $\cv \in \C$, we write using Lemma~\ref{Lemma:PsiHat} as well as the expression of $\ket{Y_\cv}$ the probability $p_\cv$ that this measurement succeeds
\begin{align}\label{eq:pc}
	p_\cv \eqdef |\braket{Y_\cv}{\widehat{\psi_{\cv}}}|^2 = \frac{1}{q^{k}} \left|\sum_{\sv} \braket{\tW_\sv}{W_\sv} \right|^2 = \frac{1}{q^{k}} \left(\sum_{\sv \in \F_q^k} n_s\right)^2
\end{align}
which immediately gives us the result. 
\end{proof}

\begin{remk}\label{rem:ns_PPGM}
Since $\ket{\widehat{\psi_\cv}}$ is of norm $1$, we have immediately from Lemma~\ref{Lemma:PsiHat} that $\sum_{\sv \in \F_q^k} n_\sv^2 = 1$. In the case where all the norms are equal, we have $n_{\sv} = \sqrt{q^{-k}}$ which gives indeed $\PPGM= 1$. On the other hand, if these norms are highly unbalanced the probability that the PGM succeeds is very low. 
\end{remk}

\subsection{(In)tractability results}
\subsubsection{First computations and probabilistic arguments on random codes}\label{Section:FirstComputations}
We go back to our quantum decoding problem. Our error function corresponds to the $q$-ary symmetric channel so $f(\ev) = \left(\sqrt{1-\omega}\right)^{n - |\ev|}\left(\sqrt{\frac{\omega}{q-1}}\right)^{|\ev|}
$
we have (see Lemma \ref{lemma:QFTpsi})
$$ \hat{f}(\yv) = (\sqrt{1-\omega^\perp})^{n - |\yv|}\left(\sqrt{\frac{\omega^\perp}{q-1}}\right)^{|\yv|},$$
with 
$ \omega^\perp = \frac{\Rbra{\sqrt{(q-1)(1-\omega)} - \sqrt{\omega}}^{2}}{q}.$ For a fixed $\Gm \in \F_q^{k \times n}$ and associated code $\C$ (we will not make this dependency explicit in the notation to simplify it), we define 
\begin{eqnarray*}
	n_{\sv,\C}& \eqdef &\norm{\sum_{\yv \in \Csp} \hat{f}(\yv)\ket{\yv}}\\
	a_{\sv,\C}(t) &\eqdef & \left|\{\yv \in \C^\bot_\sv : |\yv| = t\} \right|
\end{eqnarray*}
We also define $S(t)  \eqdef \frac{(q-1)^t \binom{n}{t}}{q^k}$. Notice that $n_{\sv,\C}$ corresponds exactly to $n_\sv$ defined in the previous section but we made the dependency in $\C$ explicit. Our goal is to compute the success probability of the PGM on average on $\Gm$ so using Proposition~\ref{Proposition:PGMValue}, we want to bound the quantity 
$$ \PPGM = \E_{\Gm}\left[\frac{1}{q^k}\left(\sum_{\sv \in \F_q^k} n_{\sv,\C}\right)^2\right].$$
We first write 
\begin{align}\label{eq:nsv2}
	n_{\sv,\C}^2 = \sum_{y \in \Csp} |\hat{f}(\yv)|^2 = \sum_{t = 0}^n a_{\sv,\C}(t) \left(\frac{\omega^\perp}{q-1}\right)^t(1-\omega^\perp)^{n-t}
\end{align}
and recall from Remark~\ref{rem:ns_PPGM} that $\sum_{\sv \in \F_q^k} n_{\sv,\C}^2 = 1$. We see that to compute $P_{PGM}$, we have to say something about the terms $a_{\sv}(t,\C)$. We first have the following, which was proven for example in ~\cite{Deb23}:

\begin{proposition} $\forall t \neq 0, \ \E_{\Gm}\left[a_{\sv}(t,\C)\right] = S(t)$.
\end{proposition}

But the expected value will not be enough. We will need concentration bounds coming from the second moment technique

\begin{proposition}[Second moment technique, Proposition 3 from~\cite{Deb23}]\label{lem:concentration_as}
	Fix any $\sv \in \F_q^k$ and $t \in \Iint{1}{n}$. For any $\eps > 0$, we have 
	$$ \Pr_G\left[|a_{\sv,\C}(t) - S(t)| \ge \eps S(t)\right] \le \frac{q-1}{\eps^2 S(t)}.$$ 
	In particular, take $\eps = S(t)^{-1/3}$, we have 
	$$ \Pr_G\left[a_{\sv,\C}(t) \le S(t) (1 - \frac{1}{S(t)^{1/3}})\right] \le \frac{q-1}{S(t)^{1/3}}.$$
\end{proposition}

We can now observe two things

\begin{enumerate}
\item From the above proposition combined with Equation~\ref{eq:nsv2}, we have that when $S(t)$ is exponential, which happens when $t = \gamma n$ with $\gamma \in (\drgv(1-R),\drmax(1-R))$ 
\begin{equation}\label{eq:nsv2St}
	n_{\sv,\C}^2 \approx \sum_{t = \lfloor \drgv(1-R)n \rfloor }^{\lceil \drmax(1-R)n \rceil} S(t) \left(\frac{\omega^\perp}{q-1}\right)^t(1-\omega^\perp)^{n-t}.
\end{equation}
\item In order to estimate the above sum, first notice that 
\begin{align} \label{Eq:17}
	\sum_{t = 0}^n S(t) \left(\frac{\omega^\bot}{q-1}\right)^t(1 - \omega^\bot)^{n-t} = \frac{1}{q^k} \sum_{t = 0}^n \binom{n}{t} (\omega^\bot)^t (1 - \omega^\bot)^{n-t} = \frac{1}{q^k}.
\end{align}
But the above sum is actually the cumulative sum 
of the binomial distribution with parameters $n$ and $\omega^\perp$. It concentrates around the weight $n \omega^\perp$. This is formalized by the following proposition 
\end{enumerate}
\begin{proposition}\label{Proposition:18}
	For any absolute constant $ \eps > 0$, 
	\begin{align}\label{Eq:18}
		\sum_{t = \lfloor (\omega^\perp-\varepsilon) n \rfloor}^{\lceil (\omega^\perp+\varepsilon) n \rceil} S(t)  \left(\frac{\omega^\perp}{q-1}\right)^t(1-\omega^\perp)^{n-t} = \frac{1}{q^k}\left(1 - o(1)\right).\end{align}
\end{proposition}
We now have all the tools for our (in)tractability proofs. The main idea is the following: when $t = \omega n$ with $\omega < \drgv(1-R)^\bot$, 
we have $\omega^\perp \in (\drgv(1-R),\drmax(1-R))$ and so we can combine Equations~\ref{eq:nsv2St},\ref{Eq:18} to show that for most $\Gm$, $n_{\sv,\C}^2 = \frac{1}{q^k}(1 - o(1))$. On the other hand, when $\omega > \drgv(1-R)^\bot$,
we have $\omega^\perp \notin (\drgv(1-R),\drmax(1-R))$ and so we can combine Equations~\ref{eq:nsv2St},\ref{Eq:17},\ref{Eq:18} to show that for most $\Gm$, $n_{\sv,\C}^2 = o(1)$. The next sections will make these arguments formal and show how this allows us to conclude.

\subsubsection{Tractability}
We use the notations previously defined in Section~\ref{Section:FirstComputations}. $\omega$ will be considered as a fixed constant in $(0,1)$. Our main claim is the following
\begin{proposition}\label{Proposition:23}
	If $\omega < (\drgv(1-R))^\bot$ 
		then $\PPGM = \frac{1}{q^k} \E_{\Gm}\left[\left(\sum_{\sv \in \F_q^k} n_{\sv,\C}\right)^2\right] = 1 - o(1)$.
\end{proposition}
\begin{proof}

 Using \eqref{eq:nsv2}
we know that
$
		n_{\sv,\C}^2 =   \sum_{t = 0}^n a_{\sv,\C}(t) \left(\frac{\omega^\perp}{q-1}\right)^t(1-\omega^\perp)^{n-t}.
$
	Since $\omega < \left(\drgv(1 - R)\right)^\perp$, we have $\omega^\perp > \drgv(1 - R)$ so we fix $\delta > 0$ s.t.  $\omega^\perp - \delta > \drgv(1 - R).$ We therefore write 
	$$ n_{\sv,\C}^2 \ge \sum_{t = \lfloor(\omega^\perp - \delta) n \rfloor}^{\lfloor(\omega^\perp + \delta) n \rfloor} a_{\sv,\C}(t) \left(\frac{\omega^\perp}{q-1}\right)^t(1-\omega^\perp)^{n-t}.$$
	We define $t_0 =  \lfloor(\omega^\perp - \delta) n \rfloor$ and $t_1 = \lfloor(\omega^\perp + \delta) n \rfloor$.
	          Recall that $a_{\sv,\C}(t)$ is typically close to $S(t)$ as shown by Lemma \ref{lem:concentration_as}.
			This gives for $t \in \left(\lfloor(\omega^\perp - \delta) n \rfloor,\lfloor(\omega^\perp + \delta) n \rfloor\right)$
	\begin{align*}
	\Pr_G\left[a_{\sv,\C}(t) \le S(t) (1 - \frac{1}{S(t_0)^{1/3}})\right] & \le \Pr_G\left[a_{\sv,\C}(t) \le S(t) (1 - \frac{1}{S(t)^{1/3}})\right] \\
	& \le  \frac{q-1}{S(t)^{1/3}} 
	  \le \frac{q-1}{S(t_0)^{1/3}}.
	\end{align*}
	and
$$
\Pr_G\left[\forall t \in \Iint{t_0}{t_1}, \ a_{\sv,\C}(t) \ge S(t) (1 - \frac{1}{S(t_0)^{1/3}})\right]  
\ge 1 - \frac{(q-1)(t_1 - t_0 + 1)}{S(t_0)^{1/3}}$$

	This implies 
	\begin{align}\label{eq:first_bound_on_ns}
		\Pr_G \left[n_{\sv,\C}^2 \ge \sum_{t = t_0}^{t_1} S(t) (1 - S(t_0)^{-1/3}) \left(\frac{\omega^\perp}{q-1}\right)^t(1-\omega^\perp)^{n-t}\right] \ge 1 - \frac{(q-1)(t_1 - t_0 + 1)}{S(t_0)^{1/3}}.
	\end{align}
We have 
\begin{align*}
\sum_{t = t_0}^{t_1} S(t) (1 - S(t_0)^{-1/3}) \left(\frac{\omega^\perp}{q-1}\right)^t(1-\omega^\perp)^{n-t} & = 
(1 - S(t_0)^{-1/3}) \sum_{t = t_0}^{t_1} S(t) \left(\frac{\omega^\perp}{q-1}\right)^t(1-\omega^\perp)^{n-t} \\ & = \frac{(1 - S(t_0)^{-1/3})}{q^k}K(n)
\end{align*}
where $K(n) \eqdef \sum_{t = t_0}^{t_1} \binom{n}{t} \left(\omega^\perp\right)^t(1-\omega^\perp)^{n-t}$ and $K(n)=1-o(1)$ as $n$ tends to infinity 
by Proposition~\ref{Proposition:18}.
By plugging this equality in the left-hand side of \eqref{eq:first_bound_on_ns} we obtain
\begin{align*}
	\Pr_G \left[n_{\sv,\C} \ge \sqrt{\frac{1 - S(t_0)^{-1/3}}{q^k} K(n)} \right] \ge  1 - \frac{(q-1)(t_1 - t_0)}{S(t_0)^{1/3}}.
\end{align*}
Since $t_0 =  \lfloor(\omega^\perp - \delta) n \rfloor$ with $\drgv(1-R) < (\omega^\perp - \delta) < \drmax(1-R) $, we have $S(t_0) = q^{\Omega(n)}$ which implies 
\begin{align*}
	\E_G[n_{\sv,\C}] & \ge \sqrt{\frac{1 - S(t_0)^{-1/3}}{q^k} K(n)}  \cdot \Pr_G \left[n_{\sv,\C} \ge \sqrt{\frac{1 - S(t_0)^{-1/3}}{q^k}K(n)} \right] \\
	& \ge \sqrt{\frac{1 - S(t_0)^{-1/3}}{q^k}K(n)} \cdot \left(1 - \frac{(q-1)(t_1 - t_0)}{S(t_0)^{1/3}}\right) \\ & = \frac{1}{\sqrt{q^k}}\left(1 - o(1)\right)
\end{align*}
which gives 
\begin{align}\label{Eq:EG}
\E_G[\sum_{\sv \in \F_q^k} n_{\sv,\C}] \ge \sqrt{q^k}(1-o(1))
\end{align} 

In order to conclude, we use Jensen's inequality $\E_G(X^2) \geq (\E_G(X))^2$ 
and Equation~\ref{Eq:EG} to get 
\begin{align*}
P_{PGM} & = \frac{1}{q^k} \E_G\left[\left(\sum_{\sv \in \F_q^n} n_{\sv,\C}\right)^2\right] \ge \frac{1}{q^k}\left(\E_G\left[\sum_{\sv \in \F_q^k} n_{\sv,\C}\right]\right)^2
 \ge 1 - o(1)
\end{align*}
\end{proof}

\subsubsection{Intractability}
Again, we use the same notation as in the previous sections with a fixed $\omega \in (0,1)$. We show that if $\omega$ is too large 
then 
$P \eqdef \E_\Gm(\PPGM)$ is an $o(1)$ as shown by 
\begin{theorem}\label{thm:upper_bound_PPGM}
	Let $\Gm \in \F_q^{k \times n}$ be a random generating matrix and $\C$ be the associated code. Let $\omega > \left(\drgv(1 - R)\right)^\perp$ and let the states $\ket{\psi_{\cv}} = \sum_{\ev \in \F_q^n} f(\ev) \ket{\cv + \ev}$ with 
	$$ f(\ev) = \sqrt{\left(\frac{\omega}{q-1}\right)^{|\ev|} \left(1-\omega\right)^{n - |\ev|}}.$$
	The pretty good measurement distinguishes the states $\ket{\psi_{\cv}}$ w.p. $P=o(1)$.
\end{theorem}
Again, we will heavily build on the expression of $n_{\sv,\C}$ given by \eqref{eq:nsv2} in terms of the $\asCt$'s. The proof is based on the following steps
\paragraph{Step 1.}
Let us start by  giving an upper-bound on $\asCt$ which holds with probability close to $1$ for large values of $K$
\begin{lemma}\label{lem:first_moment}
For any $K >0$, any $t$ in $\Iint{1}{n}$ and any $\sv \in \F_q^k$, we have
$
\Pr_{\Gm} \left[\asCt \geq K {\cdot} S(t)\right] \leq \frac{1}{K}
$
which directly implies
\begin{equation}\label{eq:first_moment}
\Pr_{\Gm} \left[\asCt \leq K {\cdot} S(t)\right] \geq  1 - \frac{1}{K}.
\end{equation}
\end{lemma}
\begin{proof}
This is just Markov's inequality 
$\Pr_{\Gm} \left[\asCt \leq K \E_{\Gm}(\asCt) \right] \leq  \frac{1}{K}$ by recalling that \\
$\E_{\Gm}(\asCt) = S(t)$.
\end{proof}
A rather immediate corollary of this result is that
\begin{corollary}\label{cor:first_moment}
For any $\delta>0$, $\sv$ in $\F_q^k$ and $t$ in 
$\Iint{1}{n} \setminus [(\drgv(1-R)-\delta)n,(\drmax(1-R)+\delta)n]$, we have with probability $1-q^{-\Om{n}}$ that $a_{\sv,\C}(t)=0$.
\end{corollary}
\begin{proof}
In such a case we have $S(t) = q^{-\Om{n}}$ , since 
$S(t) \leq q^{n (h_q(t/n)-k/n)}$ and we use Lemma \ref{lem:first_moment} with $K= \frac{1}{\sqrt{S(t)}}$ to obtain
$$
\Pr_{\Gm} \left[\asCt \leq \sqrt{S(t)} \right] \geq 1 - \sqrt{S(t)} = 1-q^{-\Omega(n)}.$$
 We can conclude by using the fact that $a_{\sv,\C}(t)$ is an non negative integer so if $a_{\sv,\C}(t) \le \sqrt{S(t)} < 1$ then necessarily $a_{\sv,\C}(t) = 0$.
\end{proof}

\paragraph{Step 2.}
The previous results allow to show that
\begin{lemma}\label{lem:ns_small}
 Let $\omega > \left(\drgv(1 - R)\right)^\perp$. There exists an $\varepsilon >0$ such that for any $\sv \in \F_q^k$ which is non zero we have 
$$
\Pr_{\Gm} \left[ n_{\sv,\C} \geq q^{-k/2-\varepsilon n}\right] = q^{-\Om{n}}.
$$
\end{lemma}
\begin{proof}
Let us recall \eqref{eq:nsv2} 
$$
n_{\sv,\C}^2 = \sum_{y \in \Csp} |\hat{f}(\yv)|^2 = \sum_{t = 0}^n \asCt \left(\frac{\omega^\perp}{q-1}\right)^t(1-\omega^\perp)^{n-t}.
$$
By using Corollary \ref{cor:first_moment} and $a_{\sv,\C}(0)=0$ for $\sv \neq 0$, we obtain that for any absolute constant $\delta >0$,  
\begin{equation}
\forall \sv \in F_q^k \backslash{\{\mathbf{0}\}}, \ \Pr_{\Gm}\left[n^2_{\sv,\C} = \sum_{t = \left\lfloor (\drgv-\delta)n \right\rfloor}^ {\left\lceil (\drmax+\delta)n \right\rceil} \asCt \left(\frac{\omega^\perp}{q-1}\right)^t(1-\omega^\perp)^{n-t}\right] \ge 1 - q^{\Omega(n)},
\end{equation}
where to simplify notation we simply write $\drgv$ and $\drmax$ for $\drgv(1-R)$ and $\drmax(1-R)$ respectively.
Then by using Lemma \ref{lem:first_moment} we also deduce that for any $\delta, \delta'>0$,
 $$
 \forall \sv \in F_q^k \backslash{\{\mathbf{0}\}}, \ \Pr_{\Gm}\left[n_{\sv,\C}^2 \leq \sum_{t = \left\lfloor (\drgv-\delta)n \right\rfloor}^ {\left\lceil (\drmax+\delta)n \right\rceil} q^{\delta' n} S(t) \left(\frac{\omega^\perp}{q-1}\right)^t(1-\omega^\perp)^{n-t}\right] \ge 1 - q^{-\Omega(n)}
 $$
 We observe now that 
 \begin{equation}
 \label{eq:upper_bound_on_ns}
  \sum_{t = \left\lfloor (\drgv-\delta)n \right\rfloor}^ {\left\lceil (\drmax+\delta)n \right\rceil} q^{\delta' n} S(t) \left(\frac{\omega^\perp}{q-1}\right)^t(1-\omega^\perp)^{n-t}
  = q^{\delta' n-k}  \sum_{t = \left\lfloor (\drgv-\delta)n \right\rfloor}^ {\left\lceil (\drmax+\delta)n \right\rceil} p(t)
 \end{equation}
 where $p(t) \eqdef \binom{n}{t}  ({\omega^\perp})^t(1-\omega^\perp)^{n-t}$ is the probability that a binomial variable of parameters $n$ and $\omega^\perp$ takes the
 value $t$. By using the fact that $\omega^\perp \leq (\drgv-\delta")n$ for some $\delta">0$ and the Hoeffding inequality (see Lemma \ref{lem:Hoeffding}) we 
 deduce that for $\delta=\delta"/2$, it holds that
  $\sum_{t = \left\lfloor (\drgv-\delta)n \right\rfloor}^ {\left\lceil (\drmax+\delta)n \right\rceil} p(t) \leq q^{- \delta''' n}$ for some $\delta'''>0$. By choosing $\delta' < \delta'''$, we obtain that $n_{\sv,\C}$ is less than $q^{-k/2-\frac{\delta'''-\delta'}{2} n}$ with probability $1-q^{-\Om{n}}$. We just have to choose $\varepsilon = (\delta'''-\delta')/2$ to finish the proof.
\end{proof}

We are ready now to prove Theorem \ref{thm:upper_bound_PPGM}.
\begin{proof}[Proof of Theorem \ref{thm:upper_bound_PPGM}]
For $\sv \in \F_q^k$, let $G_\eps(\sv) = \{\Gm \in \F_q^{k \times n} : n_{\sv,\C} \ge q^{-k/2 - \eps n}\}$. Also, for $\Gm \in \F_q^{k \times n}$, let $S_{\eps}(\Gm) = \{\sv \neq 0 : n_{\sv,\C}   \ge q^{-k/2 - \eps n}\}$. The previous lemma tells us that $\forall \sv \neq \mathbf{0}, |G_{\eps}(\sv)| = o(|G|)$ where $|G|= q^{nk}$ is the total number of possible matrices $\Gm \in \F_q^{k \times n}$. Now, notice that 
$$ \sum_{\sv \neq 0} |G_\eps(\sv)| = \left|\{(\Gm,\sv) : n_{\sv,\C} \ge q^{-k/2 - \eps n}\}\right| = \sum_{\Gm} |S_{\eps}(\Gm)|.$$
This implies that $\sum_{G \in \Gm} |S_{\eps}(G)| = o(|G|q^k)$ and $\E_\Gm[|S_\eps(\Gm)|] = o(q^k)$. Now fix $\Gm \in \F_q^{k \times n}$. We write 
\begin{align}
	\left(\sum_{\sv \in \F_q^n} n_{\sv,\C}\right)^2 & = \left(n_{\mathbf{0,\C}} + \sum_{\sv \in S_{\eps}(\Gm)} n_{\sv,\C} + \sum_{\sv \notin S_{\eps}(\Gm)} n_{\sv,\C}\right)^2 
	 \le 3n^2_{\mathbf{0},\C} + 3  \left(\sum_{\sv \in S_{\eps}(\Gm)} n_{\sv,\C}\right)^2 + 3 \left(\sum_{\sv \notin S_{\eps}(\Gm)} n_{\sv,\C}\right)^2 
\label{eq:inequality}\\
	 & \le 3n^2_{\mathbf{0},\C} + 3 |S_{\eps}(\Gm)| \sum_{\sv \neq \mathbf{0}} n_{\sv,\C}^2 + 3 \left(q^{k/2 - \eps n}\right)^2 \label{eq:Cauchy-Schwartz}                               \\
	 & \le 3|S_\eps(\Gm)| + o(q^k)  \label{eq:sumnssquare}
\end{align}
Here \eqref{eq:inequality} follows from the inequality $(x+y+z)^2 \leq 3x^2+3y^2+3z^2$ (which can be proved by noticing that 
$3x^2+3y^2+3z^2-(x+y+z)^2=(x-y)^2+(y-z)^2+(x-z)^2$). \eqref{eq:Cauchy-Schwartz} follows from the Cauchy-Schwartz inequality and 
\eqref{eq:sumnssquare} is a consequence of $ \sum_{\sv \in \F_q^k} n_{\sv,\C}^2 = 1$ which also gives $n^2_{\mathbf{0},\C} \le 1$.
In order to conclude, we write 
\begin{align*}
	P = \E_{\Gm}\left[\frac{1}{q^k}\left(\sum_{\sv \in F_q^k} n_{\sv,\C}\right)^2\right] \le \frac{1}{q^k}\left(3 \E_{\Gm}[|S_{\eps}(\Gm)|] + o(q^k)\right) = o(1).
\end{align*}
\end{proof}

\section{From the quantum decoding problem to the short codeword problem}
In this section, we show how to apply our algorithm for the quantum decoding problem into Regev's reduction in order to obtain quantum algorithms for the short codeword problem. 

We fix $n,k',q \in \mathbb{N}$ with $q \ge 2$ as well as $\omega' \in (0,1)$. We start from a random instance $\Gm' \in \F_q^{k' \times n}$ of $\SCP(q,n,k',\omega')$. Let $\C'$ be the code associated to $\Gm'$ and $\C = (\C')^\perp$ the dual code of $\C'$. The idea will be to solve a quantum decoding problem associated to $\C$, {\ie} from the state 
$ \ket{\psi_\cv} \eqdef \sum_{\ev \in \F_q^n} f(\ev) \ket{\cv + \ev}$ where $\cv$ belongs to $\C$, we want to recover $\cv$. Then we apply the quantum Fourier transform and measure in the computational basis to obtain a short codeword of $\C'$. We also define $k = n-k'$ and 

\begin{align*} 
	\omega  \eqdef (\omega')^\perp = \frac{\Rbra{\sqrt{(q-1)(1-{\omega'})} - \sqrt{{\omega'}}}^{2}}{q}\quad ; \quad 
	f(\ev)  \eqdef \left(\sqrt{\frac{\omega}{q-1}}\right)^{|\ev|}\left(\sqrt{1 - {\omega}}\right)^{n - |\ev|}
\end{align*}
Recall also, using $\omega^\perp = {\omega'}$ that  
$$\QFT{f}(\yv) = \left(\sqrt{\frac{{\omega'}}{q-1}}\right)^{|\yv|}\left(\sqrt{1 - {{\omega'}}}\right)^{n - |\yv|}.$$

\paragraph{Remark.} We use this notation $k',\omega'$ so that the problem we reduce to is a $\QDP(q,n,k,\omega)$ with a generating matrix $\Gm \in \F_q^{k \times n}$. This allows us to keep notation consistent with the previous section but be aware that the Short Codeword problem we are solving is on $\C' = \C^\perp$.

\subsection{Regev's reduction for codes}

We now describe Regev's reduction for codes. As we will see, this does not necessarily give a reduction from the short codeword problem to the quantum decoding problem because of the small error in the quantum decoding algorithm. We consider the formulation of this reduction from~\cite{SSTX09} and adapted in~\cite{DRT23} in the context of codes. \\

We first construct 
$$\ket{\Omega_0} = \frac{1}{\sqrt{|\C|}} \sum_{\cv \in \C} \sum_{\ev \in \F_q^n} f(\ev)  \ket{\cv} \ket{\ev}
$$
and add $\cv$ to the second register to obtain
$$\ket{\Omega_1} = \frac{1}{\sqrt{|\C|}} \sum_{\cv \in \C} \ket{\cv}\ket{\psi_\cv}, \qquad \qquad \textrm{ where } \ket{\psi_\cv} = \sum_{\ev \in \F_q^n} f(\ev) \ket{\cv + \ev} .$$
The idea is then to recover $\cv$ from $\ket{\psi_{\cv}}$ using an algorithm for the quantum decoding problem. If this can be done perfectly, we can actually use this algorithm to erase the first register and obtain the state 
$$ \ket{\Omega_2} = \frac{1}{\sqrt{|\C|}} \sum_{\cv \in \C} \ket{\psi_\cv}.$$
We then apply the Quantum Fourier Transform on this state to get 
$$ \ket{\QFT{\Omega_3}} = \frac{1}{\sqrt{|\C|}} \sum_{\cv \in \C} \ket{\QFT{\psi_\cv}} = \sqrt{|\C|} \sum_{\yv \in \C'} \QFT{f}(\yv) \ket{\yv}.$$
This follows from Proposition \ref{proposition:periodic}.
Finally, we measure this state in the computational basis and hope to find a small codeword. 

The algorithm can be summarized by

\noindent\fbox{\parbox{\textwidth}{
\begin{center}{\bf Algorithm of the quantum reduction.}
\end{center}
\begin{align*}
	&\text{Initial state preparation : } &  & \quad \ket{\Omega_0} = \frac{1}{\sqrt{|\C|}} \sum_{\cv \in \C} \sum_{\ev \in \F_q^n} f(\ev)  \ket{\cv} \ket{\ev} \nonumber \\
	&\text{adding $\cv$ to $\ev$:} & \mapsto &  \quad \ket{\Omega_1} =\quad \frac{1}{\sqrt{|\C|}} \sum_{\cv \in \C} \sum_{\ev \in \F_q^n}  f(\ev) \ket{\cv}\ket{\cv+\ev}= \quad \frac{1}{\sqrt{|\C|}} \sum_{\cv \in \C} \ket{\cv}\ket{\psi_\cv} \nonumber \\
	&\text{decoding and erasing 1st register} &\mapsto &  \quad \ket{\Omega_2} =\quad \frac{1}{\sqrt{|\C|}} \sum_{\cv \in \C} \ket{\zerom}\ket{\psi_\cv}  \\
	&\text{QFT on the $2$nd register:} & \mapsto & \quad \ket{\Omega_3} = \frac{1}{\sqrt{|\C|}} \sum_{\cv \in \C}  \ket{\zerom} \ket{\QFT{\psi_\cv}} = \sqrt{|\C|} \sum_{\yv \in \C'} \QFT{f}(\yv)  \ket{\zerom} \ket{\yv}\\
	&\text{measuring the whole state:} & \mapsto & \quad \ket{\zerom} \ket{\yv} \;\;\text{(where $\yv \in \C'=\C^\perp$)}
\end{align*}}}

\noindent There are a two issues that can make the above algorithm not work as we want:
\begin{itemize}
	\item The quantum decoding problem used in order to go from $\ket{\Omega_1}$ to $\ket{\Omega_2}$ does not work perfectly in many cases. Even if we have an algorithm which works w.p. $1 - o(1)$, this can greatly change the state $\ket{\Omega_3}$ that we have at the end\footnote{This seems counterintuitive at first as we would expect the final state to be $\eps$-close to the ideal state if the quantum decoding succeeds w.p. $1-\eps$. However, we are in regimes where an ideal quantum decoder does not exist so such continuity arguments will not hold. As it will appear in our analysis, it is possible to slightly tweak the measurements used in the Quantum Decoding Problem and greatly change the outcome state.}. This also means we have to explicit each time our quantum decoding procedure and analyze thoroughly the resulting state. 
	\item Even if we obtain exactly the state $\sqrt{|\C|} \sum_{\yv \in \C'} \QFT{f}(\yv) \ket{\yv}$ we want, for values of ${\omega'}$ which are too large, this algorithm will actually always output $\yv = \mathbf{0}$ but we want a small non-zero codeword so our algorithm will not work. 
\end{itemize}

In this section, we show that our algorithms (or slight variants of our algorithms) can be successfully used in Regev's reduction in order to solve the Short Codeword Problem despite the above shortcomings. We show the following:
\begin{enumerate}
	\item If we take our polynomial time algorithms for the quantum decoding problem (Section~\ref{Section:Polynomial}) we can find in quantum polynomial time small codewords down to Prange's bound, {\ie} down to $\frac{(n-k')(q-1)}{q} = \frac{k(q-1)}{q}$. Notice however, that our algorithm obtains a state $\ket{\Omega_2}$ very far from the theoretical state $ \ket{\Omega_2} = \frac{1}{\sqrt{|\C|}} \sum_{\cv \in \C} \ket{\psi_\cv}$, but we still show how to obtain a small codeword in $\C$ after performing $\QFTt$ and measuring. 
	\item If we consider the tractability regime and if we take the Pretty Good Measurement associated to the states $\ket{\QFT{\psi_\cv}}$, we show that we actually exactly get the state $\ket{\Omega_2}$ we are looking for (up to a normalization factor). We then look at this PGM and $2$ variants:
	 \begin{enumerate}
	 	\item If we finish the analysis with the PGM, we will most often be in regimes where we measure $\mathbf{0}$ in the final step so we will not be able to solve the Short Codeword problem. 
	 	\item We can slightly tweak the PGM so that it will give us a short codeword down to the tractability bound. 
	 	\item We also show another example where we can slightly tweak the PGM but where the reduction utterly fails, meaning that the state we obtain before measuring is $\ket{\perp}$. This shows that there is no hope to perform generic reduction between the quantum decoding problem and the short codeword problem with this method. 
\end{enumerate}
\end{enumerate}

\subsection{The quantum reduction with unambiguous state discrimination}\label{Section:ReductionUSD}
We first show how to use our quantum polynomial time algorithms for the quantum decoding in Regev's reduction. We first construct 
$$ \ket{\Omega_1} = \frac{1}{\sqrt{|\C|}} \sum_{\cv \in \C} \ket{\cv}\ket{\psi_\cv}, \qquad \qquad \textrm{ where } \ket{\psi_\cv} = \sum_{\ev \in \F_q^n} f(\ev) \ket{\cv + \ev} .$$
We then apply unambiguous state discrimination measurement on $\ket{\psi_\cv} = \bigotimes_{i = 1}^n \ket{\psi_{c_i}}$. Recall that by using the version of unambiguous state discrimination presented in Proposition \ref{proposition:efficientUSD} for each $i$, we perform a unitary $U$ on the $i^{th}$ register of $\ket{\Omega_1}$ that does the following for each $c_i \in \F_q$:
$$ U \ket{\psi_{c_i}}\ket{0} = \sqrt{\pusd} \ket{c_i}\ket{0} + \sqrt{1 - \pusd} \QFT{\ket{0}}\ket{1}.$$
Here
\begin{equation}
\label{eq:pusd}
\pusd = q \cdot \frac{\omega^\perp}{q-1}=\frac{q}{q-1} {\omega'}.
\end{equation} 
After applying this (coherent) USD, we obtain the state 
\begin{align*}
	\ket{\Omega_2} & = \frac{1}{\sqrt{|\C|}} \sum_{\cv \in \C} \left(\ket{\cv} \otimes \left( \bigotimes_{i = 1}^n  \sqrt{\pusd} \ket{c_i}\ket{0} + \sqrt{1 - \pusd} \QFT{\ket{0}}\ket{1} \right)\right) \\
	& =  \frac{1}{\sqrt{|\C|}} \sum_{\cv \in \C} \ket{\cv} \sum_{J \subseteq [n]} \beta_J \ket{\wcv_J}.
\end{align*}
where 
$$
\ket{\wcv_J} = \bigotimes_{i = 1}^n \ket{\gamma_i} \quad \textrm{ with } \left\{\begin{tabular}{rl} $\ket{\gamma_i} =\ket{c_i} \ket{0}  $ & if $i \in J$ \\ $\ket{\gamma_i}  = \QFT{\ket{0}}\ket{1}$ & \textrm{otherwise}\end{tabular}\right.$$
and $\beta_J = \sqrt{(1-\pusd)^{n-|J|}(\pusd)^{|J|}}$. $J$ here corresponds to the set of indices where the USD succeeded.  Notice that one can efficiently recover $J$ from $\ket{\wcv_J}$ by looking at the outcome registers, so we can add it to obtain the state 
$$ \ket{\Omega_3} = \frac{1}{\sqrt{|\C|}} \sum_{\cv \in \C} \ket{\cv} \sum_{J \subseteq [n]} \beta_J \ket{\wcv_J}\ket{J}.$$
We now measure $J$ to obtain the state 
$$ \ket{\Omega_4(J)} = \frac{1}{\sqrt{|\C|}} \sum_{\cv \in \C} \ket{\cv}\ket{\wcv_J}.$$

Notice that $|J|$ follows the distribution $D(\pusd)$ with $\pusd > R$ so, there exists an absolute constant $\eps > 0$, s.t.
$ (R + \eps)n \le |J| \le (\pusd)n$ w.p. at least $\frac{1}{2} - o(1)$ (the probability that $|J| \le \pusd n$ is at least $\frac{1}{2}$). Moreover, using the same argument as in Section~\ref{Section:BinaryPolynomial}, we can recover $\cv$ from $\cv_J$ w.p. $1 - o(1)$. This means we can erase the register $\cv$ in $\ket{\Omega_4(J)}$ to get 
$$ \ket{\Omega_5(J)} = \frac{1}{\sqrt{|\C|}} \sum_{\cv \in \C} \ket{\cv_J} =
\frac{1}{\sqrt{|\C|}} \sum_{\cv \in \C_J} \ket{\cv} .$$ 

We apply the Fourier transform on this state to get 
$$ \ket{\QFT{\Omega_5(J)}} = \frac{1}{\sqrt{|(\C_J)^\perp|}}  \sum_{\yv \in (\C_J)^\perp} \ket{\yv} .$$ 
By measuring this state, we get a vector $\yv \in (\C_J)^\perp$ of weight at most $\frac{(q-1)|J|}{q} \le \frac{(q-1)\pusd n}{q} = {\omega'} n$ w.p. $\Theta(1)$. Here we used \eqref{eq:pusd} for the last equality. The crux is that by Lemma \ref{lem:puncture_shorten} we have
$$
(\C_J)^\perp = (\C^\perp)^J.
$$
In other words, we get words in $\C^\perp$ shortened at $J$, meaning dual codewords that are $0$ outside $J$.
We have therefore constructed a word $\zv \in \C^\perp$ s.t. $\zv_j = \yv_j$ if $j \in J$ and $\zv_j = 0$ otherwise. 

In conclusion, we just proved the following
\begin{theorem}
	The above algorithm, that performs Regev's reduction and uses unambiguous state discrimination for the quantum decoding problem, can solve in polynomial time $\SCP(q,n,k',{\omega'} )$ for ${\omega'} > \frac{(q-1)k}{q}=\frac{(q-1)(n-k')}{q}$ w.p. $\Theta(1)$. 
\end{theorem}

Notice that we can repeat this algorithm to amplify the success probability. Our algorithm can go down to Prange's bound $\frac{(q-1)(n-k')}{q}$, which is the best known bound for polynomial time algorithms for the short codeword problem. 
 The whole algorithm is summarized by:
 
\noindent\fbox{\parbox{\textwidth}{
\begin{center}
{\bf Algorithm of the quantum reduction in the case of USD.}
\end{center}
\begin{align*}
	&\text{Initial state preparation:} &  & \quad \ket{\Omega_0} &=& \quad\frac{1}{\sqrt{|\C|}} \sum_{\cv \in \C} \sum_{\ev \in \F_q^n} f(\ev)  \ket{\cv} \ket{\ev} \nonumber \\
	&\text{adding $\cv$ to $\ev$:} & \mapsto &  \quad \ket{\Omega_1} &=& \quad \frac{1}{\sqrt{|\C|}} \sum_{\cv \in \C} \sum_{\ev \in \F_q^n}  f(\ev) \ket{\cv}\ket{\cv+\ev}= \quad \frac{1}{\sqrt{|\C|}} \sum_{\cv \in \C} \ket{\cv}\ket{\psi_\cv} \nonumber \\
	&\text{applying coherent USD:} & \mapsto  & \quad \ket{\Omega_2}  &=& \quad \frac{1}{\sqrt{|\C|}} \sum_{\cv \in \C} \left(\ket{\cv} \otimes \left( \bigotimes_{i = 1}^n  \sqrt{\pusd} \ket{c_i}\ket{0} + \sqrt{1 - \pusd} \QFT{\ket{0}}\ket{1} \right)\right) \\
	& & & \quad &=& \quad \frac{1}{\sqrt{|\C|}} \sum_{\cv \in \C} \ket{\cv} \sum_{J \subseteq [n]} \beta_J \ket{\wcv_J} \\
	&\text{put $J$ in the last register using $\ket{\wcv_J}$} &\mapsto & \quad  \ket{\Omega_3} &=& \frac{1}{\sqrt{|\C|}} \sum_{\cv \in \C} \ket{\cv} \sum_{J \subseteq [n]} \beta_J \ket{\wcv_J}\ket{J} \\
	&\text{measure $J$} &\mapsto &  \quad \ket{\Omega_4} &=& \quad \frac{1}{\sqrt{|\C|}} \sum_{\cv \in \C} \ket{\cv}\ket{\wcv_J}  \\
	&\text{erase $\cv$} & \mapsto & \quad \ket{\Omega_5} &=& \quad \frac{1}{\sqrt{|\C|}} \sum_{\cv \in \C} \ket{\cv_J} =
\frac{1}{\sqrt{|\C|}} \sum_{\cv \in \C_J} \ket{\cv}\\
	&\text{QFT:} & \mapsto & \quad \ket{\Omega_6} &=&\quad  \frac{1}{\sqrt{|(\C_J)^\perp|}}  \sum_{\yv \in (\C_J)^\perp} \ket{\yv}\\
	&\text{measuring the whole state:} & \mapsto & \quad  & & \quad \ket{\yv} \;\;\text{(where $\yv \in (\C_J)^\perp = (\C^\perp)^J \subset \C^\perp$)}
\end{align*}}}

\subsection{The Quantum reduction with the Pretty Good Measurement}~\label{Section:PGMReduction}
We now study Regev's reduction when we use the PGM for the quantum decoding problem. We consider the basis $\{\ket{Y_\cv}\}_{\cv \in \C}$ described in the previous section associated to the states $\{\QFT{\ket{\psi_\cv}}\}_{\cv \in \C}$. We showed that 
\begin{align*}
\forall \cv \in \C, \ \braket{\QFT{\psi_\cv}}{Y_\cv} = \sqrt{\PPGM}
\end{align*}
where $\PPGM$ is the probability that the Pretty Good Measurement succeeds. We now unfold Regev's reduction. We start from the state $\ket{\Omega_1}$ with a slight change, we apply namely immediately the $\QFTt$ on the second register to get
\begin{align*}
\ket{\Omega_1} = \frac{1}{\sqrt{|\C|}} \sum_{\cv \in \C} \ket{\cv}\ket{\QFT{\psi_\cv}}
\end{align*}
with $\ket{\psi_{\cv}} = \sum_{\ev} f(\ev)\ket{\cv + \ev}$. We then perform coherently the PGM on the second register and write the output on the third register. This means that if we write each $\ket{\QFT{\psi_\cv}} = \sum_{\cv' \in \C} \alpha_{\cv,\cv'} \ket{Y_{\cv'}}$, we obtain
\begin{align*}
\ket{\Omega_2} = \frac{1}{\sqrt{|\C|}} \sum_{\cv \in \C} \ket{\cv} \sum_{\cv' \in \C} \alpha_{\cv,\cv'} \ket{Y_{\cv'}}\ket{\cv'}
\end{align*}
We then subtract the value of the third register in the first register to get
\begin{align*}
\ket{\Omega_3} = \frac{1}{\sqrt{|\C|}} \sum_{\cv,\cv' \in \C} \alpha_{\cv,\cv'} \ket{\cv - \cv'} \ket{Y_{\cv'}}\ket{\cv'}
\end{align*}
Finally, we reverse the PGM between registers $2$ and $3$ to obtain the state 
\begin{align*}
\ket{\Omega_4} = \frac{1}{\sqrt{|\C|}} \sum_{\cv,\cv' \in \C} \alpha_{\cv,\cv'} \ket{\cv - \cv'} \ket{Y_{\cv'}}\ket{\zerom}
\end{align*}

From the discussion at the beginning of this section, we have that for any $\cv \in \C, \alpha_{\cv,\cv} = \sqrt{\PPGM}$. This means we can rewrite the above state as
\begin{align*}
\ket{\Omega_4} & = \frac{1}{\sqrt{|\C|}} \left(\sum_{\cv' \in \C} \sqrt{\PPGM} \ket{\zerom}\ket{Y_{\cv'}} + \sum_{\cv, \cv' \neq \cv} \alpha_{\cv,\cv'}\ket{\cv - \cv'}\ket{Y_{\cv'}}\right) \\ 
& = \sqrt{\PPGM}\ket{\zerom}\left(\frac{1}{\sqrt{|\C|}} \sum_{\cv \in \C} \ket{Y_{\cv}}\right) + \sum_{\cv, \cv' \neq \cv} \alpha_{\cv,\cv'}\ket{\cv - \cv'}\ket{Y_{\cv'}}.
\end{align*}

The next step of the reduction is to measure the first register of $\ket{\Omega_4}$. Since the states $\ket{Y_{\cv'}}$ are orthogonal and of norm $1$, we measure $0$ w.p. $\PPGM$ in the first register and the second register becomes 
\begin{align*}
\ket{\Omega_5} = \frac{1}{\sqrt{|\C|}} \sum_{\cv \in \C} \ket{Y_\cv} = \ket{\tW_0} = \frac{1}{n_0} \sum_{\yv \in \C'} \QFT{f}(\yv)\ket{\yv}
\end{align*}
where $n_0 \eqdef  \ \norm{\sum_{\yv \in \C'} \QFT{f}(\yv)\ket{\yv}}$.
We measure this final state to potentially measure a small codeword. Let $a(t) = |\{\yv \in \C' : |\yv| = t\}$. Note that this quantity corresponds to
$a_{\zerom}(t,\C')$ as defined in Subsection \ref{Section:FirstComputations}. The probability $p(t)$ that the above algorithms finds a word of weight $t$ in $\C'$ is 
\begin{align}
	p(t) = \frac{1}{n_0^2} a(t) |\QFT{f}(t)|^2,
\end{align}
where we overload the notation $\QFT{f}$ to mean that $\QFT{f}(t) = \QFT{f}(\yv)$ for any $\yv$ s.t. $|\yv| = t$ (recall that $\QFT{f}(\yv)$ is constant for any of these $\yv$). The issue here is that for ${\omega'}$ small enough, we will almost always measure $0$. Indeed,
$$p_0 = \frac{|\QFT{f}(0)|^2}{n_0^2} = \frac{|\QFT{f}(0)|^2}{\sum_{t} a(t) |\QFT{f}(t)|^2} = \frac{|\QFT{f}(0)|^2}{|\QFT{f}(0)|^2 + \sum_{t \neq 0} a(t) |\QFT{f}(t)|^2}.$$
Here, notice that  
\begin{align*} 
|\QFT{f}(0)|^2 = (1- {\omega'})^n \\
\Pr_G\left[\sum_{t \neq 0} a(t)|\QFT{f}(t)|^2 \le \frac{2}{q^k}\right] \ge 1 -o(1) 
\end{align*}
where for the last inequality, we use the concentration bounds for $a(t)=a_{\zerom}(t,\C')$ of Section~\ref{Section:FirstComputations}. So when ${\omega'} < 1 - q^{-\frac{k}{n}}$, we measure $0$ with high probability. This unfortunately happens quite often and it is a problem because in our short codeword problem, we want to find a small non-zero vector. 

\subsubsection{A counterexample that shows complete failure}~\label{Section:CompleteFailure}
We show that things can go even worse when slightly changing the measurement used. We show that instead of measuring $\ket{\mathbf{0}}$, we can measure some given state $\ket{\bot}$ orthogonal to all the $\ket{\tW_\sv}$. 
Recall from Proposition~\ref{Proposition:PGMDescription} that 
\begin{align*}
\ket{{Y}_\cv} = \frac{1}{\sqrt{q^k}} \sum_{\sv \in \F_q^k} \chi_{\cv}(u_\sv)\ket{\tW_\sv}
\end{align*} 
Also,  the state resulting from Regev's reduction is the state $\sum_{\cv \in \C} \ket{Y_\cv} = \ket{\tW_0} $. Our modified measurement can give an extra outcome which will be an $\uv \in \F_q^n \backslash \C$ and we define $\S = \C \cup \{\uv\}$. Let
\begin{align*}
\ket{Z_\cv} & \eqdef \frac{1}{\sqrt{q^k}} \left(\ket{\perp} + \sum_{\sv \neq \mathbf{0}} \chi_\cv(u_\sv) \ket{\tW_\sv} \right) \quad \forall \cv \in \C \\
\ket{Z_{\uv}} & \eqdef \ket{{\tW}_0}
\end{align*}

Notice that the $\ket{Z_\yv}$ are pairwise orthogonal and $\myspan(\{\ket{Z_{\yv}}\}_{\yv \in \S}) = \myspan(\{\ket{\tW_\sv}\}_{\sv \in \F_q^k}, \ket{\bot}) = \myspan(\{\QFT{\ket{\psi_{\cv}}}\}_{\cv \in \C},\ket{\bot})$. This means the measurement $\{\ket{Z_{\yv}}\}_{\yv \in \S}$ will be complete when measuring any $\ket{\QFT{\psi_{\cv}}}$.

Recall that $\ket{\QFT{\psi}_\cv} = \frac{1}{\sqrt{q^k}} \sum_{\sv \in \F_q^k} \chi_{\sv}(\uv_\sv)\ket{W_\sv}$, so we have 
\begin{align*}
\forall \cv \in \C, \ \braket{\QFT{\psi_\cv}}{{Z}_\cv} = \frac{1}{\sqrt{q^k}} \sum_{\sv \neq 0} n_s = \sqrt{\PPGM} - \frac{n_0}{\sqrt{q^k}} \ge \sqrt{\PPGM} - \frac{1}{\sqrt{q^k}}
\end{align*}
where we use in the second equality that $\sqrt{\PPGM}=\frac{1}{\sqrt{q^k}}\sum_{\sv \in \F_q^k}n_{\sv}$ and in the last inequality that
$n_0 \leq 1$.
This means the above measurement solves the quantum decoding problem wp. $\left(\sqrt{\PPGM} - \frac{n_0}{\sqrt{q^k}}\right)^2 \ge \left(\sqrt{\PPGM} - \frac{1}{\sqrt{q^k}}\right)^2$ which is $1 - o(1)$ as long as $P_{PGM} = 1 - o(1)$. 

Now, we perform the reduction presented in Section~\ref{Section:PGMReduction}. We just rewrite the states of the reduction
\begin{align*}
\ket{\Omega_1} & = \frac{1}{\sqrt{|\C|}} \sum_{\cv \in \C} \ket{\cv}\ket{\QFT{\psi_\cv}} \\
\ket{\Omega_2} & = \frac{1}{\sqrt{|\C|}} \sum_{\cv \in \C} \ket{\cv} \sum_{\yv \in \S} \beta_{\cv,\yv} \ket{Z_{\yv}}\ket{\yv} \qquad \qquad \textrm{ where } \beta_{\cv,\yv} = \braket{\QFT{\psi_\cv}}{Z_\yv} \\
\ket{\Omega_3} & = \frac{1}{\sqrt{|\C|}} \sum_{\cv \in \C, \yv \in \S} \beta_{\cv,\yv} \ket{\cv - \yv} \ket{Z_{\yv}}\ket{\yv} \\
\ket{\Omega_4} & = \frac{1}{\sqrt{|\C|}} \sum_{\cv \in \C, \yv \in \S} \beta_{\cv,\yv} \ket{\cv - \yv} \ket{Z_{\yv}}\ket{\zerom} = \left(\sqrt{P_{PGM}} - \frac{n_0}{\sqrt{q^k}}\right) \ket{\zerom}\frac{1}{\sqrt{|\C|}}\sum_{\cv \in \C} \ket{Z_\cv} + \frac{1}{\sqrt{|\C|}} \sum_{\cv \in \C}\sum_{\yv \in \S, \yv \neq \cv} \beta_{\cv,\yv}\ket{\cv - \yv}\ket{Z_\yv}
\end{align*}
where in the last line, we dropped in the last equality the third register and we used that $\beta_{\cv,\cv} = \left(\sqrt{P_{PGM}} - \frac{n_0}{\sqrt{q^k}}\right)$ for each $\cv \in \C$. This means that when we measure the first register, we obtain $\mathbf{0}$ w.p. $\left(\sqrt{\PPGM} - \frac{n_0}{\sqrt{q^k}}\right)^2$, and the resulting state  is $\ket{\Omega_{5}} = \sum_{\cv \in \C} \ket{Z_\cv} = \ket{\bot}$ which shows that the reduction entirely fails in this case. 
\subsubsection{A measurement that works}\label{Section:MeasurementThatWorks}
Finally, we show a measurement that will make the reduction work when $\PPGM = 1 - o(1)$. The idea is similar to the one of Section~\ref{Section:CompleteFailure}. We add an extra outcome $\uv \in \F_q^n\backslash \C$ and define $\S = \C \cup \{\uv\}$. We now define
\begin{align*}
	\ket{U_{0}} & = \frac{\sum_{\yv \in \C' : \yv \neq \mathbf{0}} \QFT{f}(\yv)\ket{\yv}}{\norm{\sum_{\yv \in \C' : \yv \neq \mathbf{0}} \QFT{f}(\yv)\ket{\yv}}} \\
	\forall \cv \in \C, \ \ket{{Z}_\cv} & = \frac{1}{\sqrt{q^k}} \left(\ket{U_0} + \sum_{\sv \neq \mathbf{0}} \chi_\cv(u_\sv) \ket{\tW_\sv} \right) \\
	\ket{Z_{\uv}} & = \ket{\mathbf{0}}
\end{align*}

The $\ket{Z_\cv}$ are exactly the states $\ket{Y_{\cv}}$ of the pretty good measurement but we removed the $\ket{\mathbf{0}}$ component of $\ket{\tW_{0}}$. As in the previous subsection, the $\ket{Z_\yv}$ are orthogonal. In order to make the measurement complete, we added the extra basis element $\ket{Z_{\uv}}=\ket{\zerom}$. We therefore have a projective measurement $\{\ket{Z_{\yv}}\}_{\yv \in \S}$.
Again, we have 
$
	\braket{\QFT{\psi_\cv}}{{Z}_\cv} \ge \sqrt{\PPGM} - \frac{1}{\sqrt{q^k}}
$
and independent of $\cv$ so w.p. at least $\left(\sqrt{\PPGM} - \frac{1}{\sqrt{q^k}}\right)^2$, we get the state $\ket{U_0}$. Then, if we measure this state $\ket{U_0}$,  we will get a codeword of weight $t$ w.p. 
$$ p(t) = \frac{a(t)|\QFT{f}(t)|^2}{\sum_{t \neq 0} a(t)|\QFT{f}(t)|^2}, \ \forall t \neq 0 $$
and $p(0) = 0$, where $a(t)$ is the number of codewords of weight $t$ in $\C'$. Recall that 
$$\QFT{f}(t) = \left(\sqrt{\frac{{\omega'}}{q-1}}\right)^{t}\left(\sqrt{1 - {{\omega'}}}\right)^{n - t}.$$
and we are in the regime where $P_{PGM} = 1 -o(1)$ which means that $\omega < (\drgv(1-\frac{k}{n}))^\bot$  and hence $\omega' > \drgv(\frac{n-k}{n}) = \drgv(\frac{k'}{n}) $.

 Using Proposition~\ref{Proposition:18} and the expression of $\QFT{f}(t)$, as well as concentration bounds for $a(t)$, we have that for any absolute constant $\eps > 0$,  $\sum_{t = \lfloor({\omega'} - \eps)n\rfloor}^{\lfloor({\omega'} + \eps)n\rfloor} p(t) = 1 - o(1)$. This means we will measure a word of weight approximately $\lfloor {\omega'} n\rfloor$ in $\C'$.
\paragraph{Discussion.}
These $3$ examples above show that it is very easy to slightly modify the algorithm for solving the quantum decoding problem and drastically change the result after Regev's reduction. We therefore cannot have proper reduction theorems between the quantum decoding problem and the short codeword problem but we have to analyze on a case by case basis whether an algorithm for the quantum decoding problem can be used for finding a short codeword.  On the positive side of the reduction, we can summarize our results as follows:
\begin{proposition}
	Let $q,n,k \in \mathbb{N}$ with $q \ge 2$ and $\omega \in (0,1)$. Let also  $R = \lfloor \frac{k}{n}\rfloor$, $\omega' = \omega^\perp$ and $k' = n-k$.
	\begin{itemize}
		\item For $\omega < (\frac{q-1 R}{q})^\perp$, there exists a quantum algorithm running in time $\poly(n,\log(q))$ that solves $\QDP(q,n,k,\omega)$ w.p. $1 - o(1)$ (Theorem~\ref{Theorem:polyFull}) . Moreover, this algorithm can be used using Regev's reduction to solve $\SCP(q,n,k',\omega')$ in time $poly(n,\log(q))$ (Section~\ref{Section:ReductionUSD}).
		\item For $\omega < (\delta_{\min}(1-R))^\perp$, there exists a quantum algorithm (for which we don't specify the running time but which could be exponential in $n$) that solves $\QDP(q,n,k,\omega)$ w.p. $1 - o(1)$ (Proposition~\ref{Proposition:23}). This algorithm can be (slightly tweaked but with success probability still $1 - o(1)$) and used in Regev's reduction to solve $\SCP(q,n,k',\omega')$ w.p. $\Theta(1)$ (Section~\ref{Section:MeasurementThatWorks}).
	\end{itemize}
\end{proposition}

\newcommand{\etalchar}[1]{$^{#1}$}

\newpage
\begin{appendix}
	\section{General phases} \label{Appendix:A} 
	We consider more general error functions $f(\ev) = \sqrt{1-t}^{(n - |\ev|)}(e^{i\theta} \sqrt{t})^{|\ev|}$ with $t \in [0,\frac{1}{2}]$ and $\theta \in [0,2\pi)$, which means we consider the states 
	$$ \ket{\psi_f(\cv)} = \sum_{\ev \in \zo^n} f(\ev)\ket{\cv + \ev} = \bigotimes_{i = 1}^n \sqrt{1-t} \ket{c_i} + e^{i \theta}\sqrt{t} \ket{1 - c_i}.$$
	Again, we consider unambiguous states discrimination between the following two states 
	$$ \ket{\zeta_0} = \sqrt{1-t} \ket{0} + e^{i\theta}\sqrt{t} \ket{1} \quad ; \quad \ket{\zeta_1} = e^{i\theta} \sqrt{t} \ket{0} + \sqrt{1-t} \ket{1}.$$
	We have $|\braket{\zeta_0}{\zeta_1}| = |\sqrt{t(1-t)}(e^{i\theta} + e^{-i\theta})| = 2\sqrt{t(1-t)}|\cos(\theta)|$. From there, we have that decoding with this unambiguous measurement is possible w.h.p as long as 
	\begin{align}\label{Eq:4}
	1 - 2\sqrt{t(1-t)}|\cos(\theta)| > \frac{k}{n}
	\end{align}
	Now, what do we get in the dual? Let 
	$$\ket{\psi_f} = \frac{1}{\sqrt{Z}} \sum_{\cv \in \C, \ev \in \zo^n} f(\ev) \ket{\cv + \ev}$$
	where $Z$ is a normalizing constant.
	We write 
	$$ \QFT{\ket{\psi_f}}  = \frac{2^k}{\sqrt{2^n \cdot Z}} \sum_{\yv \in \C^\bot} \hat{f}(\yv) \ket{\yv}.$$
	Moreover, we write 
	\begin{align*}
	\QFT{ \ket{\zeta_0}} = \frac{1}{\sqrt{2}}\left(\sqrt{1-t} + e^{i\theta} \sqrt{t}\right) \ket{0} + \frac{1}{\sqrt{2}}\left(\sqrt{1-t} - e^{i\theta} \sqrt{t}\right) \ket{1} \\
	\QFT{\ket{\zeta_1}} = \frac{1}{\sqrt{2}}\left(\sqrt{1-t} + e^{i\theta} \sqrt{t}\right) \ket{0} - \frac{1}{\sqrt{2}}\left(\sqrt{1-t} - e^{i\theta} \sqrt{t}\right) \ket{1} 
	\end{align*}
	This means that $\hat{f}(\yv) =  \frac{1}{\sqrt{2}}\left(\sqrt{1-t} + e^{i\theta} \sqrt{t}\right)^{(n - |\yv|)}\frac{1}{\sqrt{2}}\left(\sqrt{1-t} - e^{i\theta} \sqrt{t}\right)^{|\yv|}$. 
	The probability $p(t,\theta)$ to measure $1$ on each coordinate in the above is given by
	\begin{align}
	p(t,\theta) & = \frac{1}{2}\left|\sqrt{1-t} - e^{i\theta}\sqrt{t}\right|^2 = \frac{1}{2}\left((\sqrt{1-t} - \sqrt{t}\cos(\theta))^2 + t\sin^2(\theta)\right) \\ & = \frac{1}{2}\left(1 - 2\sqrt{t(1-t)}\cos(\theta)\right)
	\end{align}
	In the case Equation~\ref{Eq:4} is saturated, meaning $1 - 2\sqrt{t(1-t)}|\cos(\theta)| \approx \frac{k}{n}$, we have $p(t,\theta) = \frac{k}{2n}$ as long as $\cos(\theta) \ge 0$ (otherwise, we have the symmetric for large weights) which is Prange's bound. Notice that the above only works when Equation~\ref{Eq:4} can be saturated, so we can not take $\theta = \pi/2$ for example.

\end{appendix} \end{document}